\documentclass[12pt]{article}
\usepackage{amssymb,latexsym}
\usepackage{amsmath}
\usepackage{hyperref}
\usepackage{amsthm}
\usepackage{cite}
\usepackage{graphicx}
\usepackage[dvipsnames]{xcolor}
\usepackage{mathtools}
\usepackage{cleveref}

\usepackage{tikz} 
\usepackage{pgfplots}
\usepackage{esvect}

\usetikzlibrary{intersections}
\usetikzlibrary{shapes,arrows}
\usetikzlibrary{calc}
\usetikzlibrary{shapes.geometric}
\usetikzlibrary{shapes.arrows}
\usetikzlibrary{arrows.meta}
\usetikzlibrary{decorations.markings}
\usetikzlibrary{patterns}
\usetikzlibrary{fit}
\usetikzlibrary{hobby}
\usepgfplotslibrary{patchplots}
\pgfplotsset{compat=1.12}

\usepackage[top=2cm,bottom=2cm,left=2cm,right=2cm]{geometry}
\usepackage{colonequals}
\usepackage[normalem]{ulem}
\numberwithin{equation}{section}

\newtheorem{Thm}{Theorem}[section]
\newtheorem{thm}[Thm]{Theorem}
\newtheorem{Lem}[Thm]{Lemma}
\newtheorem{Def}[Thm]{Definition}
\newtheorem{defi}[Thm]{Definition}
\newtheorem{defprop}[Thm]{Proposition and Definition}
\newtheorem{deflem}[Thm]{Lemma and Definition}
\newtheorem{Thmdef}[Thm]{Theorem and Definition}
\newtheorem{Prop}[Thm]{Proposition}
\newtheorem{prop}[Thm]{Proposition}
\newtheorem{Rem}[Thm]{Remark}
\newtheorem{Cor}[Thm]{Corollary}
\newtheorem{Ex}[Thm]{Example}

\newcommand\R{{\mathbb R}}
\newcommand\Sphere{{\mathbb S}}

\newcommand\hull{\operatorname{span}}
\newcommand{\photo}{P^{\:\!n}}
\renewcommand\d{\partial}

\newcommand\grad{\operatorname{grad}}

\renewcommand\a{\lambda}

\newcommand{\definedas}{\mathrel{\raise.095ex\hbox{\rm :}\mkern-5.2mu=}}
\newcommand{\asdefined}{\mathrel{=\mkern-5.2mu}\raise.095ex\hbox{\rm :}\;}
\newcommand{\surf}{\Sigma^{n-1}}

\newcommand{\slice}{M^{n}}

\newcommand\beq{\begin{equation}}
\newcommand\eeq{\end{equation}}
\newcommand\ben{\begin{enumerate}}
\newcommand\een{\end{enumerate}}
\newcommand\bit{\begin{itemize}}
\newcommand\eit{\end{itemize}}

\newcommand{\Ric}{\operatorname{Ric}}

\newcommand{\defeq}{\vcentcolon=}
\newcommand{\eqdef}{=\vcentcolon}
\newcommand{\sgn}{\operatorname{sgn}}

\definecolor{greyblue}{RGB}{0, 104, 139}
\definecolor{gruen}{RGB}{13,146,8}

 \makeatletter
\newenvironment{proofof}[1]{\par
  \pushQED{\qed}%
  \normalfont \topsep6\p@\@plus6\p@\relax
  \trivlist
  \item[\hskip\labelsep
        \itshape
    Proof of #1\@addpunct{.}]\ignorespaces
}{%
  \popQED\endtrivlist\@endpefalse
}
\makeatother

\title{On equipotential photon surfaces in (electro-)static spacetimes of arbitrary dimension}

\author{Carla Cederbaum\thanks{cederbaum@math.uni-tuebingen.de}, Sophia Jahns\thanks{jahns@math.uni-tuebingen.de}, Olivia Vi\v{c}\'{a}nek Mart\'{i}nez\thanks{olivia.vicanek-martinez@math.uni-tuebingen.de} 
 \\[1ex]Department of Mathematics, University of T\"ubingen} 
 
\begin{document}

\date{}
\maketitle

\begin{abstract}
We study timelike, totally umbilic hypersurfaces -- called photon surfaces -- in $n+1$-dimensional static, asymptotically flat spacetimes, for $n+1\geq4$. First, we give a complete characterization of photon surfaces in a class of spherically symmetric spacetimes containing the (exterior) subextremal Reissner--Nordstr\"om spacetimes, and hence in particular the (exterior) positive mass Schwarzschild spacetimes. Next, we give new insights into the spacetime geometry near equipotential photon surfaces and provide a new characterization of photon spheres (not appealing to any field equations).

We furthermore show that any asymptotically flat electrostatic electro-vacuum spacetime with inner boundary consisting of equipotential, (quasi-locally) subextremal photon surfaces and/or non-degenerate black hole horizons must be isometric to a suitable piece of the necessarily subextremal Reissner--Norstr\"om spacetime of the same mass and charge. Our uniqueness result applies work by Jahns and extends and complements several existing uniqueness theorems. Its proof fundamentally relies on the lower regularity rigidity case of the Riemannian Positive Mass Theorem.
\end{abstract}

\section{Introduction}\label{sec:intro}
\emph{Photon surfaces} are timelike, totally umbilic hypersurfaces in Lorentzian manifolds of dimension $n+1$ (see \cite{CVE,Perlick} for more information). They naturally occur in spacetimes with symmetries. In \cite{cederbaum2019photon}, the first named author and Galloway discuss photon surfaces in static, spherically symmetric spacetimes of dimension $n+1$. The first goal of our  work is to provide an extensive existence and uniqueness analysis of photon surfaces in this setting. More specifically,  \cite[Theorem 3.5]{cederbaum2019photon} characterizes all spherically symmetric photon surfaces arising in manifolds in the \emph{class $\mathcal{S}$} consisting of smooth Lorentzian manifolds $(\R\times \mathcal{I}\times\Sphere^{n-1},\mathfrak{g})$  for an open interval $\mathcal{I}\subseteq(0,\infty)$, and so that there exists a smooth, positive function $f\colon\mathcal{I}\to\R$ for which one can express the metric $\mathfrak{g}$ as
\begin{align*}
\mathfrak{g}&=-f(r)dt^{2}+\frac{1}{f(r)}dr^{2}+r^{2}\,\Omega
\end{align*}
in the global coordinates $t\in\R$, $r\in \mathcal{I}$, where $\Omega$ denotes the canonical metric on~$\Sphere^{n-1}$. 

\textbf{\emph{Analysis of photon surface ODE:}} It is asserted in~\cite[Theorem 3.5]{cederbaum2019photon} that the radial profile curve $r=r(t)$ of any spherically symmetric photon surface in a Lorentzian manifold of class $\mathcal{S}$ must obey the \emph{photon surface ODE}
\begin{align}\label{eq:photonsurfintro}
\dot{r}(t)^{2}&=\frac{f(r(t))^{2}(\lambda^{2}r(t)^{2}-f(r(t)))}{\lambda^{2}r(t)^{2}},
\end{align}
where $\lambda>0$ is a parameter arising as the (necessarily constant, necessarily positive) umbilicity factor of the photon surface in the ambient manifold $(\R\times \mathcal{I}\times\Sphere^{n-1},\mathfrak{g})$. 
 The first goal of this paper (see \Cref{sec:spherical}) is to fully analyze and classify all solutions of the photon surface ODE~\eqref{eq:photonsurfintro} in a rather rich subclass $\mathcal{S}_{\text{ext}}$ of $\mathcal{S}$. For example, the class $\mathcal{S}_{\text{ext}}$ of $\mathcal{S}$ contains the (exterior) subextremal Reissner--Nordstr\"om and hence the (exterior) positive mass Schwarzschild spacetimes which are of central importance in General Relativity and Geometric Analysis. More generally, $\mathcal{S}_{\text{ext}}$ contains all asymptotically flat non-degenerate black hole spacetimes in class $\mathcal{S}$ in which a suitable auxiliary function $v_\text{eff}^{f}$ computed from the metric coefficient $f$ behaves schematically like in~\Cref{fig:yintro}. The function $v_\text{eff}^{f}$ is related to the effective potential in the analysis of null geodesics (see \Cref{subsec:generating}). 
 
\begin{figure}[h!]
\centering
\includegraphics[scale=0.95]{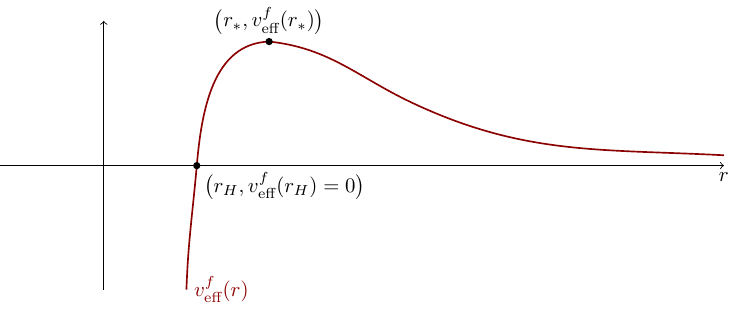}
\caption{Sketch of the auxiliary function $v_{\text{eff}}^f(r)=\frac{f(r)}{r^2}$.}
\label{fig:yintro}
\end{figure}
 
We obtain a complete qualitative understanding of all spherically symmetric photon surfaces in spacetimes of class $\mathcal{S}_{\text{ext}}$ and beyond. In a nutshell, we find that the parameter $\lambda_{*}^{2}=v_{\text{eff}}^f(r_{*})$ acts as a threshold, separating different types of behavior of solutions as indicated in \Cref{fig:casesintro}.

\begin{figure}[h!]
\centering
\includegraphics[scale=0.79]{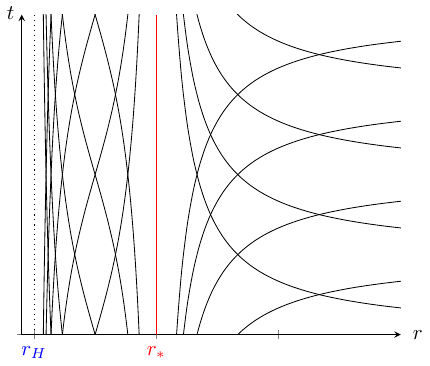}
\includegraphics[scale=0.79]{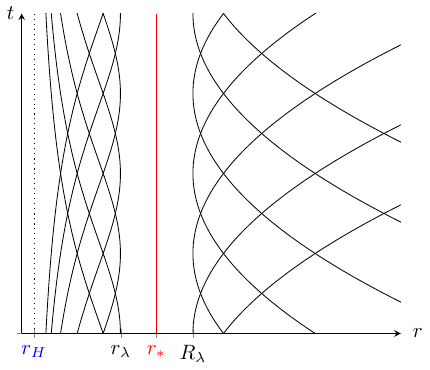}
\includegraphics[scale=0.79]{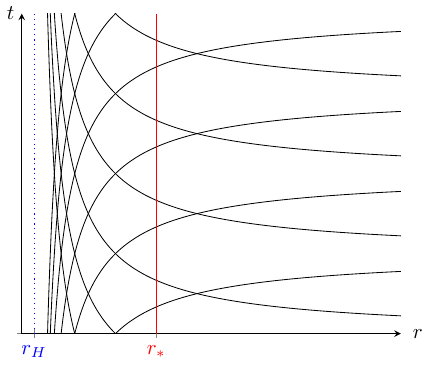}
\caption{Examples of photon surfaces for umbilicity factors $\lambda=\lambda_{*}$, $0<\lambda<\lambda_{*}$, and $\lambda>\lambda_{*}$, respectively, including time-translates and time-reflections thereof. The \textcolor{red}{vertical red line} represents the photon sphere at $r=r_{*}$ which has $\lambda=\lambda_{*}$.}
\label{fig:casesintro}
\end{figure}

Foertsch, Hasse, and Perlick~\cite{Perlicketal} found the same types of photon surfaces in the $2+1$-dim.\ spacetime obtained by restricting the $3+1$-dimensional positive mass Schwarzschild spacetime to its equatorial plane. This spacetime lies in the class $\mathcal{S}_{\text{ext}}$ studied here. Interestingly, they used a very different, dimension-specific method based on the integrability of a certain null frame.\\

\textbf{\emph{Relation to null geodesics:}} All spherically symmetric photon surfaces in a spacetime of class $\mathcal{S}$ are generated by rotating suitable null geodesics through $\mathbb{S}^{n-1}$ (see \cite[Prop. 3.14]{cederbaum2019photon}). Hence, our classification of spherically symmetric photon surfaces in spacetimes in class $\mathcal{S}_{\text{ext}}$ is closely related to the classification of null geodesics in these spacetimes. This explains the close analogy of our analysis with the well-known Schwarzschild null geodesic analysis, see for example \cite[pp. 380--384]{ONeill}. This aspect will be discussed further in \Cref{subsec:generating}.\\

\textbf{\emph{Topology of lifts to phase space:}} Moving to the phase space (i.e., to the cotangent bundle), it was shown in \cite[Prop.\ 3.18]{cederbaum2019photon} that the canonical lifts of the null bundles over the \emph{maximal} (under inclusion) spherically symmetric photon surfaces, together with the canonical lifts of the null bundles over suitably defined (maximal) principal null hypersurfaces partition the null section of the phase space of any spacetime in class~$\mathcal{S}$. From Bugden~\cite{Bugden}, we know that the lift of the (future) null bundle over the photon sphere in the (exterior) $n+1$-dimensional Schwarzschild spacetime of positive mass is diffeomorphic to $T^{1}\mathbb{S}^{n-1}\times\R^{2}$, where $T^{1}\mathbb{S}^{n-1}$ denotes the unit tangent bundle over $\mathbb{S}^{n-1}$ and where $n+1\geq4$. In \Cref{sec:topologyphasespace}, we will generalize this result to photon spheres in spacetimes of class $\mathcal{S}$ and show that it also applies to the canonical lift of the null bundle over any maximal spherically symmetric photon surface in these spacetimes. We will also compute the topology of the canonical lift of the (future or past) null bundle over any maximal principal null hypersurface in these spacetimes to be $\mathbb{S}^{n-1}\times \mathbb{R}^2$. In particular, we prove that the partition identified in~\cite[Prop. 3.18]{cederbaum2019photon} is \emph{not} a smooth foliation.\\

\textbf{\emph{Uniqueness of asymptotically flat electrostatic electro-vacuum spacetimes with equipotential photon surfaces:}} 
The above analysis in combination with a uniqueness result for photon surfaces in spacetimes in class $\mathcal{S}$ \cite[Theorem 3.8, Remark 3.10, Remark 3.11]{cederbaum2019photon} allows us to characterize \emph{all} photon surfaces in (exterior) subextremal Reissner--Nordstr\"om spacetimes of dimension $n+1\geq4$ (see \Cref{sec:prelim}). Complementing this complete characterization, we prove that any \emph{suitably asymptotically flat electrostatic electro-vacuum} spacetime of dimension $n+1\geq4$ which has an inner boundary consisting of \emph{equipotential (quasi-locally) subextremal} photon surfaces and non-degenerate black hole horizons must be isometric to an (exterior) subextremal Reissner--Nordstr\"om spacetime (see \Cref{sec:unique} and in particular \Cref{thm:uniqueness}). Here, an \emph{electrostatic spacetime} is meant to be a (standard) static spacetime carrying a time-independent electric potential. A photon surface $P^{n}$ in an electrostatic spacetime is called \emph{equipotential} if the static lapse function $N$, the electric potential $\Psi$, and the length of its derivative $\vert d\Psi\vert$ are constant on each time-slice $P^{n}\cap\{t=\text{const}\}$ of the photon surface $P^{n}$ (but may depend on the time coordinate~$t$); last but not least, we use a quasi-local definition of subextremality in terms of suitably defined quasi-local (electric) charge and mass (see \Cref{sec:prelimunique}).

\Cref{thm:uniqueness} generalizes to the electrostatic electro-vacuum setting the corresponding result~\cite[Theorem 4.1]{cederbaum2019photon} for static vacuum spacetimes, where the condition of being (quasi-locally) subextremal simplifies to being \emph{outward directed}, i.e., such that the (outward) normal derivative $\eta(N)$ of the static lapse function $N$ is strictly positive. It is left open in \cite{cederbaum2019photon} whether outward directed equipotential photon surfaces in asymptotically flat static vacuum spacetimes must have positive mean curvature. We prove said positivity in our more general setting (see \Cref{prop: H>0}). Moreover, we derive an evolution equation for the static lapse function $N$ along equipotential photon surfaces in static spacetimes without assuming any field equations (see \Cref{prop:evou}) and obtain further insights into the spacetime geometry near an equipotential photon surface (see \Cref{sec:prelimunique,sec: H>0}). 

The vacuum uniqueness result \cite[Theorem 4.1]{cederbaum2019photon} relies on work by the first named author~\cite{ndimunique} which has recently been reproven and generalized in the case of spin manifolds by Raulot~\cite{Raulot}. Both results in turn generalize results on and apply ideas from static vacuum photon sphere uniqueness in $3+1$ dimensions by the first named author and Galloway~\cite{cedergal}, see also \cite{CederPhoto}, as well as results on and ideas from static vacuum black hole uniqueness in $n+1\geq4$ dimensions by Gibbons, Ida, and Shiromizu~\cite{GIZ} and Hwang~\cite{Hwang} and in $3+1$ dimensions by Bunting and Masood-ul-Alam~\cite{BMuA}. Static vacuum black hole uniqueness for connected black holes has also been proved by many other authors, see the reviews \cite{Heusler,RobinsonReview,HeuslerLivRev}, for a (then) complete list of references on further contributions. In addition, see Simon's spinor proof recently described in Raulot's article~\cite[Appendix~A]{Raulot} and the recent article by Agostiniani and Mazzieri~\cite{AM} for a new approach in the connected case.

Sticking with electrostatic spacetimes, our uniqueness result generalizes and heavily relies on the second named author's electrostatic photon sphere uniqueness result~\cite{jahns2019photon}. This result extends and uses ideas from the first named author and Galloway's electrostatic photon sphere uniqueness result~\cite{Cedrgal2} to higher dimensions (and establishes how to apply it in the case only a weaker version of the electro-vacuum Einstein equation holds) and thereby also extends and/or relies on ideas from Ruback's and Masood-ul-Alam's electrostatic black hole uniqueness results~\cite{Ruback,MuA} and builds on work by Kunduri--Lucietti~\cite{bubble}, see also \cite{Rogatko1,Rogatko2,Chrusciel,CT,Heus1,Heus2}. For further results on photon spheres, in particular for uniqueness results in the electro-vacuum case and in the case of other matter fields, see for example~\cite{VE1,VE2,Heusler,RobinsonReview,HeuslerLivRev,GibbonsWarnick,yazadjiev, YazaLazov, YazaLazov2, Shoom, Tomi, Tomi2, Yoshino,Yoshino2,PhysRevD.93.064003,Koga:2020gqd,MS,PhysRevD.104.124016,Reiris,Beig_Simon,Harvie}.\\

\textbf{This article is structured as follows:} In \Cref{sec:prelim}, we will briefly introduce the definitions and known facts relevant for the analysis and classification of photon surfaces in spacetimes of class $\mathcal{S}$ and their lifts to phase space to be performed in \Cref{sec:spherical} and \Cref{sec:nullgeophase}, respectively. We will then formulate and prove  our uniqueness result in \Cref{sec:unique}. This section also includes a derivation of the relevant quasi-local and global properties of equipotential photon surfaces in asymptotically flat electrostatic electro-vacuum spacetimes. \Cref{sec:appendix} contains an ODE result we will exploit in \Cref{sec:spherical,sec:unique}; this may be of independent interest. A table of contents can be found on page \pageref{toc}.\\[1ex]

\textbf{Acknowledgements.} The authors would like to thank Greg Galloway, Anabel Miehe, Karim Mosani, Volker Perlick, Bernard Whiting, and Markus Wolff for helpful discussions, comments, and questions and Felix Salfelder and Georgios Vretinaris for help with the graphics. The first named author would like to extend thanks to the University of Vienna and to the Mittag-Leffler Institute for allowing her to work in stimulating environments. She is indebted to the Baden-W\"urttemberg Stiftung for the financial support of this research project by the Eliteprogramme for Postdocs. Her work is supported by the focus program on Geometry at Infinity (Deutsche Forschungsgemeinschaft,  SPP 2026). The work of the first and second named authors was supported by the Institutional Strategy of the University of T\"ubingen (Deutsche Forschungsgemeinschaft, ZUK 63).

\newpage
\tableofcontents \label{toc}
\newpage

\section{Preliminaries}\label{sec:prelim}
In this section, we will introduce some basic concepts, definitions, and important results we will make use of later. We will first introduce and discuss relevant properties of the (exterior) subextremal Reissner--Nordstr\"om spacetimes which will play a central role in the uniqueness theorem we will prove in \Cref{sec:unique} and which will be the prime example of spacetimes to which our analysis of spherically symmetric photon surfaces in \Cref{sec:spherical} applies. Next, we will recall the definitions and describe some well-known properties of photon surfaces in general spacetimes and of photon spheres in static spacetimes, leaving those properties and facts only relevant for the discussion of the relationship between null geodesics and photon surfaces for \Cref{sec:topologyphasespace} and those only relevant for the uniqueness theorem proved in \Cref{sec:unique} to be discussed in \Cref{sec:unique}. We will then introduce the class $\mathcal{S}$ of spacetimes forming the basis for the analysis of photon surfaces in \cite{cederbaum2019photon} and cite the relevant results from that work. Finally, we will explain how the (exterior) subextremal Reissner--Nordstr\"om spacetimes fit in the framework of those results.\\[-4ex]

\subsection{The Reissner--Nordstr\"om spacetimes}\label{sec:RN}
The \emph{(exterior) subextremal Reissner--Nordstr\"om spacetime of mass $m\in\R$ and charge $q\in\R$} of dimension $n+1\geq4$, satisfying the \emph{subextremality condition} $m>\vert q\vert$, is the spherically symmetric Lorentzian manifold given by $(\R\times(r_{m,q},\infty)\times\mathbb{S}^{n-1},\mathfrak{g}_{m,q})$, 
\begin{align}\label{rnmetric}
\mathfrak{g}_{m,q}&=-N_{m,q}^{2}dt^{2}+N_{m,q}^{-2}dr^{2}+r^{2}\Omega,
\end{align}
where $\Omega$ denotes the standard metric on $\mathbb{S}^{n-1}$, $N_{m,q}\colon(r_{m,q},\infty)\to\R^+$ is the \emph{(static) lapse function}
\begin{align}\label{rnlapse}
N_{m,q}(r)&=\sqrt{1-\frac{2m}{r^{n-2}}+\frac{q^2}{r^{2(n-2)}}},
\end{align}
and $r_{m,q}\definedas\left(m+\sqrt{m^2-q^2}\right)^{\frac{1}{n-2}}$ is the largest root of the radicand of the lapse function $N_{m,q}$. It carries an \emph{electric potential} $\Psi_{m,q}\colon (r_{m,q},\infty)\to\R$ given by
\begin{align}\label{rnelectric}
\Psi_{m,q}(r)&= \frac{\sqrt{(n-1)}\,q}{\sqrt{2(n-2)}\, r^{n-2}}.
\end{align}
The (exterior) subextremal Reissner--Nordstr\"om spacetime of mass $m$ and charge $q$ models the exterior region of a static, electrically charged black hole of mass $m$ and electric charge $q$ surrounded by electro-vacuum and suitably isolated from exterior influences. The black hole horizon is located "at" $r=r_{m,q}$ (and is not part of the spacetime as we have written it; however, the metric smoothly extends to the horizon upon a suitable coordinate change). 

The formulas~\eqref{rnmetric}, \eqref{rnlapse} also define a smooth, spherically symmetric Lorentzian manifold if $m$ and $q$ violate the physically reasonable subextremality condition $m>\vert q\vert$; we then speak of the \emph{extremal} or \emph{superextremal Reissner--Nordstr\"om spacetimes} if $m=\vert q\vert\neq0$ or if $m<\vert q \vert$ or $m=q=0$, and set $r_{m,q}\definedas m^\frac{1}{n-2}\geq0$ and $r_{m,q}\definedas0$, respectively. All Reissner--Nordstr\"om spacetimes satisfy the vacuum Einstein--Maxwell equations, see \Cref{sec:AF}. This also applies to the \emph{(interior) subextremal Reissner--Nordstr\"om spacetime of mass $m\in\R$ and charge $q\in\R$}, i.e., when $m>\vert q\vert$ but $0<r<\overline{r}_{m,q}\definedas \left(m-\sqrt{m^2-q^2}\right)^{\frac{1}{n-2}}$, the other root of the lapse function $N_{m,q}$.\\

For $q=0$, one recovers the well-known \emph{Schwarzschild(--Tangherlini) spacetimes} (with vanishing electric potential). If $m>0$, one finds $r_m\definedas r_{m,q=0}=(2m)^\frac{1}{n-2}$, and the metric \eqref{rnmetric} specializes to the exterior Schwarzschild metric describing the exterior region of a static black hole of mass $m$ surrounded by vacuum and suitably isolated from exterior influences. On the other hand, $\overline{r}_{m,q=0}=0$. If $m=0$, one finds $r_{m=0,q=0}=0$ and recovers the Minkowski metric written in spherical polar coordinates, and if $m<0$, one has $r_{m,q=0}=0$ and obtains the negative mass Schwarzschild spacetime which is of more theoretical interest only. All Schwarzschild spacetimes satisfy the vacuum Einstein equations, see \Cref{sec:AF}.

The Reissner--Nordstr\"om spacetime $(\R\times(r_{m,q},\infty)\times\mathbb{S}^{n-1},\mathfrak{g}_{m,q})$ of mass $m$, charge $q$, and dimension $n+1\geq4$ can be rewritten in \emph{isotropic coordinates} on $\R\times(s_{m,q},\infty)\times\mathbb{S}^{n-1}$ by setting $s=s(r)$ via
\begin{align}\label{eq:defrs}
r &\asdefined s\bigg(1+\frac{m+q}{2s^{n-2}}\bigg)^{\frac{1}{n-2}}\bigg(1+\frac{m-q}{2s^{n-2}}\bigg)^{\frac{1}{n-2}},
\end{align}
and keeping the angular and time coordinates the same. Here, 
\begin{align}\label{smq}
s_{m,q}&\definedas s(r_{m,q})=\begin{cases}\left(\frac{m^2-q^2}{4}\right)^{\frac{1}{2(n-2)}}&\text{ for }m>\vert q\vert\\\left(\frac{\vert q\vert-m}{2}\right)^{\frac{1}{n-2}}&\text{ for }m\leq\vert q\vert\end{cases}.
\end{align}
The metric then reads
\begin{align}\label{eq:isotropic}
\mathfrak{g}_{m,q}&=-\widetilde{N}^2_{m,q}dt^{2} + \varphi_{m,q}^{\frac{2}{n-2}}\,\delta\asdefined -\widetilde{N}^2_{m,q}dt^{2} + g_{m,q},
\end{align}
where $\delta=ds^2+s^2\Omega$ denotes the Euclidean metric in spherical polar coordinates, the lapse function $N_{m,q}$ and electric potential $\Psi_{m,q}$ transform to $\widetilde{N}_{m,q}=\widetilde{N}_{m,q}(s)$ defined by $\widetilde{N}_{m,q}(s(r))=N_{m,q}(r)$ and and $\widetilde{\Psi}_{m,q}=\widetilde{\Psi}_{m,q}(s)$ defined by $\widetilde{\Psi}_{m,q}(s(r))=\Psi_{m,q}(r)$. They are given by
\begin{align}
\widetilde{N}_{m,q}(s) &= \frac{ 1-\frac{m^2-q^2}{4s^{2(n-2)}}}{\left(1+\frac{m}{2s^{n-2}}\right)^2-\frac{q^2}{4s^{2(n-2)}}},\\
\widetilde{\Psi}_{m,q}(s) &=\frac{\sqrt{n-1}\,q}{\sqrt{2(n-2)}\, s^{n-2}\left( \left(1+\frac{m}{2s^{n-2}}\right)^2-\frac{q^2}{4s^{2(n-2)}}\right)},
\end{align}
respectively, and the conformal factor $\varphi_{m,q} = \varphi_{m,q}(s)$ reads
\begin{align}\label{eq:varphiiso}
\varphi_{m,q}(s)&=\bigg(1+\frac{m+q}{2s^{n-2}}\bigg)\bigg(1+\frac{m-q}{2s^{n-2}}\bigg).
\end{align} 
To distinguish verbally between the radial coordinates, we will refer to $r$ as the \emph{area radius} and to $s$ as the \emph{isotropic radius}.

For completeness sake, we note that this transformation does not apply when $m>\vert q\vert$ but $0<r<\overline{r}_{m,q}$ as \eqref{eq:defrs} maps onto $(r_{m,q},\infty)$ and $\overline{r}_{m,q}<r_{m,q}$.

We will use the Reissner--Nordstr\"om spacetimes in isotropic coordinates as a template to define a suitable asymptotic flatness condition for the uniqueness results proven in \Cref{sec:unique}, see \Cref{sec:AF}.

\subsection{Asymptotically flat electrostatic electro-vacuum spacetimes}\label{sec:AF}
Before we touch more abstract definitions, let us agree that all manifolds and tensor fields appearing in this paper are smooth unless explicitly stated otherwise, that submanifolds are necessarily embedded without a need to state it, and that spacetimes are nothing else but Lorentzian manifolds; in fact, all Lorentzian manifolds appearing here are (standard) static (see below for a definition) and hence automatically time-oriented. Also, we assume that all manifolds are connected.

\begin{Def}[Static spacetimes]\label{def:static}
A spacetime $(\mathfrak{L}^{n+1},\mathfrak{g})$ is called \emph{(standard) static} if it is a warped product of the form
\begin{align}\label{static}
\mathfrak{L}^{n+1}&=\R\times \slice,\quad \mathfrak{g}=-N^{2}dt^{2}+g,
\end{align}
where $(\slice,g)$ is a smooth Riemannian manifold and $N\colon\slice\to\R^{+}$ is a smooth, positive function called the \emph{(static) lapse function} of the spacetime.
\end{Def}

\begin{Rem}[Static spacetimes cont., (canonical) time-slices]\label{rem:slice}
We will slightly abuse standard terminology and also call a spacetime $(\mathfrak{L}^{n+1},\mathfrak{g})$ static if it is the closure of an open subset of a warped product static spacetime $(\R\times\slice,\mathfrak{g}=-N^{2}dt^{2}+g)$, $\mathfrak{L}^{n+1}\subseteq\R\times\slice$, provided $g$ and $N$ extend smoothly to this boundary. We do so to allow for inner boundary $\partial\mathfrak{L}$ not arising as a warped product. We will denote the \emph{(canonical) time-slices} $\lbrace{t=\text{const.}\rbrace}$ of a static spacetime $(\mathfrak{L}^{n+1},\mathfrak{g})$, $\mathfrak{L}^{n+1}\subseteq\R\times\slice$ by $\slice(t)$ for $t\in\R$ and continue to denote the induced metric and (restricted) lapse function on $\slice(t)$ by $g$ and $N$, respectively.
\end{Rem}

\begin{Def}[Electrostatic spacetimes]\label{def:electrostatic}
A system $(\R\times\slice,\mathfrak{g},\Psi)$ consisting of a static spacetime $(\R\times\slice,\mathfrak{g})$ and a smooth \emph{electric potential} $\Psi\colon\R\times\slice\to\R$ is called \emph{electrostatic} if $\Psi$ is time-independent, i.e., if $\Psi\colon\slice\to\R$.
\end{Def}

Next, recall that \emph{electro-vacuum spacetimes} are spacetimes satisfying the source-free Einstein--Maxwell equations
\begin{align}\label{eq:EM}
\mathfrak{Ric}_{\alpha\beta}-\frac 1 2 \mathfrak R \mathfrak g_{\alpha \beta}&=2\left(F_{\alpha\gamma}F_\beta^{\phantom{\beta}\gamma}-\frac 1 {4} \mathfrak g_{\alpha\beta}F^{\mu\nu}F_{\mu\nu}\right)\\\label{eq:EM2}
\left(\operatorname{\mathfrak{div}}F\right)_{\alpha}&=0
\end{align}
on $\mathfrak{L}^{n+1}$, where $F$ denotes the Maxwell tensor of the spacetime, $\mathfrak{Ric}$ and $\mathfrak{R}$ denote the Ricci tensor and scalar curvature of, and $\operatorname{\mathfrak{div}}$ denotes the divergence with respect to $(\mathfrak L^{n+1}, \mathfrak g)$. In the electrostatic case $(\mathfrak{L}^{n+1}=\R\times M^{n},\mathfrak{g}=-N^{2}dt^{2}+g,\Psi)$, it is typically assumed that
\begin{align}\label{eq:F}
F&=d\Psi\wedge dt.
\end{align}
In $n+1=4$ dimensions, electromagnetic duality allows us to assume \eqref{eq:F} without loss of generality, while for $n+1>4$ the vanishing of the magnetic field indeed poses an additional assumption (see e.g.\cite{bubble,jahns2019photon}). It is well-known that, assuming \eqref{eq:F}, the source-free Einstein--Maxwell equations reduce to the \emph{electrostatic electro-vacuum equations}
\begin{align}\label{EEVE3}
N\operatorname{Ric} &=\nabla^2 N-\frac{2\,d\Psi\otimes d\Psi}{N}+\frac{2|d\Psi|^2}{(n-1)N}g\\\label{EEVE4}
\operatorname{R}&=\frac{2|d\Psi|^2}{N^2}\\\label{EEVE2}
0&=\operatorname{div} \left(\frac{\operatorname{grad}\Psi}{N}\right)
\end{align} 
on $M^{n}$, where $\operatorname{Ric}$, $\operatorname{R}$, $\nabla^{2}$, $\operatorname{div}$, and $\operatorname{grad}$ denote the Ricci and scalar curvature, the Hessian, the divergence, and the gradient of $\left(M^n, g\right)$, respectively. Contracting~\eqref{EEVE3} and using~\eqref{EEVE4}, one finds
\begin{align}\label{EEVE1}
\Delta N &=\frac{2(n-2)|d\Psi|^2}{(n-1)N}
\end{align}
on $M^{n}$, where $\Delta$ denotes the Laplace-Beltrami operator with respect to $g$.

\begin{defi}[Electrostatic electro-vacuum spacetimes]
Let $\left( M^n, g, N, \Psi\right)$ be an \emph{electrostatic system}, i.e., such that $(\mathfrak{L}^{n+1}=\R\times M^{n},\mathfrak{g}=-N^{2}dt^{2}+g,\Psi)$ is an electrostatic spacetime. Then $\left( M^n, g, N, \Psi\right)$ and the associated spacetime are called \emph{electro-vacuum} if they satisfy~\eqref{EEVE3}--\eqref{EEVE2}.
\end{defi}

\begin{Rem}[Static vacuum equations]\label{rem:vacuum}
In case the electric potential $\Psi$ vanishes, the electrostatic electro-vacuum equations~\eqref{EEVE3}--\eqref{EEVE2} reduce to the \emph{static vacuum equations}
 \begin{align}\label{EVE1}
N\operatorname{Ric} &=\nabla^2 N\\\label{EVE2}
\operatorname{R}&=0,
\end{align} 
and an electrostatic electro-vacuum system $\left( M^n, g, N, \Psi=0\right)$ becomes a static vacuum system $\left( M^n, g, N\right)$. All photon surface and photon sphere notions that  will be introduced in this paper reduce to the well-known notions for static vacuum systems when $\Psi=0$.
\end{Rem}

It is well-known and easy to check that the Reissner--Nordstr\"om spacetimes satisfy~\eqref{EEVE3}--\eqref{EEVE2} and the Schwarzschild spacetimes satisfy \eqref{EVE1}, \eqref{EVE2}, see also \Cref{prop:genRN}. 

\begin{Rem}[Scaling invariance]\label{rem:scaling}
Let $\alpha,\mu>0$, $\beta\in\R$, and let $(\slice,g,N,\Psi)$ be an electrostatic electro-vacuum system. Then the rescaled system $(\slice,\mu^{2} g,\alpha N,\pm\alpha\Psi+\beta)$ is an electrostatic electro-vacuum system, as well.
\end{Rem}

We will use the following notion of asymptotic flatness for electrostatic systems and the corresponding spacetimes (as in~\cite{jahns2019photon,cederbaum2019photon}).
\begin{defi}[Asymptotic conditions]\label{def:asymptotics}
Let $m,q\in\R$. An \emph{electrostatic system} $\left( M^n, g, N, \Psi\right)$ of dimension $n\geq3$ is called \emph{asymptotically flat}, or, more precisely, \emph{asymptotic to (isotropic) Reissner-Nordstr\"om of mass $m$ and charge $q$} if the manifold $\slice$ is diffeomorphic to the union of a (possibly empty) compact set and an open \emph{end} $E^{n}$ which is diffeomorphic to $\R^{n}\setminus \overline{B}$, $\Phi=(x^{i})\colon E^{n}\to\R^{n}\setminus \overline{B^n_S(0)}$ for some $S>s_{m,q}$, and if
\begin{align}\label{AF}
(\Phi_{*}g)_{ij}-(g_{m,q})_{ij}&=\mathcal{O}_{2}(s^{-(n-1)}),\\\label{AFN}
\Phi_{*}N-\widetilde{N}_{m,q}&=\mathcal{O}_{2}(s^{-(n-1)}),\\\label{ADPsi}
\Phi_{*}\Psi-\widetilde{\Psi}_{m,q}&=\mathcal{O}_{2}(s^{-(n-1)})
\end{align}
for $i,j=1,\dots,n$ on $\R^{n}\setminus\overline{B}$ as $s\definedas \sqrt{(x^{1})^{2}+\dots+(x^{n})^{2}}\to\infty$. Here, $s_{m,q}$, $(g_{m,q})_{ij}$, $\widetilde{N}_{m,q}$, and $\widetilde{\Psi}_{m,q}$ denote the isotropic horizon radius, spatial Reissner--Nordstr\"om metric, lapse function, and electric potential of mass $m$ and charge $q$ in isotropic coordinates, respectively, see \Cref{sec:RN}.

We will abuse language and call a subset $\mathfrak{L}^{n+1}\subseteq\R\times \slice$ of an asymptotically flat electrostatic spacetime $(\R\times\slice,\mathfrak{g},\Psi)$ asymptotically flat as long as $\mathfrak{L}^{n+1}$ has timelike inner boundary $\partial\mathfrak{L}$. Also, we will call a static system $(\slice,g,N)$ asymptotically flat if it satisfies~\eqref{AF}, \eqref{AFN}.

Here and in the following, we say that a smooth function $f\colon \R^{n}\setminus\overline{B_{S}^{n}(0)} \to\R$ \emph{satisfies $f = O_{2}(r^{p})$ as $r\to\infty$}  for some $p\in\R$ if there exists a constant $C > 0$ such that $\vert D^{\alpha} f \vert\leq C r^{-p-\vert\alpha\vert}$ in $\R^{n}\setminus\overline{B_{S}^{n}(0)} $ for every multi-index $\alpha$ satisfying
$\vert\alpha\vert\leq2$.
\end{defi}

Obviously, the Reissner--Nordstr\"om electrostatic system of mass $m$ and charge $q$ is asymptotically flat of mass $m$ and charge $q$ in this sense and it can be shown by a rather straightforward computation that the asymptotic mass and charge parameters are in fact uniquely determined. 

\begin{Rem}[Metric completeness]\label{rem:complete}
Asymptotically flat electrostatic systems, with or without boundary, are necessarily metrically complete and geodesically complete (up to the boundary) with at most finitely many boundary components which are necessarily all closed, see for example~\cite[Appendix]{CGM}.
\end{Rem}

\subsection{Generalized Reissner--Nordstr\"om spacetimes}\label{sec:genRN}
As a preparation for the local characterization of equipotential photon surfaces in electrostatic, electro-vacuum spacetimes in \Cref{sec: H>0}, let us introduce the following generalization of Reissner--Nordstr\"om spacetimes. Some generalized Reissner--Nordstr\"om spacetimes are also interesting examples for spacetimes in class $\mathcal{S}$ (see \Cref{sec:classS}).

\begin{Def}[Generalized Reissner--Nordstr\"om spacetimes]\label{def:genRN}
Let $\kappa\in\{0,\pm1\}$, $m,q\in\R$, $n\geq3$, and let $(\surf_{\ast},\sigma_{\ast})$ be an $(n-1)$-dimensional Einstein manifold of scalar curvature 
\begin{align}
\operatorname{R}_{\sigma_{\ast}}&=(n-1)(n-2)\kappa.
\end{align}
The \emph{generalized Reissner--Nordstr\"om spacetime of dimension $n+1$, mass $m$, (electric) charge $q$, parameter $\kappa$, and base $(\surf_{\ast},\sigma_{\ast})$} is the (possibly disconnected) Lorentzian manifold given by $(\R\times J_{\kappa,m,q}\times \surf_{\ast},\mathfrak{g}_{\sigma_{\ast},m,q})$, with
\begin{align}\label{genrnmetric}
\mathfrak{g}_{\sigma_{\ast},m,q}&\definedas-N_{\kappa,m,q}^{2}dt^{2}+N_{\kappa,m,q}^{-2}dr^{2}+r^{2}\sigma_{\ast},
\end{align}
where $N_{\kappa,m,q}\colon J_{\kappa,m,q}\to\R^+$ is the \emph{(static) lapse function}
\begin{align}\label{genrnlapse}
N_{\kappa,m,q}(r)&\definedas\sqrt{\kappa-\frac{2m}{r^{n-2}}+\frac{q^2}{r^{2(n-2)}}}.
\end{align}
Here, we assume that the parameters are chosen such that the domain of definition $J_{m,q,\kappa}$ given by
\begin{align}\label{genrndomain}
J_{\kappa,m,q}&\definedas\left\{r\in\R^{+}\,\vert\,\kappa-\frac{2m}{r^{n-2}}+\frac{q^{2}}{r^{2(n-2)}}>0\right\}
\end{align}
is non-empty, $J_{\kappa,m,q}\neq\emptyset$. It carries an \emph{electric potential} $\Psi_{\kappa,m,q}\colon J_{\kappa,m,q}\to\R$ given by
\begin{align}\label{genrnelectric}
\Psi_{\kappa,m,q}(r)&\definedas \frac{\sqrt{n-1}\,q}{\sqrt{2(n-2)}\, r^{n-2}}.
\end{align}
\end{Def}

\begin{Rem}[On $J_{\kappa,m,q}$]
Depending on $\kappa$, $m$, and $q$, the set $J_{\kappa,m,q}$ can be disconnected (as it is well-known to be for subextremal Reissner--Nordstr\"om spacetimes). See \Cref{rem:BHgenRN} for more details.
\end{Rem}

\begin{Rem}[On the range of $\kappa$]
In view of \Cref{rem:scaling}, the only reason why we ask $\kappa\in\{0,\pm1\}$ is that the radial coordinate $r$ has a specific geometric meaning (only) in case $\kappa=\pm1$, see e.g.\ \Cref{defprop:projection}.
\end{Rem}

The most important fact about generalized Reissner--Nordstr\"om spacetimes is that they solve the electrostatic electro-vacuum equations.
\begin{Prop}[Generalized Reissner--Nordstr\"om spacetimes are electrostatic electro-vacuum spacetimes]\label{prop:genRN}
Let $\kappa\in\{0,\pm1\}$, $m,q\in\R$, and let $(\surf_{\ast},\sigma_{\ast})$ be an $(n-1)$-dimensional Einstein manifold of scalar curvature $\operatorname{R}_{\sigma_{\ast}}=(n-1)(n-2)\kappa$ for $n\geq3$. The generalized Reissner--Nordstr\"om spacetime $(\R\times J_{\kappa,m,q}\times\surf_{\ast},\mathfrak{g}_{\sigma_{\ast},m,q})$ is an electrostatic spacetime which solves \eqref{EEVE3}-\eqref{EEVE2}.
\end{Prop}
\begin{proof}
First of all, $(\R\times J_{\kappa,m,q}\times\surf_{\ast},\mathfrak{g}_{\sigma_{\ast},m,q})$ is clearly electrostatic by definition. Let us drop all references to $\kappa$, $m$, and $q$ for notational simplicity. We exploit the warped product form of the metric: By ``tangential'' and ``normal'', we will mean tangential and normal to the leaves $\{(t,r)\}\times\surf_{\ast}$ of $\{t\}\times J\times\surf_{\ast}$, respectively. Then first of all, the tangential-normal component of \eqref{EEVE3} is trivially satisfied as both sides vanish. Next, let us verify the tangential-tangential component of \eqref{EEVE3}: From the warped product structure, one finds 
\begin{align*}
\Ric&=\,^{\sigma_{\ast}}\!\Ric-\left((n-2)N+rN'\right)N\sigma_{\ast},
\end{align*}
where $'$ denotes taking an $r$-derivative and $^{\sigma_{\ast}}\!\Ric$ denotes the Ricci tensor of $\sigma_{\ast}$. Using the fact that $\sigma_{\ast}$ is an Einstein metric with $\operatorname{R}_{\sigma_{\ast}}=(n-1)(n-2)\kappa$, the tangential-tangential component of \eqref{EEVE3} reduces to the ODE 
\begin{align*}
N^{2}-\kappa&=-\frac{r(N^{2})'}{n-2}-\frac{r^{2}(\Psi')^{2}}{n-1}
\end{align*}
for $N^{2}$ and $\Psi'$. As both sides are effectively independent of $\kappa$, this ODE is satisfied because it is satisfied in the Reissner--Nordstr\"om case. Similarly, the normal-normal component of \eqref{EEVE3} reduces to the ODE
\begin{align}\label{eq:proofgen}
(N^{2})''+\frac{(n-1)(N^{2})'}{r}&=\frac{4(n-2)(\Psi')^{2}}{n-1}
\end{align}
which involves only $(N^{2})'$ and $\Psi'$ and is hence also satisfied because it is satisfied in the Reissner--Nordstr\"om case. To see that \eqref{EEVE4} is satisfied, we use that we can equivalently verify \eqref{EEVE1} which in fact also reduces to \eqref{eq:proofgen} and is hence satisfied. Finally, \eqref{EEVE2} is equivalent to
\begin{align*}
\Psi''+\frac{(n-1)\Psi'}{r}&=0
\end{align*}
as $N>0$, which is manifestly satisfied as $\Psi$ does not depend on $\kappa$.
\end{proof}

\subsection{Photon surfaces, photon spheres, and black hole horizons}\label{sec:photonsphere}
We are now ready to recall the definition of photon surfaces, the central object in this paper (see \cite{CVE,Perlick} for more information).

\begin{Def}[Photon surfaces]\label{def:photo-surf}
A timelike hypersurface $\photo\hookrightarrow\mathfrak{L}^{n+1}$ in a Lorentzian manifold $(\mathfrak{L}^{n+1},\mathfrak{g})$ is called a \emph{photon surface} if every null geodesic initially tangent to $\photo$ remains tangent to $\photo$ as long as it exists or in other words if $\photo$ is \emph{null totally geodesic}.
\end{Def}

It will be useful to know that, by an algebraic observation, being a null totally geodesic timelike hypersurface is equivalent to being a totally umbilic timelike hypersurface:
\begin{Prop}[\!\!{\cite[Theorem II.1]{CVE}, \cite[Proposition 1]{Perlick}}]\label{prop:umbilic}
Let $(\mathfrak{L}^{n+1},\mathfrak{g})$ be a Lorentzian manifold and $\photo\hookrightarrow\mathfrak{L}^{n+1}$ an embedded, timelike hypersurface. Then $\photo$ is a photon surface if and only if it is \emph{totally umbilic}, that is, if and only if its second fundamental form is pure trace.
\end{Prop}

It is well-known that the Minkowski spacetime hosts many photon surfaces, namely all one-sheeted hyperboloids and all timelike hyperplanes (see for example \cite[pp. 116--117]{ONeill}). As stated above, all Reissner--Nordstr\"om spacetimes are static, by the following definition.

In the context of static spacetimes, we will use the following definition of photon spheres, extending that of~\cite{CDiss,CederPhoto,CVE,cederbaum2019photon}, as introduced in \cite{YazaLazov,Cedrgal2,jahns2019photon}, see also \Cref{rem:normpsi}. 
\begin{Def}[Photon spheres]\label{def:photonsphere}
Let $(\mathfrak{L}^{n+1},\mathfrak{g},\Psi)$ be an electrostatic spacetime, $\photo\hookrightarrow\mathfrak{L}^{n+1}$ a photon surface. Then $\photo$ is called a \emph{photon sphere} if $\photo = \R \times\surf$ for some smooth hypersurface $\surf\hookrightarrow \slice$ and if the lapse function $N$, the electric potential $\Psi$ and the length of its derivative $\vert d\Psi\vert_{\mathfrak{g}}$ are constant along $P^n$. In a static spacetime, the conditions on $\Psi$ are dropped.
\end{Def}

The most important example of a photon sphere is the spherically cylindrical hypersurface $\lbrace r = (nm)^{\frac{1}{n-2}} \rbrace$ of the Schwarzschild spacetime of mass $m>0$ and dimension $n+1\geq4$. An example of an (electrostatic) photon sphere is the spherically cylindrical hypersurface 
\begin{align}
\left\lbrace r =  \left(\frac{nm + \sqrt{n^2m^2 - 4(n-1)q^2}}{2}\right)^{\frac{1}{n-2}}\right\rbrace
\end{align}
of the subextremal Reissner--Nordstr\"om spacetime of mass $m$, charge $q$, and dimension $n+1\geq4$, see \Cref{cor:photonsphere in rn}. We will make use of the following properties defined for photon surfaces, generalizing and adopting the corresponding definitions from~\cite{cederbaum2019photon}, see also \Cref{rem:normpsi}.
\begin{defi}[Equipotential photon surfaces]\label{def:equipot}
A photon surface $P^n$ in an electrostatic spacetime $(\mathfrak{L}^{n+1},\mathfrak{g},\Psi)$ is called \emph{equipotential} if the static lapse $N$, the electric potential $\Psi$ and the length of its derivative $\vert d\Psi\vert_{\mathfrak{g}}$ are constant on each standard time slice $\Sigma^{n-1}(t_{0})\definedas P^n\cap \slice(t_{0})$. In a static spacetime $(\mathfrak{L}^{n+1},\mathfrak{g})$, the conditions on $\Psi$ are dropped.
\end{defi}

\begin{defi}[Non-degenerate and outward/inward directed photon surfaces]
Let $(\mathfrak{L}^{n+1},\mathfrak{g})$ be a static spacetime with lapse function $N$, $\photo\hookrightarrow\mathfrak{L}^{n+1}$ a photon surface. Then $\photo$ is called \emph{non-degenerate} if $dN(x)\neq0$ for all $x\in\photo$. If, in addition, $(\mathfrak{L}^{n+1},\mathfrak{g})$ is asymptotically flat with unit normal $\eta$ pointing \emph{outward}, i.e., towards the asymptotic region of the spacetime, then $\photo$ is called \emph{outward} resp.\ \emph{inward directed} if the normal derivative $\eta(N)$ is positive resp.\ negative along~$\photo$.
\end{defi}

As $N=1$ in the Minkowski spacetime, all photon surfaces in the Minkowski spacetime are automatically equipotential and degenerate. As was shown in \cite{cederbaum2019photon} and will be discussed further in \Cref{sec:classS,sec:spherical}, more general static, spherically symmetric spacetimes host many spherically symmetric and hence equipotential photon surfaces; the conditions of being non-degenerate or outward/inward directed then boil down to conditions on the radial derivative of the lapse function along them, see \Cref{prop:outwardspherical} and \Cref{coro:RNoutward}.

\begin{Rem}[Symmetries acting on photon surfaces]\label{rem:tinv}
Let $(\R\times\slice,\mathfrak{g})$ be a static spacetime. As $\d_{t}$ is a Killing vector field of $(\R\times\slice,\mathfrak{g})$, the time-translation of any (outward directed, equipotential) photon surface $P^n\hookrightarrow(\R\times\slice,\mathfrak{g})$ will also be an (outward directed, equipotential) photon surface in $(\R\times \slice,\mathfrak{g})$. As $(\R\times \slice,\mathfrak{g})$ is automatically also time-reflection symmetric (i.e., $t \mapsto -t$ is an isometry), the time-reflection of any (outward directed, equipotential) photon surface $P^n\hookrightarrow(\R\times\slice,\mathfrak{g})$ will also be a (outward directed, equipotential) photon surface in $(\R\times \slice,\mathfrak{g})$. Photon spheres are fixed points of both of these symmetries.

Moreover, if $(\R\times\slice,\mathfrak{g})$ is electrostatic with electric potential $\Psi$, then $\Psi$ is also invariant under both of these symmetries by definition.
 \end{Rem}

Before we move on, let us briefly recall the relevant definitions of black hole horizons.
\begin{Def}[Killing horizons, non-degeneracy]\label{def:Killing}
Let $(\mathfrak{L}^{n+1},\mathfrak{g})$ be a static spacetime. We say that a smooth null hypersurface $\mathcal{H}^{n}\hookrightarrow (\mathfrak{L}^{n+1},\mathfrak{g})$ is a \emph{Killing horizon} if the static lapse function $N$ of the spacetime vanishes along $\mathcal{H}^{n}$. A Killing horizon is\footnote{This is equivalent to the standard definition, see \cite{KW}.} \emph{non-degenerate} if $dN\neq0$ along $\mathcal{H}^{n}$.
\end{Def}
This definition also applies when $\mathcal{H}^{n}\subseteq\partial\mathfrak{L}^{n+1}$. It is well-known and can be checked by changing to isotropic coordinates that subextremal Reissner--Nordstr\"om spacetimes have a non-degenerate Killing horizon ``at'' $r=r_{m,q}$ and by changing to Gaussian null coordinates that extremal Reissner--Nordstr\"om spacetimes have a degenerate Killing horizon ``at'' $r=r_{m,q}$ (see e.g. \cite[Section 2.1]{KL}) when $m=\vert q\vert\neq0$. Superextremal Reissner--Nordstr\"om spacetimes do not possess Killing horizons.

In the language of canonical time slices $(\slice(t),g,N,\Psi)$ of electrostatic electro-vacuum spacetimes, non-degenerate Killing horizons give rise to hypersurfaces $\surf(t)\hookrightarrow(\slice(t),g,N,\Psi)$ where $N=0$ (where $\surf(t)\subseteq\partial\slice(t)$ is also permitted). From \eqref{EEVE4}, it follows that $\vert d\Psi\vert=N^{2}\operatorname{R}=0$ on a (time slice of a) non-degenerate Killing horizon $\surf(t)$ so that in particular $\Psi$ is constant along non-degenerate Killing horizons. From \eqref{EEVE3}, we thus learn that $\nabla^{2}N=0$ on $\surf(t)$ which implies that (time slices of) non-degenerate Killing horizons $\surf(t)\hookrightarrow(\slice(t),g,N,\Psi)$ are totally geodesic. We summarize this as follows.

\begin{deflem}[Properties of non-degenerate electrostatic horizons]
Let $(\slice,g,N,\Psi)$ be an electrostatic electro-vacuum system. A \emph{non-degenerate (electrostatic) (black hole) horizon in $(\slice,g,N,\Psi)$} is a smooth, closed, totally geodesic hypersurface of $\slice$ where $N=0$, $\nu(N)\neq0$, $\Psi=\text{const.}$, and $\nu(\Psi)=0$  for some choice of unit normal $\nu$ to $\surf$.
\end{deflem} 

For later reference, let us document the precise forms of the sets $J_{\kappa,m,q}$ and the locations of the Killing horizons in generalized Reissner--Nordstr\"om spacetimes.
\begin{Rem}[Killing horizons in Generalized Reissner--Nordstr\"om spacetimes]\label{rem:BHgenRN}
By a straightforward computation, one finds that generalized Reissner--Nordstr\"om spacetimes with $\kappa=1$ have $J_{1,m,q}=(r_{m,q},\infty)$ if $m\leq\vert q\vert$ or if $q=0$, while $J_{1,m,q}=(0,\overline{r}_{m,q})\cup(r_{m,q},\infty)$ if $m>\vert q\vert>0$ as for Reissner--Nordstr\"om spacetimes, using $r_{m,q}$ and $\overline{r}_{m,q}$ from \Cref{sec:RN}. For $\kappa=0$, $J_{0,m,q}=\R^{+}$ for $m<0$ and for $m=0$ if $q\neq0$, $J_{0,m=0,q=0}=\emptyset$, and $J_{0,m,q}=\{r\in\R^{+}\,\vert\,r^{n-2}<\frac{q^{2}}{2m}\}$ for $m>0$. For $\kappa=-1$, $J_{-1,m,q=0}=\emptyset$ when $m\geq0$ and $q=0$, and $J_{-1,m,q}=\{r\in\R^{+}\,\vert\,r^{n-2}<-m+\sqrt{m^{2}+q^{2}}\}$ when $m<0$ or $m\geq0$ and $q\neq0$. 

For discussing Killing horizons, we will continue to abuse notation as in the classical Reissner--Nordstr\"om spacetimes (see \Cref{sec:RN}), hiding suitable changes of coordinates (to isotropic or Gaussian null coordinates) near the Killing horizons. For $\kappa=1$,  $\{r=r_{m,q}\}$ and $\{r=\overline{r}_{m,q}\}$ are non-degenerate Killing horizons when $m>\vert q\vert>0$ and $\{r=r_{m,q}\}$ is a non-degenerate Killing horizon when $m>0$, $q=0$. Furthermore, $\{r^{n-2}=m\}$ can be seen to be a degenerate Killing horizon when $m=\vert q\vert>0$, and there are no Killing horizons for $m<\vert q\vert$ or $m=q=0$. For $\kappa=0$, $\{r^{n-2}=\frac{q^{2}}{2m}\}$ is a non-degenerate Killing horizon when $m>0$ and $q\neq0$, otherwise there are no Killing horizons. For $\kappa=-1$, $\{r^{n-2}=-m+\sqrt{m^{2}+q^{2}}\}$ is a non-degenerate Killing horizon when $m<0$ or when $m\geq0$ and $q\neq0$, otherwise there are no Killing horizons.
\end{Rem}
\newpage
\subsection{Photon surfaces in class $\mathcal{S}$}\label{sec:classS}
We consider the class $\mathcal{S}$ of spacetimes $(\R\times \slice,\mathfrak{g})$ of the form 
\begin{align}\label{def:ST}
\slice&= \mathcal{I}\times\Sphere^{n-1}\ni (r,\xi)
\end{align}
for an open interval $\mathcal{I}\subseteq(0,\infty)$, finite or infinite, and so that there exists a smooth, positive function $f\colon\mathcal{I}\to\R^{+}$ called \emph{metric coefficient} for which we can express the spacetime metric $\mathfrak{g}$ as
\begin{align}\label{def:STmetric}
\mathfrak{g}&=-f(r)dt^{2}+\frac{1}{f(r)}dr^{2}+r^{2}\Omega
\end{align}
in the global coordinates $t\in\R$, $(r,\xi)\in \mathcal{I}\times\Sphere^{n-1}$, where $\Omega$ again denotes the canonical metric on~$\Sphere^{n-1}$. Every spacetime $(\R\times \slice,\mathfrak{g})\in\mathcal{S}$ is clearly spherically symmetric and moreover naturally (standard) static with lapse function $N=\sqrt{f}$ and Riemannian metric $g=f(r)^{-1}dr^2+r^{2}\Omega$, where we are slightly abusing notation as $N$ is defined on $\slice$ while $f$ is defined only on $\mathcal{I}$. 

We note that the spacetimes $(\R\times \slice,\mathfrak{g})\in\mathcal{S}$ are not assumed to satisfy any kind of Einstein equations or have any special type of asymptotic behavior towards the boundary of the radial interval $\mathcal{I}$. All Reissner--Nordstr\"om spacetimes and hence in particular all Schwarzschild spacetimes and Minkowski spacetimes lie in class $\mathcal{S}$ (when written in the area radius coordinate $r$). Of course, the class is much richer and contains for example the (anti-)de Sitter and the Schwarzschild--(anti-) de Sitter spacetimes. 

For spacetimes in class $\mathcal{S}$, it is useful to make the following definition of spherically symmetric photon surfaces and their radial profiles.

\begin{Def}[Spherically symmetric timelike hypersurfaces, radial profiles {\cite[Definition~3.3]{cederbaum2019photon}}]
Let $(\R\times \slice,\mathfrak{g})\in\mathcal{S}$ for $n\geq2$. A connected, timelike hypersurface $\photo\hookrightarrow(\R\times \slice,\mathfrak{g})$ will be called \emph{spherically symmetric} if, for each $t_{0}\in \R$ for which the intersection $\Sigma^{n-1}(t_0) \defeq  \photo\cap\lbrace t = t_{0}\rbrace\neq\emptyset$ is non-empty, there exists a radius $r_{0}\in \mathcal{I}$ (where $\slice=\mathcal{I}\times\Sphere^{n-1}$) such that
\begin{align}\label{def:sphsymmphoto}
\Sigma^{n-1}(t_0) &=\lbrace t_{0}\rbrace\times\lbrace r_{0}\rbrace\times\Sphere^{n-1}\subset \lbrace t_{0}\rbrace\times \slice.
\end{align}
A future timelike curve $\gamma\colon I\to \photo$ parametrized by arclength on some open interval $I\subset\R$ is called a \emph{radial profile} of $\photo$ if $\gamma'\in\hull\lbrace\d_{t},\d_{r}\rbrace\subset T_{\gamma'}(\R\times \slice)$ on $I$ and if the orbit of $\gamma$ under the rotation generates $\photo$.
\end{Def}
Obviously spherically symmetric photon surfaces are necessarily equipotential (if there is no electric field $\Psi$, then we just ignore the corresponding part of the definition, otherwise we assume that the electric potential is also spherically symmetric).

With this definition at hand, we can state one of the main results of \cite{cederbaum2019photon}, a local characterization of spherically symmetric photon surfaces in spacetimes of class $\mathcal{S}$.

\begin{Thm}[\!\!{\cite[Theorem 3.5]{cederbaum2019photon}}]\label{thm:sphsymm}
Let $(\R\times \slice,\mathfrak{g})\in\mathcal{S}$ and let $\photo\hookrightarrow(\R\times \slice,\mathfrak{g})$ be a 
spherically symmetric timelike hypersurface. Assume that $\photo\hookrightarrow(\R\times \slice,\mathfrak{g})$ is a photon surface, with umbilicity factor $\a$, i.e., $\mathfrak{h}=\a p$, where $p$ and $\mathfrak{h}$ are the induced metric and second fundamental form of $\photo\hookrightarrow(\R\times \slice,\mathfrak{g})$, respectively. Let $\gamma\colon I\to \photo$ be a radial profile for $\photo$ and write $\gamma(s)=(t(s),r(s),\xi_{*})\in\R\times\mathcal{I}\times\Sphere^{n-1}$ for some $\xi_{*}\in\Sphere^{n-1}$. Then $\a$ is a positive constant and \underline{\emph{either}} $r\equiv r_{*}$ along $\gamma$ for some $r_{*}\in \mathcal{I}$ at which the \emph{photon sphere condition} 
\begin{align}\label{eq:photonsphere}
f'(r_{*})r_{*}=2f(r_{*})
\end{align}
holds, $\a=\frac{\sqrt{f(r_{*})}}{r_{*}}$, and $(\photo,p)=(\R\times\Sphere^{n-1},-f(r_{*})dt^{2}+r^{2}_{*}\Omega)$ is a cylinder and thus a photon sphere, \underline{\emph{or}} $r=r(t)$ can globally be written as a smooth, non-constant function of $t$ in the range of $\gamma$ and $r=r(t)$ satisfies the \emph{photon surface ODE}
\begin{align}\label{eq:ODE3}
\left(\frac{dr}{dt}\right)^{2} &= \frac{f(r)^{2}\,(\a^2 r^2 - f(r))}{\a^{2} r^{2}}.
\end{align}

Conversely, whenever the photon sphere condition \eqref{eq:photonsphere} holds for some $r_{*}\in \mathcal{I}$, then the cylinder $(\photo,p)=(\R\times\Sphere^{n-1},-f(r_{*})dt^{2}+r^{2}_{*}\Omega)$ is a photon sphere in $(\R\times \slice,\mathfrak{g})$ with umbilicity factor $\a=\frac{\sqrt{f(r_{*})}}{r_{*}}$. Also, any smooth, non-constant solution $r=r(t)$ of the photon surface ODE~\eqref{eq:ODE3} for some constant $\a>0$ gives rise to a photon surface in $(\R\times \slice,\mathfrak{g})$ with umbilicity factor $\a$.  
\end{Thm}

The complementary question about existence of non-spherically symmetric photon surfaces in spacetimes of class $\mathcal{S}$ was also answered in \cite{cederbaum2019photon}.
\begin{Thm}[\!\!{\cite[Theorem 3.8]{cederbaum2019photon}}]\label{thm:isotropic}
Let $n\geq3$, $I\subseteq\R^{+}$ an open interval, and let $\widetilde{N},\psi\colon I\to\R^{+}$ be smooth, positive functions. Set $D^{n}\definedas \lbrace{y\in\R^{n}\,\vert\,\vert y\vert=s\in I\rbrace}$ and consider the \emph{static, isotropic spacetime}
$\left(\R\times D^{n},\mathfrak{g}=-\widetilde{N}^{2}dt^{2}+\psi^{2}\,\delta\right)$ of lapse $\widetilde{N}=\widetilde{N}(s)$ and conformal factor~$\psi=\psi(s)$. We write $\widetilde{g}\definedas \psi^{2}\,\delta$. A timelike hypersurface $\photo$ in $(\R\times D^{n},\mathfrak{g})$ is called \emph{isotropic} if $\photo\cap \lbrace{t=\text{const.}\rbrace}=\mathbb{S}^{n-1}_{s(t)}(0)\subset D^{n}$ for some radius $s(t)\in I$ for every $t$ for which $\photo\cap \lbrace{t=\text{const.}\rbrace}\neq\emptyset$. A \emph{(partial) centered vertical hyperplane} in $(\R\times D^{n},\mathfrak{g})$ is the restriction of a timelike hyperplane in the Minkowski spacetime containing the $t$-axis to $\R\times D^{n}$, i.e., a set of the form $\lbrace{(t,y)\in\R\times D^{n}\,\vert\, y\cdot u=0\rbrace}$ for some fixed Euclidean unit vector $u\in\R^{n}$, where $\cdot$ denotes the Euclidean inner product. Centered vertical hyperplanes are totally geodesic in $(\R\times D^{n},\mathfrak{g})$. Assume furthermore that the functions $\widetilde{N}$ and $\psi$ satisfy
\begin{align}\label{eq:unless}
\frac{\widetilde{N}'(s)}{\widetilde{N}(s)}&\neq\frac{\psi'(s)}{\psi(s)}
\end{align}
for all $s\in I$. Then any photon surface in $(\R\times D^{n},\mathfrak{g})$ is either (part of) an isotropic photon surface or (part of) a centered vertical hyperplane.
\end{Thm}

For the sake of completeness, let us note that locally, any spacetime of class $\mathcal{S}$ can be rewritten in isotropic form and will satisfy condition \eqref{eq:unless} provided it is nowhere conformally flat, see \cite[Remarks 3.10,  3.11]{cederbaum2019photon} for more details.

Let us close this section by noting that the Reissner--Nordstr\"om spacetimes, written in isotropic coordinates \eqref{eq:isotropic} are isotropic spacetimes as described in \Cref{thm:isotropic} whenever $r>r_{m,q}$, with $\psi=\varphi_{m,q}^\frac{1}{n-2}$. They satisfy condition \eqref{eq:unless} throughout unless $0<m<\vert q\vert$ as can be seen by a straightforward computation (and, of course, except in the Minkowski case $m=q=0$, where translated hyperpoloids also occur and where \eqref{eq:unless} is violated everywhere). For $0<m<\vert q\vert$, they satisfy \eqref{eq:unless} except for one specific isotropic radius $s_{*}$ satisfying $s^{n-2}_{*}>\frac{q^{2}-m^{2}}{2m}$, as can be seen by a tedious but straightforward computation. At this specific radius $s_{*}$, these superextremal Reissner--Nordstr\"om spacetimes are in fact (infinitesimally) conformally flat. This proves the following corollary, extending \cite[Corollary 3.9]{cederbaum2019photon} to electro-vacuum and to negative mass.
\begin{Cor}[Photon surfaces in Reissner--Nordstr\"om spacetimes]\label{coro:RNphoton}
Let $n+1\geq4$, $m, q\in\R$,  and consider the $n+1$-dimensional Reissner--Nordstr\"om spacetime of mass $m$ and charge $q$ with $r>r_{m,q}$ and suppose that $m$ and $q$ do not satisfy $0<m<\vert q\vert$. Then any connected photon surface in this spacetime is either (part of) a centered vertical hyperplane as described in \Cref{thm:isotropic} or (part of) a spherically symmetric photon surface as characterized in \Cref{thm:sphsymm}.
\end{Cor}

\section{Spherically symmetric photon surfaces in class $\mathcal{S}_{\text{ext}}$}\label{sec:spherical}
In this section, we will prove existence and uniqueness\footnote{modulo the symmetries of the problem discussed in \Cref{rem:tinv}.} results for spherically symmetric photon surfaces in spacetimes arising in a certain subclass $\mathcal{S}_{\text{ext}}\subset\mathcal{S}$ which among others contains the (exterior) positive mass Schwarzschild and (exterior) subextremal Reissner--Nordstr\"om spacetimes of dimension $n+1\geq4$. 

Given a spacetime $(\R\times \mathcal{I}\times\mathbb{S}^{n-1},-f(r)dt^2+f(r)^{-1}dr^2+r^2\Omega)\in\mathcal{S}$, we will use the auxiliary function 
\begin{align}\label{eq:veff}
v_{\text{eff}}^f\colon \mathcal{I}\to \R^+\colon r\mapsto \frac{f(r)}{r^2}
\end{align}
to simplify the photon sphere condition~\eqref{eq:photonsphere} as well as the photon surface ODE~\eqref{eq:ODE3} to
\begin{align}\label{eq:photonspherev}
(v_{\text{eff}}^{f})'(r_{*})&=0,\\\label{eq:photonsurfv}
\left(\frac{dr}{dt}\right)^{2} &= \frac{f(r)^{2}}{\a^{2}}\left(\a^2  - v_{\text{eff}}^{f}(r)\right),
\end{align}
respectively. The index 'eff' in $v_{\text{eff}}^{f}$ stands for \emph{effective} and $v^f_\text{eff}$ is related to the effective potential in the analysis of null geodesics (see \Cref{subsec:generating}). Hence, by \eqref{eq:photonspherev}, we see that spherically symmetric photon spheres correspond exactly to critical points of $v_{\text{eff}}^{f}$:

\begin{Prop}[Photon sphere revisited]
Let $(\R\times \mathcal{I}\times\mathbb{S}^{n-1},-f(r)dt^2+f(r)^{-1}dr^2+r^2\Omega)\in\mathcal{S}$ and let $v_{\text{eff}}^f$ be as in \eqref{eq:veff}. Then a surface $P^{n}=\{r=r_{*}\}\hookrightarrow(\R\times \mathcal{I}\times\mathbb{S}^{n-1},-f(r)dt^2+f(r)^{-1}dr^2+r^2\Omega)$ is a photon sphere if and only if $r_{*}$ is a critical point of $v_{\text{eff}}^f$, i.e., if and only if $(v_{\text{eff}}^{f})'(r_{*})=0$.
\end{Prop}

As an example, let us apply this proposition to the Reissner--Nordstr\"om spacetimes.
\begin{Cor}[Photon spheres in Reissner--Nordstr\"om spacetimes]\label{cor:photonsphere in rn}\label{coro:SSphoto}
Let $n+1\geq4$. Consider the $(n+1)$-dimensional Reissner--Nordstr\"om spacetime of mass $m$ and charge $q$. If $m\leq0$, there is no photon sphere. Now assume $m>0$. If $m\geq\vert q\vert$, the (exterior) Reissner--Nordstr\"om spacetime has a unique photon sphere which lies at 
\begin{align}\label{eq:r*}
r_{\ast} &= \left(\frac{nm + \sqrt{n^2m^2 - 4(n-1)q^2}}{2}\right)^{\frac{1}{n-2}},
\end{align}
with $r_{*}=(nm)^{\frac{1}{n-2}}$ for $q=0$ (and the interior subextremal Reissner--Bordstr\"om spacetimes have no photon spheres). If $\vert q\vert>m>\frac{2\sqrt{n-1}}{n}\vert q\vert$, the Reissner--Nordstr\"om spacetime has precisely two photon spheres lying at 
\begin{align}\label{eq:r*super}
r_{\ast,\pm} &= \left(\frac{n m \pm \sqrt{n^2m^2 - 4(n-1)q^2}}{2}\right)^{\frac{1}{n-2}},
\end{align}
while it has precisely one at $r_{*,-}=r_{*,+}=r_{*}$ from \eqref{eq:r*super} for $m=\frac{2\sqrt{n-1}}{n}\vert q\vert$, and none if $0<m<\frac{2\sqrt{n-1}}{n}\vert q\vert$. In particular, in the Schwarzschild case $q=0$, there is a unique photon sphere at $r_\ast=(nm)^{\frac{1}{n-2}}$ for $m>0$ and none for $m\leq0$.
\end{Cor}
\begin{proof}
For $r>r_{m,q}$ as given in \Cref{sec:RN}, the metric coefficient $f_{m,q}=N_{m,q}^{2}$ (see \eqref{rnlapse}) is given by
\begin{align*}
f_{m,q}(r)&=1-\frac{2m}{r^{n-2}}+\frac{q^2}{r^{2(n-2)}}=\frac{a(r^{n-2})}{r^{2(n-2)}},
\end{align*}
for the quadratic polynomial $a(x)\definedas x^2-2mx+q^2$, where $x>m+\sqrt{m^2-q^2}>0$ or $0<x<m-\sqrt{m^{2}-q^{2}}$ if $m> \vert q\vert$, $x>m$ if $m=\vert q\vert$, and $x>0$ otherwise. By \Cref{thm:sphsymm} and the equivalent formulation \eqref{eq:photonspherev} of the photon sphere condition, finding all photon spheres is equivalent to finding all critical points of $v(r)\definedas v_{\text{eff}}^{f_{m,q}}(r)$. We find
\begin{align*}
v'(r)&= -\frac{2b(r^{n-2})}{r^{2n-1}},
\end{align*}
where $b$ is the quadratic polynomial $b(x)\definedas x^{2}-nmx+(n-1)q^{2}$. The polynomial $b$ has the real zeros 
\begin{align}\label{eq:xpm}
x_{\pm}&=\frac{nm\pm\sqrt{n^{2}m^{2}-4(n-1)q^{2}}}{2}
\end{align}
(if any). These solutions must both be non-positive, $x_{\pm}\leq0$ (if real) if $m\leq0$ and hence do not lead to photon spheres. So let $m>0$. If $m> \vert q\vert$, the radicand of \eqref{eq:xpm} is positive and one finds $x_+>m+\sqrt{m^2-q^2}> x_{-}>m-\sqrt{m^2-q^2}$ and hence there is indeed precisely one photon sphere in the exterior region $r>r_{m,q}$ at $r_*=x_+^\frac{1}{n-2}$ as in \eqref{eq:r*} and none in the interior region $0<r<\overline{r}_{m,q}$. Similarly, if $m=\vert q\vert\neq0$, one finds $x_{+}>m>x_{-}$ and hence precisely one photon sphere at $r_{*}=x_{+}^{n-2}$ as in \eqref{eq:r*}.

If $\vert q\vert>m>\frac{2\sqrt{n-1}}{n}\vert q\vert$, the radicand of \eqref{eq:xpm} is non-negative (zero iff $m=\frac{2\sqrt{n-1}}{n}\vert q\vert$) and both $x_{\pm}>0$ which proves existence of precisely two photon spheres at $r_{*,\pm}=x_{\pm}^{\frac{1}{n-2}}$ which coincide when $m=\frac{2\sqrt{n-1}}{n}\vert q\vert$. Finally, the condition $0<m<\frac{2\sqrt{n-1}}{n}\vert q\vert$ leads to a negative radicand and hence no critical points of $v$ and no photon spheres in the corresponding Reissner--Nordstr\"om spacetimes.
\end{proof}

Let us now introduce the class $\mathcal{S}_\text{ext}$ in which we will give a complete characterization of spherically symmetric photon surfaces. 
\begin{Def}[Class $\mathcal{S}_\text{ext}$]\label{def:Sext}
Let $(\R\times \slice,\mathfrak{g})\in\mathcal{S}$ with $\slice=\mathcal{I}\times\mathbb{S}^{n-1}$ and smooth, positive metric coefficient $f\colon \slice\to\R^+$ such that $\mathfrak{g}=-f(r)dt^2+f(r)^{-1}dr^2+r^2\Omega$. Set
\begin{align*}
v_{\text{eff}}^f\colon \mathcal{I}\to \R^+\colon r\mapsto \frac{f(r)}{r^2}.
\end{align*}
Then $(\R\times \slice,\mathfrak{g})\in\mathcal{S}_\text{ext}$ if $\mathcal{I}=(r_H,\infty)$ for some $r_H>0$ and if, in addition, 
\begin{enumerate}\itemsep0em
\item $f$ extends $C^3$ to $r_H$ with $f(r_H)=0$, $f'(r_H)>0$,
\item satisfies $f(r)\to1$ and $f'(r)\to 0$ as $r\to\infty$,
\item and $v_\text{eff}^{f}$  is strictly increasing on $(r_{H},r_{*})$ up to one global maximum at some $r_\ast \in \mathcal{I}$ and then strictly decreasing on $(r_\ast, \infty)$.
\end{enumerate}
\end{Def}

\begin{Rem}[Interpretation of class $\mathcal{S}_{\text{ext}}$]\label{rem: conditions S_ext} 
Condition 1 is equivalent to $r = r_H$ representing a Killing horizon in the spacetime under consideration, with $f'(r_H)>0$ corresponding to non-degeneracy (i.e., non-vanishing surface gravity) of this Killing horizon. Condition 2 is a weak asymptotic flatness condition on the spacetime under consideration. Condition 3 is a condition imposed only to simplify the exposition in this section, see also \Cref{rem: generalization PS analysis}. We also note that these conditions on $f$ in particular imply that $f$ is Lipschitz continuous on $[r_H, \infty)$.
\end{Rem}

\begin{Rem}[Examples of spacetimes in class $\mathcal{S}_{\text{ext}}$]\label{rem:examplesSext}
A straightforward computation shows that class $\mathcal{S}_{\text{ext}}$ contains the (exterior) subextremal Reissner--Nordstr\"om spacetimes of dimension $n+1\geq4$ and hence in particular the positive mass Schwarzschild spacetimes. Moreover, it contains the restriction of the $3+1$-dimensional (exterior) positive mass Schwarzschild spacetime to its equatorial plane considered by Foertsch, Hasse, and Perlick~\cite{Perlicketal}. Our existence and uniqueness results pertaining to spherically symmetric photon surfaces in spacetimes of class $\mathcal{S}_\text{ext}$ hence apply to all of these examples. 
\end{Rem}

As we cannot solve the photon surface ODE \eqref{eq:ODE3} explicitly for generic spacetimes in class $\mathcal{S}_{\text{ext}}$ (nor in the concrete examples mentioned in \Cref{rem:examplesSext}), we will instead discuss existence and uniqueness of solutions of \eqref{eq:ODE3} abstractly and will analyze some qualitative properties of the solutions of \eqref{eq:ODE3}, using only Conditions 1--3. For proving existence and uniqueness of solutions of \eqref{eq:ODE3}, we are going to make use of the following proposition which is based on the global Picard--Lindel\"of theorem (see e.g. \cite[p.\,55]{GDG}).

\begin{Prop}[Barrier principle]\label{prop:Psi-Fct.}
Let $\Psi\colon\R\to \R$ be a Lipschitz continuous function which vanishes on an interval $(-\infty,R]$ for some number $R \in \R$, $\Psi\vert_{(-\infty, R]} \equiv 0$. Then, for each $t_0, r_{0} \in \R$, there is a unique global solution $r_{t_0,r_{0}}\colon\R\to\R$ of the initial value problem
\begin{align}\label{eq:IVP}
\dot{r}(t) &=\Psi\left(r(t)\right), \quad r(t_0)=r_0.
\end{align}
This solution $r_{t_0,r_{0}}$ has $C^{1}$-regularity, continuously depends on $t_{0}$ and $r_{0}$, and satisfies the following \emph{barrier principle:} If $r_0>R$ then $r_{t_0,r_{0}}(t)>R$ for all $t \in \R$.
\end{Prop}

\begin{proof}
Since $\Psi$ is independent of the parameter $t$, it is naturally Lipschitz continuous w.r.t.~$r$ uniformly in $t$, and we can apply the global version of the Picard--Lindel\"of theorem~\cite[p.~55]{GDG} which yields global existence and uniqueness of a solution $r_{t_{0},r_{0}}$ of \eqref{eq:IVP} that continuously depends on $t_0$ and $r_{0}$ and has $C^{1}$-regularity in $t$. Now let $r_0 >R$, and assume that $r_{t_{0},r_0}(t_1)\leq R$ for some $t_1 \in \R$. Then by the assumption $\Psi\vert_{(-\infty, R]} \equiv 0$, we find
\begin{align*}
\dot{r}_{t_{0},r_0}(t_1)=\Psi\left(r_{t_{0},r_0}(t_1)\right)=0.
\end{align*}
But then both $r_{t_{0},r_{0}}$ and the constant function $r_{1}(t)\definedas r_{t_{0},r_0}(t_1)$ globally solve the initial value problem
\begin{align}
\dot{r}(t) &=\Psi\left(r(t)\right), \quad r(t_1)=r_{t_{0},r_0}(t_1)
\end{align}
but do not coincide at $t=t_{0}$, where $r_{t_{0},r_{0}}(t_{0})=r_{0}>R$ but $r_{1}(t_{0})=r_{t_{0},r_{0}}(t_{1})\leq R$. This contradicts the uniqueness of global solutions asserted by the global Picard--Lindel\"of theorem. Hence $r_{r_0, \lambda}(t)>R$ for all $t \in \R$.
\end{proof}

\subsection{Existence, Uniqueness, and Characterization in $\mathcal{S}_{\text{ext}}$}
Let us now consider the square root of the photon surface ODE \eqref{eq:ODE3}, that is
\begin{align}\label{eq:ODE3sqrt}
\dot{r}&=\pm f(r)\sqrt{1-\frac{v_{\text{eff}}^{f}(r)}{\lambda^2}}.
\end{align}
By \Cref{thm:sphsymm}, any smooth radial function $r\colon I \to \mathcal{I}$ with $I \subset \R$ open that is a solution to~\eqref{eq:ODE3sqrt} gives rise to a photon surface in the corresponding spacetime $(\R\times M^n, \mathfrak{g})$ with umbilicity factor $\lambda$.
As already stated, any time-reflection or time-translation of $r=r(t)$ will again be a solution and thus give rise to another radial profile of a photon surface. 
A radial function $r=r(t)$ solving~\eqref{eq:ODE3sqrt} has a turning point, meaning a minimum or maximum, only if its derivative vanishes, i.e., if $\pm\dot{r}(t_{0}) = 0$ at some instant of time $t_0 \in I$. This happens if either $f(r(t_0) )\asdefined f(r_0)=0$, or if 
\begin{align}\label{eq:lambda^2}
\lambda^2= v_{\text{eff}\,}^f(r_0)=\frac{f(r_0)}{r_0^2}.
\end{align}
Recall that $f$ is positive on $\mathcal{I}=(r_H, \infty)$ and smoothly extends to $r_H$ with $f(r_H)=0$, so the first case only happens for $r_0=r_H$ (which we exclude by construction).
In order to consider the other case, namely \eqref{eq:lambda^2}, let us look at the function $v_{\text{eff}}^f$, depicted in \Cref{fig:y(r)}. 
By assumption, $v_{\text{eff}}^f$ vanishes at $r_H$, is strictly increasing for $r_H<r<r_\ast$, has a maximum at $r_\ast$ and no further extrema, and is strictly decreasing for $r>r_\ast$. Furthermore, for $r\rightarrow \infty$ it has the asymptote $0$.

\begin{figure}[h!]
\centering
\includegraphics[scale=0.9]{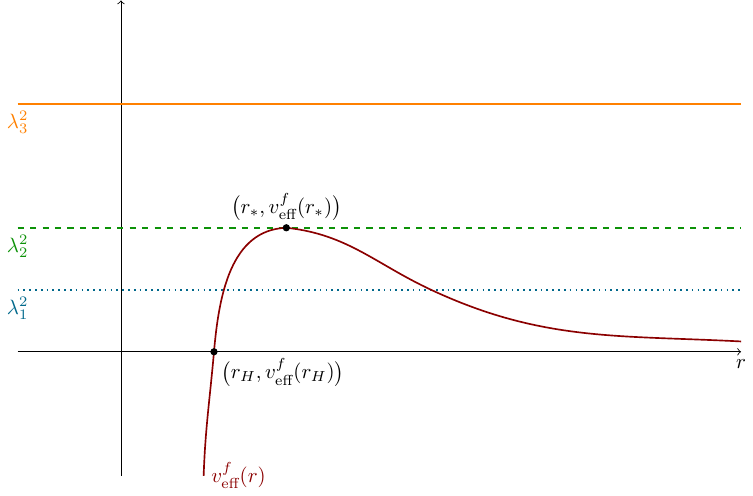}
\caption{Sketch of the function $v_{\text{eff}}^f(r)=\frac{f(r)}{r^2}$ with three different cases for $\lambda$.}
\label{fig:y(r)}
\end{figure}

This means that for $0<\lambda^2<v_{\text{eff}}^{f}(r_{*})\eqdef \lambda_\ast^2$, symbolized in \Cref{fig:y(r)} by \textcolor{greyblue}{$\lambda_1^2$}, there are always two solutions of \eqref{eq:lambda^2} in $\mathcal{I}$, i.e., two intersections with $v_{\text{eff}}^f$. Therefore any such solution $r=r(t)$ can have at most two possible turning points. For $\lambda^2=\lambda_\ast^2$, denoted by \textcolor{gruen}{$\lambda_2^2$} in \Cref{fig:y(r)}, the only possible turning point is $r_\ast$, and for $\lambda^2>\lambda_\ast^2$, symbolized by \textcolor{orange}{$\lambda_3^2$}, there is no intersection and the radial profile never turns. We are going to treat these three cases separately.\\

\newpage

\noindent \textcolor{greyblue}{\underline{\textbf{Case 1:} $0<\lambda^2<\lambda_\ast^2$.}}
\begin{thm}[Existence and uniqueness in Case 1]\label{prop:case1}
Let $f$ be the metric coefficient of a spacetime in class $\mathcal{S}_{\text{ext}}$ and let $0<\lambda^2<\lambda_\ast^2\definedas v_{\text{eff}}^{f}(r_{*})$. Let $r_H<r_{\lambda}<r_\ast<R_{\lambda}<\infty$ be the unique radii with $v_{\text{eff}\,}^{f}(r_{\lambda})=v_{\text{eff}\,}^{f}(R_{\lambda})=\lambda^{2}$. Then there exist unique smooth, time-symmetric global solutions $r_{r_{\lambda}}$ and $r_{R_{\lambda}}$ of the photon surface ODE \eqref{eq:ODE3} with this $\lambda^{2}$ and $r(0)=r_\lambda$ and $r(0)=R_{\lambda}$, respectively. Moreover, for any initial value $r_H<r_0< r_{\lambda}$ or $R_{\lambda}< r_0 <\infty$, there exist precisely two smooth, global solutions of the photon surface ODE \eqref{eq:ODE3} with this $\lambda^{2}$ and $r(0)=r_0$; these are time-reflections of each other and time-translations and possibly time-reflections of $r_{r_{\lambda}}$ or $r_{R_{\lambda}}$, respectively. Solving the initial value problem for \eqref{eq:ODE3} at generic $t_{0}$ rather than at $t_{0}=0$ results only in further time-translations. All these solutions behave as depicted in \Cref{fig:type1}. 
\end{thm}

\begin{figure}[h!]
\centering
\includegraphics[scale=1.2]{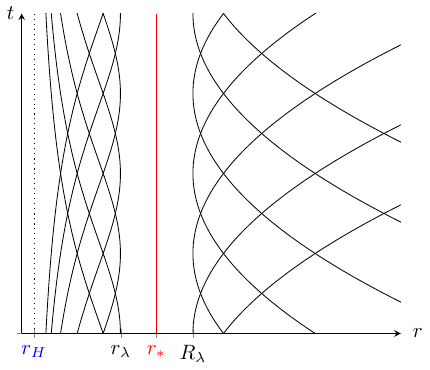}
\caption{\text{Photon surfaces of \textcolor{greyblue}{Case 1}}, $0<\lambda^2<\lambda_\ast^2$.}
\label{fig:type1}
\end{figure}

\begin{proof}
Note that the radicand on the right hand side of \eqref{eq:ODE3sqrt} is non-negative iff $\lambda^2\geq v_{\text{eff}}^{f}(r)$, and vanishes only at $r_H$ and for $r\in(r_{H},\infty)$ satisfying $\lambda^2=v_{\text{eff}}^{f}(r)$. Combined with our restriction on $\lambda$ this implies that there exist radii $r_{H}<r_{\lambda} <r_\ast<R_{\lambda}<\infty$ with $v_{\text{eff}\,}^{f}(r_{\lambda})=v_{\text{eff}\,}^{f}(R_{\lambda})=\lambda^{2}$ such that the radicand is positive on $(r_H, r_{\lambda}) \,\cup\, (R_{\lambda},\infty)$ and negative on $(r_{\lambda},R_{\lambda})$, see \textcolor{greyblue}{$\lambda_{1}^{2}$} in \Cref{fig:y(r)}. Hence we only consider~\eqref{eq:ODE3sqrt} for initial values $r_{0}$ satisfying $r_H<r_0\leq r_{\lambda}$ or $R_{\lambda}\leq r_0 <\infty$. However, the right hand side of \eqref{eq:ODE3sqrt} is not Lipschitz continuous at $r=r_{\lambda}$ and $r=R_{\lambda}$ since its derivative with respect to $r$ blows up there. In order to handle this, we are going to apply \Cref{prop: appendix} after suitably rescaling and translating the variables. We will treat the '$+$' and '$-$' cases of~\eqref{eq:ODE3sqrt} separately and first focus on $r_{0}=R_{\lambda}$ and $r_{0}=r_{\lambda}$ before moving on to $r_{0}>R_{\lambda}$ and $r_{H}<r_{0}<r_{\lambda}$.

Let us first consider the '$+$' case of~\eqref{eq:ODE3sqrt} for $r_{0}=R_{\lambda}$. For $r\geq R_{\lambda}$, we set 
\begin{align*}
w \definedas r-R_{\lambda}, \quad \widehat{v}(w)\definedas v_{\text{eff}}^f(w+R_{\lambda}),\quad \text{and} \quad \widehat{f}(w) \definedas \frac{f(w+R_{\lambda})}{\lambda}.
\end{align*}
In these variables, we can rewrite the '$+$' case of~\eqref{eq:ODE3sqrt} as
\begin{align}\label{eq: ODE w+}
\dot{w}=\widehat{f}(w)\sqrt{\lambda^2-\widehat{v}(w)}\asdefined \sqrt{F(w)}.
\end{align}
Using the definition of Class $\mathcal{S}_{\text{ext}}$, it can easily be checked from its definition and the assumptions of \Cref{prop:case1} that this function $F\colon[0,\infty)\to\R^{+}_{0}$ satisfies the assumptions of \Cref{prop: appendix}, hence we get a unique non-trivial global $C^{1}$-solution $w\colon\R\to[0,\infty)$ of \eqref{eq: ODE w+} with initial value $w(0)=0$ which furthermore satisfies $w(t)=0$ for all $t\leq0$. Transforming back to the original variable $r=w+R_{\lambda}$, this yields a unique non-constant global $C^{1}$-solution $r^{+}_{R_{\lambda}}\colon\R\to[R_{\lambda},\infty)$ of the initial value problem
\begin{align}\label{eq:IVP Psi case1, 1}
\dot{r}(t)&=\left\{
\begin{array}{ll}
f(r(t))\sqrt{1-\frac{v^f_{\text{eff}}(r(t))}{\lambda^{2}}} & \text{for }r(t) \in (R_{\lambda},\infty), \\
0 & \, \textrm{otherwise,} \\
\end{array}
\right.\\
r(0)&=R_{\lambda}
\end{align}
satisfying the '$+$' case of the photon surface ODE~\eqref{eq:ODE3sqrt} on $(R_{\lambda}, \infty)$. Moreover, we know that $r^{+}_{R_{\lambda}}(t)=R_{\lambda}$ for $t\leq0$ and $r^{+}_{R_{\lambda}}(t)>R_{\lambda}$ for $t>0$ from \Cref{prop: appendix} and \Cref{prop:Psi-Fct.}, respectively. Finally, $r^{+}_{R_{\lambda}}$ is strictly increasing on $(0,\infty)$ by \eqref{eq:IVP Psi case1, 1}. One can see as follows that in fact $r^{+}_{R_{\lambda}}(t)\to\infty$ as $t\to\infty$: Suppose towards a contradiction that $r^{+}_{R_{\lambda}}$ remains bounded above by some constant $C>0$, $r^{+}_{R_{\lambda}}(t)\leq C$ for all $t\in\R$. Then as $t\to\infty$, we find $r^{+}_{R_{\lambda}}(t)\to r^{+}$ for some $r^{+}>R_{0}$ as $t\to\infty$ and
\begin{align*}
\dot{r}^{+}_{R_{\lambda}}(t)\to f(r^{+})\sqrt{1-\frac{v^f_{\text{eff}}(r^{+})}{\lambda^2}}>0
\quad \text{ for }\quad t\to \infty,
\end{align*}
a contradiction to $r^{+}_{R_{\lambda}}(t)\to r^{+}$ as $t\to\infty$. In particular, we find $\dot{r}^{+}_{R_{\lambda}}(t)\to1$ as $t\to\infty$ so that the radial profile $r^{+}_{R_{\lambda}}$ tends to the hyperbola with asymptote $r=1$, i.e., to the radial profile of the unit radius hyperboloid in the Minkowski spacetime.

Similarly, still in the '$+$' case of~\eqref{eq:ODE3sqrt} but switching to $r_{0}=r_{\lambda}$, for $r_{H}<r\leq r_{\lambda}$, we set 
\begin{align*}
\widetilde{w} \definedas -(r-r_{\lambda})\in[0,r_{\lambda}-r_{H}),\quad\widetilde{v}(\widetilde{w})\definedas v_{\text{eff}}^f(-\widetilde{w}+r_{\lambda}),\quad \text{and} \quad \widetilde{f}(\widetilde{w}) \definedas \frac{f(-\widetilde{w}+r_{\lambda})}{\lambda}.
\end{align*}
Now extend $\widetilde{v}$ and $\widetilde{f}$ trivially by $0$ to $[r_{\lambda}-r_{H},\infty)$ in view of \Cref{rem:regularityappendix}. This allows us to rewrite the photon surface ODE~\eqref{eq:ODE3} as
\begin{align}\label{eq: ODE w}
\dot{\widetilde{w}}=\widetilde{f}(\widetilde{w})\sqrt{\lambda^2-\widetilde{v}(\widetilde{w})}\asdefined\sqrt{\widetilde{F}(\widetilde{w})}
\end{align}
with $\widetilde{F}\colon[0,\infty)\to\R^{+}_{0}$ satisfying the assumptions of \Cref{prop: appendix} (see \Cref{rem:regularityappendix}). Applying \Cref{prop: appendix} and \Cref{prop:Psi-Fct.} and transforming back to the original variable $r=-\widetilde{w}+r_{\lambda}$ yields a unique non-constant global $C^{1}$-solution $r^{+}_{r_{\lambda}}\colon\R\to(r_{H},r_{\lambda}]$ of
\begin{align}\label{eq:IVP Psi case1, 2}
\dot{r}(t)&= \left\{
\begin{array}{ll}
f(r(t))\sqrt{1-\frac{v^f_{\text{eff}}(r(t))}{\lambda^{2}}} & \text{for }r(t) \in (r_H, r_{\lambda}), \\
0 & \, \textrm{otherwise,} \\
\end{array}
\right.\\
r(0)&=r_{\lambda}
\end{align}
satisfying the '$+$' case of the photon surface ODE~\eqref{eq:ODE3sqrt} on $(r_{H},r_{\lambda})$. Moreover, we know that $r^{+}_{r_{\lambda}}(t)=r_{\lambda}$ for $t\geq0$ and $r_{H}<r^{+}_{r_{\lambda}}(t)<r_{\lambda}$ for $t<0$ from \Cref{prop: appendix} and \Cref{prop:Psi-Fct.}, respectively. Again, $r^{+}_{r_{\lambda}}$ is strictly increasing on $(-\infty,0)$ by \eqref{eq:IVP Psi case1, 2}. Finally, $r^{+}_{r_{\lambda}}$ tends to $r_{H}$ as $t\to-\infty$: Suppose towards a contradiction that there is $C>r_{H}$ such that $r^{+}_{r_{\lambda}}(t)\geq C$ for all $t\in\R$. Then as $t\to-\infty$, there is $r^{+}>r_{H}$ such that $r^{+}_{r_{\lambda}}(t)\to r^{+}$ as $t\to-\infty$ and
\begin{align*}
\dot{r}^{+}_{r_{\lambda}}(t)\to f(r^{+})\sqrt{1-\frac{v^f_{\text{eff}}(r^{+})}{\lambda^2}}>0
\quad \text{ for }\quad t\to-\infty,
\end{align*}
a contradiction to $r^{+}_{r_{\lambda}}(t)\to r^{+}$ as $t\to-\infty$. Hence $r^{+}_{r_{\lambda}}$ tends to $r=r_{H}$ as $t\to-\infty$.

In the '$-$' case of~\eqref{eq:ODE3sqrt}, the exact same arguments provide unique global non-constant $C^{1}$-solutions  $r^{-}_{R_{\lambda}}\colon\R\to[R_{\lambda},\infty)$ and $r^{-}_{r_{\lambda}}\colon\R\to(r_{H},r_{\lambda})$ of
\begin{align*}
\dot{r}(t)&=\left\{
\begin{array}{ll}
-f(r(t))\sqrt{1-\frac{v^f_{\text{eff}}(r(t))}{\lambda^{2}}} & \text{for }r(t) \in (R_{\lambda},\infty), \\
0 & \, \textrm{otherwise,} \\
\end{array}
\right.\\
r(0)&=R_{\lambda}
\end{align*}
and
\begin{align*}
\dot{r}(t)&=\left\{
\begin{array}{ll}
- f(r(t))\sqrt{1-\frac{v^f_{\text{eff}}(r(t))}{\lambda^{2}}} & \text{for }r(t) \in (r_{H},r_{\lambda}), \\
0 & \, \textrm{otherwise,} \\
\end{array}
\right.\\
r(0)&=r_{\lambda},
\end{align*}
respectively, using \Cref{cor: appendix} instead of \Cref{prop: appendix}. The solutions $r^{-}_{R_{\lambda}}$ and $r^{-}_{r_{\lambda}}$ are naturally related to the previously constructed solutions by time-reflection,
\begin{align*}
r^{-}_{R_{\lambda}}(t)&=r^{+}_{R_{\lambda}}(-t),\\
r^{-}_{r_{\lambda}}(t)&=r^{+}_{r_{\lambda}}(-t)
\end{align*}
for all $t\in\R$, as can be seen from the proof of \Cref{cor: appendix}. We will now argue that $r^{+}_{R_{\lambda}}\vert_{[0,\infty)}$ and $r^{-}_{R_{\lambda}}\vert_{(-\infty,0]}$ as well as $r^{+}_{r_{\lambda}}\vert_{(-\infty,0]}$ and $r^{-}_{r_{\lambda}}\vert_{[0,\infty)}$ can be glued together, respectively, to smooth, global, time-symmetric solutions $r_{R_{\lambda}}$, $r_{r_{\lambda}}$ of the photon surface ODE~\eqref{eq:ODE3}. First of all, by smoothness of $f$ and hence $v_{\text{eff}}^{f}$ on $(r_{H},\infty)$, we know that the right hand side of \eqref{eq:ODE3sqrt} is a smooth function of $r$ on $(r_{H},r_{\lambda})\cup(R_{\lambda},\infty)$. Hence the glued solutions $r_{R_{\lambda}}\vert_{\R\setminus\{0\}}$ and $r_{r_{\lambda}}\vert_{\R\setminus\{0\}}$ are smooth. It remains to consider what happens at $t=0$. Clearly we have that $r_{R_{\lambda}}$, $r_{r_{\lambda}}$ are $C^{1}$ across $t=0$. Formally computing the time derivative of \eqref{eq:ODE3sqrt}, suppressing the dependence on $t$, we get
\begin{align*}
\ddot{r}&=f'(r)f(r)\left(1- \frac{v^f_{\text{eff}}(r)}{\lambda^2}\right)-\frac{\left(f(r)\right)^{2}{v^{f}_{\text{eff}}}'(r)}{2\lambda^2},
\end{align*}
in both the '$+$' and the '$-$' cases, where $'$ denotes the derivative with respect to $r$. Hence $\lim_{t\searrow0}\ddot{r}_{R_{\lambda}}(t)=\lim_{t\nearrow0}\ddot{r}_{R_{\lambda}}(t)$ and $r_{R_{\lambda}}$ is indeed $C^{2}$ at $t=0$; the same argument applies to $r_{r_{\lambda}}$. Continuing this argument by induction and recalling the asserted parity properties, we see that $r_{R_{\lambda}}$ and $r_{r_{\lambda}}$ are indeed the unique smooth solutions of the initial value problems for the photon surface ODE~\eqref{eq:ODE3} with the initial conditions $r(0)=R_{\lambda}$ and $r(0)=r_{\lambda}$, respectively. This asserts the claim of \Cref{prop:case1} for $r_{0}=R_{\lambda}$ and $r_{0}=r_{\lambda}$.

Let's now turn to $r_{0}>R_{\lambda}$. By the monotonicity and asymptotic properties of $r^{+}_{R_{\lambda}}$ derived above, we know that there exists a unique time $t_{0}\in\R$ such that $r^{+}_{R_{\lambda}}(t_{0})=r_{0}$. Now set 
\begin{align*}
r^{\pm}_{r_{0}}(t)&\definedas r_{R_{\lambda}}(\pm\left(t+t_{0}\right))
\end{align*}
for $t\in\R$. By time-translation and time-reflection invariance of the photon surface ODE~\eqref{eq:ODE3} and time-reflection symmetry of the solution $r_{R_{\lambda}}$, this gives the desired smooth, global solutions $r^{\pm}_{r_{0}}$ of \eqref{eq:ODE3} with $r^{\pm}_{r_{0}}(0)=r_{0}$ and positive respectively negative slopes $\dot{r}^{+}_{r_{0}}(0)>0$, $\dot{r}^{-}_{r_{0}}(0)<0$. Now let $r$ be another global solution of \eqref{eq:ODE3} with $r(0)=r_{0}$. Then by \eqref{eq:ODE3}, either $\dot{r}(0)>0$ or $\dot{r}(0)<0$, more specifically either $\dot{r}(0)=\dot{r}^{+}_{r_{0}}(0)=\dot{r}_{R_{\lambda}}(t_{0})>0$ or $\dot{r}(0)=\dot{r}^{-}_{r_{0}}(0)=-\dot{r}_{R_{\lambda}}(-t_{0})<0$. Consequently, (a suitable restriction of) $r$ solves either the '$+$' or the '$-$' case of \eqref{eq:ODE3sqrt} with $r(0)=r_{0}$. Hence by local uniqueness in the (local) Picard--Lindel\"of theorem, we know that $r$ coincides with either $r^{+}_{r_{0}}$ or $r^{-}_{r_{0}}$ near $t=0$, in fact as long as either $\dot{r}>0$ or $\dot{r}<0$, respectively. Applying \Cref{prop: appendix} and \Cref{cor: appendix} as above and noting that $r^{+}_{R_{\lambda}}$ is not $C^{2}$ at $t=0$ by the above, we conclude that in fact the smooth, global solution $r$ must coincide with either $r^{+}_{r_{0}}$ or $r^{-}_{r_{0}}$ on all of $\R$ as claimed.

The same philosophy allows to conclude the claims for initial values $r_{H}<r_{0}<r_{\lambda}$.
\end{proof}

\begin{Rem}[Precise meaning of uniqueness]\label{rem:uniqueness}
As can be seen from the proof of \Cref{prop:case1}, there are other global $C^{1}$-solutions of the photon surface ODE  \eqref{eq:ODE3}, namely the global $C^{1}$-functions $r^{\pm}_{r_{\lambda}}$ and $r^{\pm}_{R_{\lambda}}$ constructed in the proof which coincide with $r=r_{\lambda}$ and $r=R_{\lambda}$ on real half-lines and are not $C^{2}$ at isolated points. In fact, one can produce infinitely many global $C^{1}$-solutions of \eqref{eq:ODE3} by gluing two such solutions (one '$+$' and one '$-$', same index $r_{\lambda}$ or $R_{\lambda}$) together after a finite time in which they are constantly equal to $r_{\lambda}$ or $R_{\lambda}$, respectively. In particular, the uniqueness assertion in \Cref{prop:case1} only applies in comparison with $C^{2}$-solutions. 

On the other hand, the uniqueness analysis in the proof of \Cref{prop:case1} is completely local, hence we effectively obtain local uniqueness of local $C^{2}$-solutions. Our analysis indeed even gives local uniqueness of local $C^{1}$-solutions for initial values $r_{0}\neq r_{\lambda}, R_{\lambda}$, with the uniqueness claim only extending to time intervals preventing the solutions from reaching $r=r_{\lambda}$ or $r=R_{\lambda}$.
\end{Rem}
\vspace{1ex}
\noindent\textcolor{gruen}{\underline{\textbf{Case 2:} $\lambda^2=\lambda_\ast^2$.}}

\begin{thm}[Existence and uniqueness in Case 2]\label{prop:case2}
Let $f$ be the metric coefficient of a spacetime in class $\mathcal{S}_{\text{ext}}$ and let $\lambda^2=\lambda_\ast^2\definedas v_{\text{eff}}^{f}(r_{*})$. The photon sphere $r=r_{*}$ is the unique smooth solution of the photon surface ODE \eqref{eq:ODE3} with $\lambda^{2}=\lambda_{*}^{2}$ and $r(0)=r_0=r_{*}$. Moreover, there exist two smooth, strictly increasing, global solutions $r^{+}_{*}$ and $R^{+}_{*}$ of the photon surface ODE \eqref{eq:ODE3} with $\lambda^{2}=\lambda_{*}^{2}$ with respective ranges $r^{+}_{*}(\R)=(r_{H},r_{*})$ and $R^{+}_{*}(\R)=(r_{*},\infty)$ such that the following holds: For any initial value $r_H<r_0<r_{*}$, there  exist precisely two smooth, global solutions $r^{\pm}_{r_{0}}$ of the photon surface ODE \eqref{eq:ODE3} with $\lambda^{2}=\lambda_{*}^{2}$ and initial value $r(0)=r_0$; these are time-reflections of each other and $r^+_{r_{0}}$ is a time-translation of $r^{+}_{*}$. For any initial value $r_{0}>r_{*}$, there exist precisely two smooth, global solutions $R^{\pm}_{r_{0}}$ of the photon surface ODE \eqref{eq:ODE3} with $\lambda^{2}=\lambda_{*}^{2}$ and initial value $r(0)=r_0$; these are time-reflections of each other and $R^+_{r_{0}}$ is a time-translation of $R^{+}_{*}$. Solving the initial value problem for \eqref{eq:ODE3} at generic $t_{0}$ rather than at $t_{0}=0$ results only in further time-translations.

Here, uniqueness is meant locally in comparison with other local $C^{1}$-solutions. The solutions $r^{\pm}_{r_{0}}$, $R^{\pm}_{r_{0}}$ behave as depicted in \Cref{fig:type2}, including the photon sphere for $r_0=r_\ast$.
\end{thm}

\Cref{rem:uniqueness} applies correspondingly to the uniqueness claims made in \Cref{prop:case2}.

\begin{proof}
In this case, by our assumptions on $v_{\text{eff}}^{f}$, \eqref{eq:ODE3sqrt} specializes to 
\begin{align}\label{eq:ODE3-case2}
\dot{r}(t)&= \pm f(r(t))\sqrt{1-\frac{v^f_{\text{eff}}(r(t))}{\lambda_\ast^{2}}}=\pm f(r(t))\sqrt{1-\frac{v^f_{\text{eff}}(r(t))}{v^f_{\text{eff}}(r_\ast)}}
\end{align}
with positive radicand for $r(t)\in(r_{H},r_{*})\cup(r_{*},\infty)$ and vanishing radicand where $r(t)=r_{*}$. 

It turns out that the right hand side of \eqref{eq:ODE3-case2}, suitably extended by $0$ through $r=r_{H}$,
\begin{align}\label{eq:Psi case2}
\Psi_{*}^{\pm}\colon \R\to \R\colon r \mapsto \left\{
\begin{array}{ll}
\pm f(r)\sqrt{1-\frac{v^f_{\text{eff}}(r)}{\lambda_\ast^2}} & \text{for }r \in (r_H,\infty), \\
0 & \, \text{otherwise}, \\
\end{array}
\right. 
\end{align}
is in fact globally Lipschitz continuous: Clearly, $\Psi^{\pm}_{*}$ is $C^{1}$ away from $r=r_{H}$, $r=r_{*}$ with vanishing derivative as $r\to\infty$ and on $(-\infty,r_{H})$ and we only need to look at the behaviour of the derivative $(\Psi^{\pm}_{*})'$ as $r\to r_{*}$ and as $r\to r_{H}$. We find
\begin{align}
(\Psi^{\pm}_{*})'(r)&=\pm f'(r)\sqrt{1-\frac{v^f_{\text{eff}}(r)}{\lambda_\ast^2}}\mp \frac{f(r){v^{f}_{\text{eff}}}'(r)}{2\sqrt{1-\frac{v^f_{\text{eff}}(r)}{\lambda_{*}^2}}}
\end{align}
for $r\in(r_{H},r_{*})\cup(r_{*},\infty)$. The first term vanishes as $r\to r_{*}$, and we can neglect the factor $f(r)/2$ in the second term in view of an application of l'H\^{o}pital's rule as it converges to the definite value $f(r_{*})/2>0$. Since the limits as $r\searrow r_{*}$ and $r\nearrow r_{*}$ of the second term are finite if and only if the limits of its square are, we consider the latter. For ease of notation, we set $w\definedas r-r_{*}$ and $\widehat{v}(w)\definedas v_{\text{eff}}^f(w+r_{*})$. By l'H\^{o}pital's rule, we find
\begin{align*}
\lim_{r \to r_\ast}\frac{\left({v^{f}_{\text{eff}}}'(r)\right)^{2}}{1-\frac{v_{\text{eff}}^f(r)}{\lambda_\ast^2}}&=\lambda_{*}^{2}\,\lim_{w\to0}\frac{\left(\widehat{v}\,'(w)\right)^{2}}{\lambda_{*}^{2}-\widehat{v}(w)}=-\lambda_*^{2}\,\lim_{w\to 0}\frac{2\widehat{v}\,'(w)\widehat{v}\,''(w)}{\widehat{v}\,'(w)}=-2\lambda_*^{2}\, \widehat{v}\,''(0)=-2\lambda_{*}^{2}\,{v_{\text{eff}}^{f}}''(r_{*}).
\end{align*}
On the other hand, $\lim_{r\searrow r_{H}}(\Psi^{\pm}_{*})'(r)=\pm f'(r_{H})$ so that $\Psi^{\pm}_{*}$ is piecewise $C^{1}$ with globally bounded derivative and hence globally Lipschitz continuous. By  \Cref{prop:Psi-Fct.}, we hence obtain locally unique global $C^{1}$-solutions $r_{r_{0}}^{\pm}$ of the initial value problems \eqref{eq:Psi case2}, $r(0)=r_{0}$ for any $r_{0}>r_{H}$, which in particular satisfy $r^{\pm}_{r_{0}}(t)>r_{H}$ for all $t\in\R$. Both of these solutions $r^{\pm}_{r_{0}}$ clearly satisfy the photon surface ODE \eqref{eq:ODE3} and have $\dot{r}^{\pm}_{r_{0}}(t_{*})=0$ for some $t_{*}\in\R$ if and only if $r^{\pm}_{r_{0}}(t_{*})=r_{*}$.

Let us first consider the initial value $r_{0}=r_{*}$. As $r(t)=r_{*}$ is a global solution of both the '$+$' and the '$-$' cases of \eqref{eq:Psi case2}, we conclude from the above uniqueness assertion that $r^{\pm}_{r_{*}}(t)=r_{*}$ for all $t\in\R$, as claimed.

Now consider an initial value $r_H<r_0<r_\ast$ and suppose towards a contradiction that there exists $t_{*}^{\pm}\in\R$ such that $r^{\pm}_{r_{0}}(t_{*}^{\pm})=r_{*}$. Set $r(t)\definedas r^{\pm}_{r_{0}}(t+t_{*}^{\pm})$ for $t\in\R$ and observe that $r\colon\R\to(r_{H},\infty)$ is a global $C^{1}$-solution of \eqref{eq:Psi case2} by time-translation invariance, satisfying $r(0)=r_{*}$. Hence by the above uniqueness assertion, we know that $r(t)=r^{\pm}_{r_{0}}(t+t_{*}^{\pm})=r_{*}$ or in other words $r^{\pm}_{r_{0}}(t)=r_{*}$ for all $t\in\R$, a contradiction to $r^{\pm}_{r_{0}}(0)=r_{0}<r_{*}$. Hence $r^{+}_{r_{0}}$ is strictly increasing and $r^{-}_{r_{0}}$ is strictly decreasing on $\R$, they satisfy $r^{\pm}_{r_{0}}(\R)\subseteq(r_{H},r_{*})$, and are hence in fact smooth as $\Psi^{\pm}_{*}$ is smooth on $(r_{H},r_{*})$. On the other hand, $r^{+}_{r_{0}}$ and $r^{-}_{r_{0}}$ asymptote to $r=r_{*}$ as $t\to\pm\infty$, respectively, and to $r=r_{H}$ as $t\to\mp\infty$, respectively, which can be seen by arguments similar to those in Case 1. This proves that $r^{\pm}_{r_{0}}(\R)=(r_{H},r_{*})$. Moreover, both $r^{+}_{r_{0}}$ and $r^{-}_{r_{0}}$ solve the photon surface ODE \eqref{eq:ODE3} with initial value $r^{\pm}_{r_{0}}(0)=r_{0}$.

Now let $r$ be another global $C^{1}$-solution of \eqref{eq:ODE3} with $r(0)=r_{0}$. Then by \eqref{eq:ODE3}, either $\dot{r}(0)>0$ or $\dot{r}(0)<0$, more specifically either $\dot{r}(0)=\dot{r}^{+}_{r_{0}}(0)>0$ or $\dot{r}(0)=\dot{r}^{-}_{r_{0}}(0)<0$. Consequently, (a suitable restriction of) $r$ solves either the '$+$' or the '$-$' case of \eqref{eq:Psi case2} with $r(0)=r_{0}$. Hence by the above considerations, we know that $r$ globally coincides with either $r^{+}_{r_{0}}$ or $r^{-}_{r_{0}}$. 

Finally, arguing by time-translation as above and in the '$-$' case also by time-reflection, we find that in fact the functions $r^{\pm}_{r_{0}}$ for different values of $r_{0}\in(r_{H},r_{*})$ are nothing but time-translations and, in the '$-$' case, time-reflections of one and the same smooth, strictly increasing, global solution $r^{+}_{*}$ with range $(r_{H},r_{*})$ as claimed. 

The same philosophy allows to conclude the claims for initial values $r_{0}>r_*$; arguing as in Case~1 shows that the corresponding solution $R^{+}_{*}$ asymptotes to the unit radius hyperbola as $t\to\infty$ and to the photon sphere $r=r_{*}$ as $t\to-\infty$.
\end{proof}

\begin{figure}[h]
	\centering
		\includegraphics[scale=1.2]{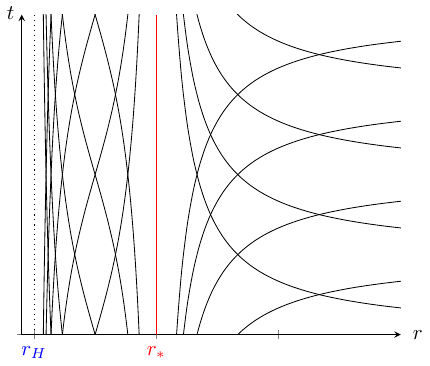}
	\caption{\text{Photon surfaces of \textcolor{gruen}{Case 2}}, $\lambda^2=\lambda_\ast^2$.}
	\label{fig:type2}
\end{figure}

\noindent \textcolor{orange}{\underline{\textbf{Case 3:} $\lambda^2>\lambda_\ast^2$.}}

\begin{thm}[Existence and uniqueness in Case 3]\label{prop:case3}
Let $f$ be the metric coefficient of a spacetime in class $\mathcal{S}_{\text{ext}}$ and let $\lambda^2>\lambda_\ast^2\definedas v_{\text{eff}}^{f}(r_{*})$. There exists a smooth, strictly increasing, global solution $r_{\lambda}$ of the photon surface ODE \eqref{eq:ODE3} with this $\lambda^{2}$ and range $r_{\lambda}(\R)=(r_{H},\infty)$ such that the following holds: For any initial value $r_H<r_0<\infty$, there exist precisely two smooth, global solutions $r^{\pm}_{r_{0}}$ of the photon surface ODE \eqref{eq:ODE3} with this $\lambda^{2}$ and initial value $r(0)=r_0$; these are time-reflections of each other and $r^+_{r_{0}}$ is a time-translation of $r_{\lambda}$. Solving the initial value problem for \eqref{eq:ODE3} at generic $t_{0}$ rather than at $t_{0}=0$ results only in further time-translations.

Here, uniqueness is meant locally in comparison with other local $C^{1}$-solutions. The solutions $r^{\pm}_{r_{0}}$ behave as depicted in \Cref{fig:type3}.
\end{thm}

\Cref{rem:uniqueness} applies correspondingly to the uniqueness claims made in \Cref{prop:case2}.

\begin{proof}
In this case, the radicand of \eqref{eq:ODE3sqrt} is always positive and we can argue analogously to the previous cases with \Cref{prop:Psi-Fct.} in order to obtain unique smooth, global, strictly increasing (in the '$+$' case of \eqref{eq:ODE3sqrt}) or strictly decreasing (in the '$-$' case of \eqref{eq:ODE3sqrt}) solutions $r^{+}_{r_{0}}$ and $r^{-}_{r_{0}}$ of the '$+$' and '$-$' cases of \eqref{eq:ODE3sqrt} for all initial values $r_0 \in (r_H, \infty)$, respectively, this time obeying the barrier principle only for $r_H$ and asymptoting to the unit radius hyperbola as $t\to\pm\infty$, respectively, and to $r=r_{H}$ as $t\to\mp\infty$, respectively. Both $r^{+}_{r_{0}}$ and $r^{-}_{r_{0}}$ are hence smooth, global solutions of the photon surface ODE \eqref{eq:ODE3} and it is straightforward to see that they are time-reflections of each other. By the same argument as before, the solutions $r^{+}_{r_{0}}$ for different values of $r_{0}\in(r_{H},\infty)$ are time-translates of each other and hence of a smooth, global, strictly increasing solution $r_{\lambda}$ of the photon surface ODE \eqref{eq:ODE3} as claimed.
\end{proof}

\begin{Rem}[More detailed asymptotic considerations]\label{rem:nearhorizon}
 Of course, the asymptotic considerations performed in the above proofs are not particularly precise. To get a more precise idea of what happens near the horizon $r=r_{H}$, the asymptotic behavior needs to be studied in horizon penetrating coordinates; to see more precisely what happens near the asymptotic zone $r\to\infty$ of the spacetime, it would be preferable to study the situation in double null coordinates. Both of these questions are addressed by the first author and Wolff~\cite{CW} in the framework of generalized Kruskal--Szekeres extensions. There, it is found that the photon surfaces with $r(t)\searrow r_H$ do actually cross the black hole horizon, while those with $r(t)\nearrow \infty$ indeed asymptote to the light cone. See also \Cref{app:B}.
\end{Rem}

\begin{figure}[h]
	\centering
		\includegraphics[scale=1.2]{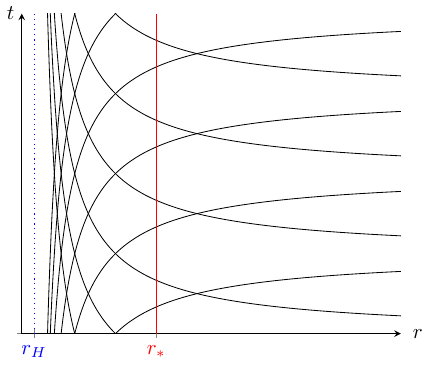}
	\caption{\text{Photon surfaces of \textcolor{orange}{Case 3}}, $\lambda^2>\lambda_\ast^2$.}
	\label{fig:type3}
\end{figure}

\begin{Rem}[Lower regularity spacetimes]
In case one wants to consider non-smooth metric coefficients $f$ as well, the solutions obtained in \Cref{prop:case1}, \Cref{prop:case2}, and \Cref{prop:case3} are of course once more continuously differentiable in all arguments than the right hand side of the photon surface ODE \eqref{eq:ODE3} which is once more continuously differentiable than $f$ is assumed to~be. 
\end{Rem}

\begin{Rem}[Alternative proofs]
Whenever the right hand side of the square root of the photon surface ODE~\eqref{eq:ODE3sqrt} does not vanish, one can likewise argue by separation of variables, namely inverting the ODE, $\frac{dt}{dr}=\pm\frac{1}{f(r)\sqrt{1-\frac{f(r)}{\lambda^2r^2}}}\asdefined \Phi(r)$, and integrating, obtaining the locally unique solutions
\begin{align}\label{eq: sol of DGL w.r.t. t(r)}
t_{r_0,t_0}(r)=\int_{r_0}^r \Phi(\tau)\,d\tau + r_0
\end{align}
with $t_{t_{0},r_{0}}(r_{0})=t_{0}$ that are continuously differentiable in $r,$ $t_0$, and $r_0$ (as often as $\Phi$ and therefore $f$ allows). Since this method is only available where the derivative $\dot{r}_{t_{0},r_{0}}$ does not vanish, the obtained map is bijective and therefore invertible, making it possible to obtain the corresponding initial value for $t_{r_0, t_0}$ from the one given for $r_{t_0, r_0}$. Moreover, since we obtain a locally unique solution on an open interval around every $r$ for which the right hand side of~\eqref{eq:ODE3sqrt} does not vanish, we can extend the unique solution to all of $I\definedas\lbrace r>r_H\,\vert\, \frac{1}{\Phi(r)}\neq 0\rbrace$. 

This perspective makes it particularly apparent that the same initial value for different times yields exactly the respective time-translated solution by an offset which can be read off of the integrability constant in~\eqref{eq: sol of DGL w.r.t. t(r)}.
\end{Rem}

Let us close this analysis by characterizing when a spherically symmetric photon surface $\photo$ in a spacetime of class $\mathcal{S}_{\text{ext}}$ is outward directed.
\begin{Prop}[Characterizing outward directedness]\label{prop:outwardspherical}
Let $(\mathfrak{L}^{n+1},\mathfrak{g})\in\mathcal{S}_{\text{ext}}$ with metric coefficient $f$, $\photo\hookrightarrow\mathfrak{L}^{n+1}$ a spherically symmetric photon surface. Then $\photo$ is non-degenerate if and only if $f'\neq0$ along $\photo$ and
outward (inward) directed if and only if $f'>(<)\,0$ along $\photo$.
\end{Prop}
\begin{proof}
Let $\gamma\colon I\to\mathfrak{L}^{n+1}$ be a radial profile of $\photo$, $\dot{\gamma}=\dot{t}\partial_{t}\vert_{\gamma}+\dot{r}\partial_{r}\vert_{\gamma}$. Using that $\gamma$ is parametrized by arclength, the outward pointing unit normal $\eta$ to $\photo$ is given by $\eta=\frac{\dot{r}}{f(r)}\partial_{t}+f(r)\dot{t}\partial_{r}$ so that
\begin{align*}
\eta(N)\vert_{\gamma}&=\eta(\sqrt{f})\vert_{\gamma}=\frac{\dot{t}\sqrt{f(r)}f'(r)}{2}\vert_{\gamma},
\end{align*}
so $\operatorname{sign}(\eta(N))=\operatorname{sign}(f')$ along $\photo$ as $\gamma$ is future pointing and hence $\dot{t}>0$.
\end{proof}

\begin{Cor}[Outward/inward directedness in Reissner--Nordstr\"om spacetimes]\label{coro:RNoutward}
Consider an $n+1$-dimensional Reissner--Nordstr\"om spacetime of mass $m$ and charge $q$ with $r>r_{m,q}$. If $m\geq\vert q\vert$ but not $m=q=0$, all  spherically symmetric photon surfaces with $r>r_{m,q}$ are outward directed. If $0<m<\vert q\vert$, spherically symmetric photon surfaces are outward directed as long as their radius satisfies $r^{n-2}>\frac{q^{2}}{m}$, inward directed where $r^{n-2}<\frac{q^{2}}{m}$ and degenerate where they cross $r^{n-2}=\frac{q^{2}}{m}$. When $m<0$ or when $m=0$ and $q\neq0$, all spherically symmetric photon surfaces are inward directed. In particular, all spherically symmetric photon surfaces in Schwarzschild spacetimes of mass $m$ are outward directed when $m>0$ and inward directed when $m<0$. Photon surfaces in the Minkowski spacetime are degenerate.
\end{Cor}
\begin{proof}
By \Cref{prop:outwardspherical}, it suffices to analyze the critical points of $f_{m,q}\definedas N_{m,q}^{2}$, see \eqref{rnlapse}. One finds $f_{m,q}'(r)>(=,<)\,0$ if and only if $m>(=,<)\,\frac{q^{2}}{r^{n-2}}$. In the subextremal and extremal cases, $m>\frac{q^{2}}{r^{n-2}}$ is globally ensured by $r>r_{m,q}$. In the other cases, one obtains exactly what is claimed.
\end{proof}
\vspace{-3ex}
\subsection{Generalizing the existence and uniqueness analysis beyond $\mathcal{S}_{\text{ext}}$}\label{rem: generalization PS analysis}
So far, our solution analysis covered all different types of spherically symmetric photon surfaces arising in spacetimes of class $\mathcal{S}_{\text{ext}}$, in particular in subextremal Reissner--Nordstr\"om and positive mass Schwarzschild spacetimes of dimensions $n+1\geq4$ (see \Cref{rem:examplesSext}). Yet, our analysis can be extended to much more general metric coefficients $f$ with only minor changes:

For example, we have restricted $v_{\text{eff}}^f$ in $\mathcal{S}_{\text{ext}}$ to have one global maximum for reasons of simplicity. However, one can proceed analogously to the above for any number of  critical points, where one will obtain a photon sphere for each critical point\footnote{We have not shown the analysis for local minima nor saddle points $r_{*}$ of $v^{f}_{\text{eff}}$ which however can be inferred from the above.} $r_{*}\in (r_H, \infty)$ of $v_{\text{eff}}^f$, and the number of different cases of $\lambda^{2}$ to be treated in the analysis is the number of different intersection possibilities, as shown in the example of \Cref{fig:v more extrema}. 

\begin{figure}[h!]
\centering
\includegraphics{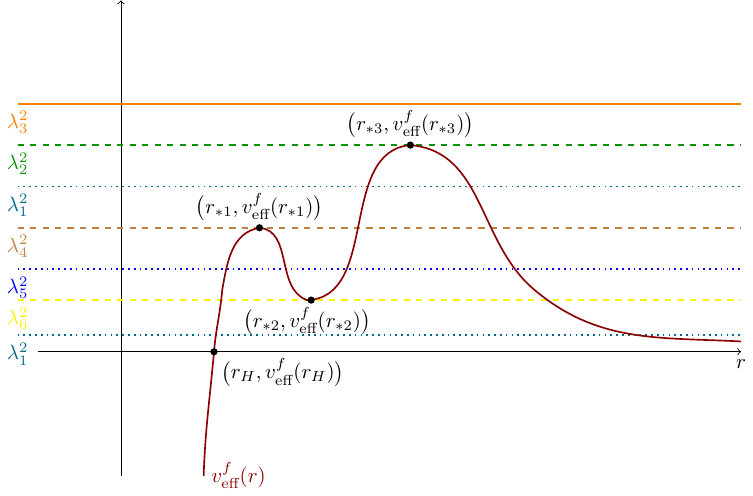}
\caption{Sketch of the function $v_{\text{eff}}^f(r)=\frac{f(r)}{r^2}$ with more extrema and seven different cases for $\lambda^{2}$.}
\label{fig:v more extrema}
\end{figure}

While the analysis remains the same, the phenomenology of the symmetric photon surfaces arising from a differently shaped effective potential $v_{\text{eff}}^{f}$ can differ drastically from what we saw above in \Cref{fig:type1,fig:type2,fig:type3}. For example, in the case labeled as \textcolor{brown}{$\lambda_{4}^{2}$} in \Cref{fig:v more extrema}, there will be symmetric photon surfaces like those asymptoting to the horizon at $r_{H}$ and to the photon sphere at $r_{\ast1}$ in \textcolor{gruen}{Case 2} (see \Cref{fig:type2}) and like those turning at some $R_{\lambda}>r_{\ast2}$ in \textcolor{greyblue}{Case 1} above (see \Cref{fig:type1}), as well as the photon spheres at $r_{\ast1}$, $r_{\ast2}$, and $r_{\ast3}$. However, in addition, there will be a new type of symmetric photon surfaces asymptoting to the photon sphere at $r_{\ast 1}$ for $t\to\pm\infty$ and turning around at some $r_{\lambda}$ with $r_{\ast2}<r_{\lambda}<r_{\ast3}$ with $v_{\text{eff}}^{f}(r_{\lambda})=v_{\text{eff}}^{f}(r_{\ast1})=\lambda_{\ast1}^{2}$, see \Cref{fig:type4}.

\begin{figure}[h]
	\centering
		\includegraphics[scale=1.2]{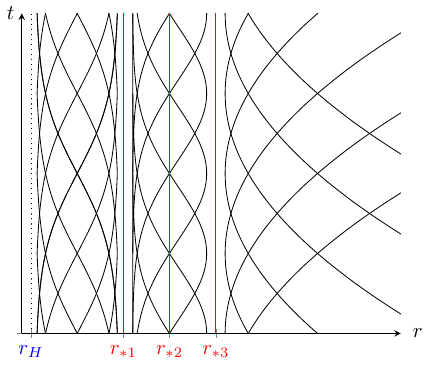}
	\caption{\text{Photon surfaces of \textcolor{brown}{Case 4}}, $\lambda^{2}=\lambda_{\ast4}^{2}$.}
	\label{fig:type4}
\end{figure}
\newpage
Next, in the case labeled as \textcolor{blue}{$\lambda_{5}^{2}$}  in \Cref{fig:v more extrema}, there will be symmetric photon surfaces like those asymptoting to the horizon at $r_{H}$ for both $t\to\pm\infty$ and turning around at some $r_{\lambda}$ with $r_{H}<r_{\lambda}<r_{\ast1}$ and like those turning at some $R_{\lambda}>r_{\ast3}$ in \textcolor{greyblue}{Case 1} above (see \Cref{fig:type1}), as well as the photon spheres at $r_{\ast1}$, $r_{\ast2}$, and $r_{\ast3}$. In addition, there will be ``trapped'' symmetric photon surfaces oscillating indefinitely between some $R_{\lambda}$ and $r_{\lambda}$ with $r_{\ast1}<R_{\lambda}<r_{\ast2}<r_{\lambda}<r_{\ast3}$, see \Cref{fig:type5} and \Cref{sec:trapped}.

\begin{figure}[h]
	\centering
		\includegraphics[scale=1.2]{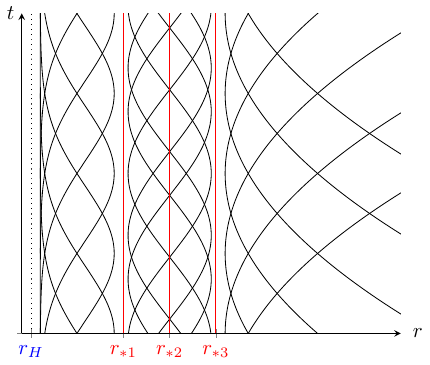}
	\caption{\text{Photon surfaces of \textcolor{blue}{Case 5}}, $\lambda^{2}=\lambda_{\ast5}^{2}$.}
	\label{fig:type5}
\end{figure}

Finally, in the case labeled as \textcolor{yellow}{$\lambda_{6}^{2}$} in \Cref{fig:v more extrema}, there will only be symmetric photon surfaces like those asymptoting to the horizon at $r_{H}$ for both $t\to\pm\infty$ and turning around at some $r_{\lambda}$ with $r_{H}<r_{\lambda}<r_{\ast1}$ and like those turning at some $R_{\lambda}>r_{\ast3}$ in \textcolor{greyblue}{Case 1} above (see \Cref{fig:type1}), as well as the photon sphere at $r_{\ast2}$, see \Cref{fig:type6}.

\begin{figure}[h!]
	\centering
		\includegraphics[scale=1.2]{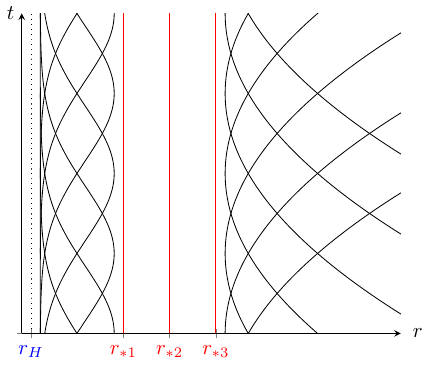}
	\caption{Photon surfaces of \textcolor{yellow}{Case 6}, $\lambda^{2}=\lambda_{\ast6}^{2}$.}
	\label{fig:type6}
\end{figure}

\subsubsection{Vertical asymptotes at some $r_{\text{v}}$}\label{sec:vertical}
One can likewise apply the techniques presented above if the function $v_{\text{eff}}^f$ has different behavior towards $r=r_{\text{v}}\geq0$, i.e., if the spacetime inner boundary is not given by a non-degenerate Killing horizon (see \Cref{rem: conditions S_ext}). In particular, the techniques can be generalized to the case of a vertical asymptote at $r=r_{\text{v}}$ as illustrated in \Cref{fig:v vertical asymptotic}. For example, in the case of a strictly decreasing $v_{\text{eff}}^f$ with vertical asymptote as depicted in \Cref{fig:v vertical asymptotic}, there is only one possibility for an intersection of the horizontal line at $\lambda^2$ and $v_{\text{eff}}^f$, and one obtains only the types of solutions for $r_0> R_{\lambda}>r_{\text{v}}$ seen in \textcolor{greyblue}{Case 1}, namely the ``hyperbolas'' on the right hand side of \Cref{fig:type1}. This is consistent with the fact that the only spherically symmetric photon surfaces in Minkowski are exactly the spatially centered one-sheeted hyperboloids of any radius. This type of effective potential occurs for negative mass Schwarzschild spacetimes and in superextremal Reissner--Nordstr\"om spacetimes with $m<\frac{2\sqrt{n-1}}{n}\vert q\vert$, i.e., when there are no photon spheres. In all those cases, one has $r_{\text{v}}=0$.

\begin{figure}[h]
\centering
\includegraphics[scale=0.9]{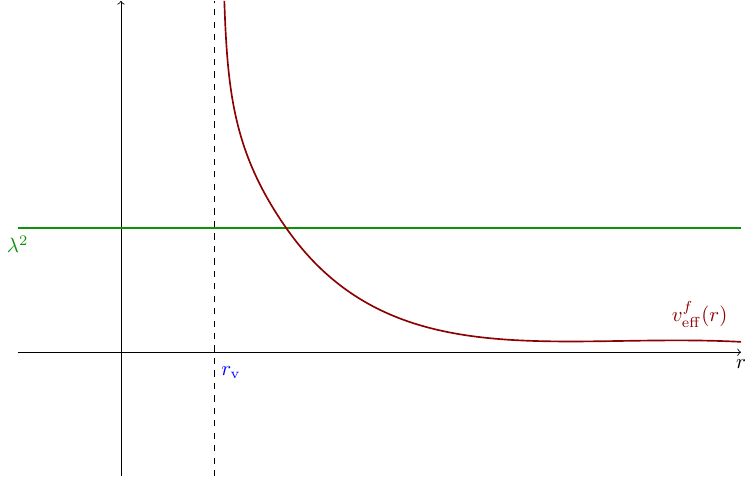}
\caption{Sketch of the function $v_{\text{eff}}^f$ with vertical asymptote at $r\equiv r_{\text{v}}$ and the only case for $\lambda^{2}$, occurring in Minkowski spacetime, in negative mass Schwarzschild spacetimes and in superextremal Reissner--Nordstr\"om spacetimes with $m<\frac{2\sqrt{n-1}}{n}\vert q\vert$, all with $r_{\text{v}}=0$.}
\label{fig:v vertical asymptotic}
\end{figure}

For superextremal Reissner--Nordstr\"om spacetimes with $\vert q\vert>m=\frac{2\sqrt{n-1}}{n}\vert q\vert$, i.e., when there is precisely one photon sphere, the effective potential will be strictly decreasing with a saddle point at $r_{\ast}=\left(\frac{nm}{2}\right)^{\frac{1}{n-2}}$, hence there will be precisely one intersection with the effective potential for each $\lambda^{2}$. If $\lambda^{2}\neq\lambda_{\ast}^{2}$, we have the same type of photon surface behavior as the hyperboloidal ones in \textcolor{greyblue}{Case 1} (see \Cref{fig:type1}), while for $\lambda^{2}=\lambda_{\ast}^{2}$, one has the unique photon sphere at $r_{\ast}$ as well as symmetric photon surfaces as the ones asymptoting this photon sphere in \textcolor{gruen}{Case 2} (see \Cref{fig:type2}). Again, $r_{\text{v}}=0$ in this case.

\subsubsection{Trapped symmetric photon surfaces which are not photon spheres}\label{sec:trapped}
It becomes more interesting when one moves to superextremal Reissner--Nordstr\"om spacetimes with two photon spheres, i.e., when $\vert q\vert>m>\frac{2\sqrt{n-1}}{n}\vert q\vert$, see also the discussion of the \textcolor{blue}{Case 5} above, \Cref{fig:type5}, and \Cref{sec:vertical}. In this case, there are \emph{trapped} symmetric photon surfaces in  which are not photon spheres (see \Cref{fig:type7,fig:veffsuper2photo}). Here, by \emph{trapped} we mean that the radial coordinate is bounded away both from the singularity ``at'' $r=r_{\text{v}}=0$ and bounded away from infinity. This may be of independent interest for example for the study of stability of the wave equation on a superextremal Reissner--Nordstr\"om spacetime in this range of mass and charge.

It is not known to the authors whether the existence of the associated trapped null geodesics\footnote{running along the trapped symmetric photon surfaces, see \Cref{sec:nullgeophase}} in superextremal Reissner--Nordstr\"om and related spacetimes was previously known, but see \cite{VE2} (in Einstein massless scalar field theory). 

\begin{figure}[h]
\centering
\includegraphics[scale=0.9]{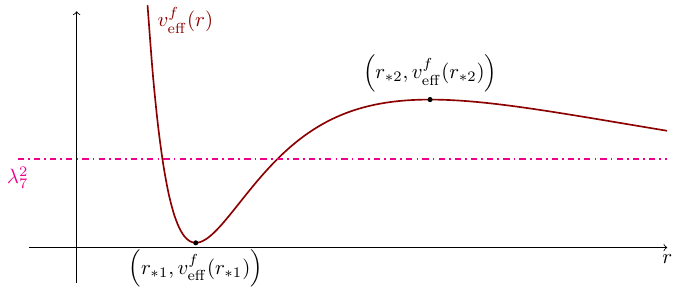}
\caption{Sketch of the function $v_{\text{eff}}^f$ for superextremal Reissner--Nordstr\"om spacetimes with $\vert q\vert>m>\frac{2\sqrt{n-1}}{n}\vert q\vert$.} 
\label{fig:veffsuper2photo}
\end{figure}

\begin{figure}[h!]
	\centering
		\includegraphics[scale=1.2]{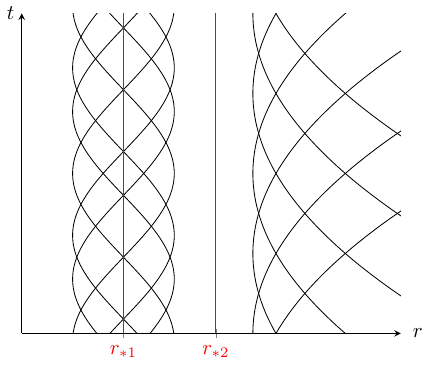}
	\caption{Trapped photon surfaces in superextremal Reissner--Norstr\"om spacetimes with two photon spheres, $\lambda^{2}=\lambda_{\ast7}^{2}$.}
	\label{fig:type7}
\end{figure}

\subsubsection{Degenerate Horizons}\label{app:B}
The same procedure can also be followed in the case of a degenerate horizon, $f(r_H)=0$ and $f'(r_H)=0$, although only on $(r_{H}+\varepsilon,\infty)$ for some fixed $\varepsilon>0$ in view of \Cref{prop: appendix}. This fits the fact that in this case -- just as in the non-degenerate case -- the surface $\lbrace r \equiv r_H\rbrace$ is in fact a principal null hypersurface and not a (timelike) photon surface. An example of a spacetime with this kind of effective potential are the extremal Reissner--Nordstr\"{o}m spacetimes. To understand more precisely what happens to symmetric photon surfaces when approaching a degenerate horizon, one needs to analyze their behavior in Gaussian null coordinates (see for example \cite[Section 2.1]{KL} or \cite{IM,FRW}), in a similar spirit as the analysis needed to understand their behavior near non-degenerate horizon (see \Cref{rem:nearhorizon}). A straightforward computation shows that any metric in Class $\mathcal{S}$ with metric coefficient $f$ can be rewritten as
\begin{align}\label{eq:GNC}
-f(r)du^{2}+du\otimes dr +dr\otimes du +r^{2}\Omega
\end{align}
where $r$ is the same radial variable as before and $u\definedas t+\int_{r_{H}}^{r}\frac{1}{f(s)}ds$. This represents a global coordinate change on $\R\times(r_H,\infty)\times\mathbb{S}^{n-1}$. If $f(r_H)=0$, \eqref{eq:GNC} smoothly extends to the hypersurface $\{r=r_H\}$ which then naturally becomes a Killing horizon of the spacetime under consideration (but not in general the entire Killing horizon).

Now, \cite[Lemma 3.4]{cedergal} informs us that along the profile curve $\gamma$ of a given symmetric photon surface in $\{r>r_H\}$ (future directed and parametrized by proper time), one has 
\begin{align}\label{eq:dotr}
\dot{r}^{2}&=\lambda^{2}r^{2}-f(r)\\\label{eq:dott}
\dot{t}&=\frac{\lambda r}{f(r)}
\end{align}
with $\lambda>0$ denoting the umbilicity constant of the photon surface. Transferring \eqref{eq:dott} into Gaussian null coordinates and applying the fundamental theorem of calculus, the last condition is equivalent to
\begin{align}\label{eq:dotu}
\dot{u}&=\frac{\lambda r+\dot{r}}{f(r)}
\end{align}
for $r>r_H$. Moreover, one finds that $\dot{u}\neq0$ along $\gamma$. The fact that $\gamma$ is future directed is equivalent to $\dot{r}<f(r)\dot{u}$ which, by the proper time parametrization condition $-f(r)\dot{u}^2+2\dot{r}\dot{u}=-1$ implies $\dot{u}>0$.

Assuming that $\gamma$ approaches the Killing horizon at $r=r_{H}$ to the future, i.e., $\dot{r}<0$ near the Killing horizon, \eqref{eq:dotr} implies $\dot{r}\to -\lambda r_{H}$. For $\dot{u}$, we compute from \eqref{eq:dotr} and \eqref{eq:dotu} that
\begin{align}
\dot{u}&=\frac{\lambda r -\sqrt{\lambda^{2}r^{2}-f(r)}}{f(r)}
\end{align}
which, assuming that the Killing horizon is non-degenerate, gives $\dot{u}\to\frac{1}{2\lambda r_{H}}$ by l'H\^opital's rule. Applying l'H\^opital's rule $k$ times, the same conclusion can also be drawn in the degenerate case provided that $f^{(k)}(r_{H})\neq0$ for some $k\geq2$. In other words, one has $\frac{dr}{du}\to-2$ as one approaches the horizon (to the future). Consequently, the symmetric photon surface must hit the Killing horizon transversally, possibly unless $f^{(k)}(r_{H})=0$ for all $k\in\mathbb{N}$.

In contrast, if $\gamma$ were to approach the Killing horizon at $r=r_{H}$ to the past, i.e., $\dot{r}>0$ near the Killing horizon, one finds
\begin{align}
\dot{r}&=\sqrt{\lambda^2r^2-f(r)}\phantom{-f(r)}\to\lambda r_H\\
\dot{u}&=\frac{\lambda r+\sqrt{\lambda^2r^2-f(r)}}{f(r)}\,\to\infty
\end{align}
and thus $\frac{dr}{du}\approx\frac{f(r)}{2}\to 0$ as one approaches the horizon (to the past). Consequently, the symmetric photon surface asymptotes to the (degenerate or non-degenerate) horizon to the past. This suggests that the symmetric photon surface leaves the part of the Killing horizon (degenerate or not) covered by the above Gaussian coordinate patch.

\subsubsection{Different asymptotic behavior}
The asymptotic behavior of symmetric photon surfaces as $r\to\infty$ in all the above cases is still captured by \Cref{rem:nearhorizon}. This changes when one allows for different (not asymptotically flat) asymptotic behavior: For our analysis for obtaining existence and uniqueness of photon surfaces, it suffices that $f$ and $f'$ are asymptotically bounded, although the asymptotic behaviour of the solutions will then differ from what is depicted above -- and one will need to use additional techniques for the proof. For example, if the function $v_{\text{eff}}^f$ has a horizontal asymptote of hight $C$ with $0<C<\lambda_{*}^{2}$ as $r\to\infty$, one continues to obtain the solutions from \textcolor{greyblue}{Case 1} with $r_{H}<r_0\leq r_{\lambda}$ for $0<\lambda^{2}<\lambda_{*}^{2}$ (see ${\color{red}\lambda_{8}^{2}}$ and ${\color{greyblue}\lambda_{1}^{2}}$ in \Cref{fig:v asymptotic}), using the same methods as in \textcolor{greyblue}{Case 1} , and for $\lambda^{2}=\lambda_{*}^{2}$ as in \textcolor{gruen}{Case 2} (see ${\color{gruen}\lambda_{2}^{2}}$ in \Cref{fig:v asymptotic}). Treating the case of initial values $r_{0}\geq R_{\lambda}$ or any initial value in case $\lambda^{2}>\lambda^{2}_{\ast}$ requires to cut off $f$ at some $r_{\mathrm{max}}<\infty$ in order to obtain the Lipschitz continuity and Lipschitz bounds necessary for the proof. One can then of course obtain global solutions by picking an unbounded sequence of such cut-off radii (see ${\color{greyblue}\lambda_{1}^{2}}$, ${\color{gruen}\lambda_{2}^{2}}$, ${\color{orange}\lambda_{3}^{2}}$ in \Cref{fig:v asymptotic} and Cases \textcolor{greyblue}{1}, \textcolor{gruen}{2}, \textcolor{orange}{3} above). This asymptotic behavior occurs for example in the (Reissner--Nordstr\"om or Schwarzschild--)Anti de Sitter spacetimes, where the asymptote is given via the cosmological constant $\Lambda$ as $C=-\frac{2}{n(n-1)}\Lambda>0$.

\begin{figure}[h!]
\centering
\includegraphics{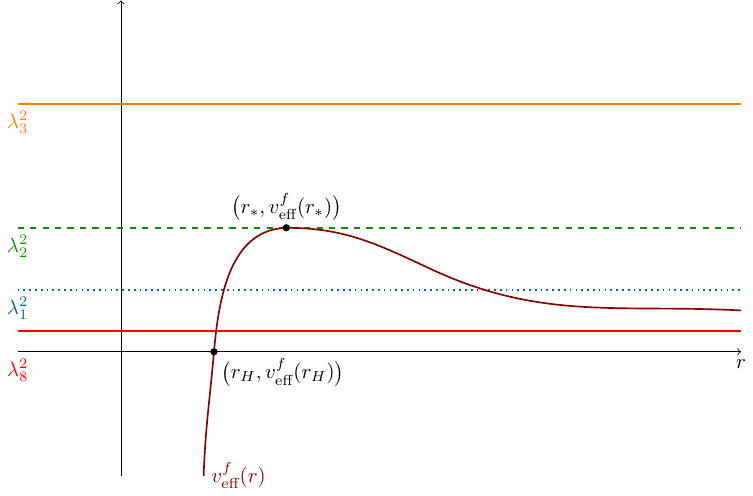}
\caption{Sketch of the function $v_{\text{eff}}^f$ with asymptote $C>0$ and four different cases for $\lambda^{2}$.}
\label{fig:v asymptotic}
\end{figure}

Another interesting case is when $v_{\text{eff}}^f$ asymptotes to $C<0$ as in the (Reissner--Nordstr\"om or Schwarzschild--)de Sitter spacetimes with $C=-\frac{2}{n(n-1)}\Lambda<0$. In this case $f$ intersects the $r$-axis for some $r_{\text{max}}>0$ (the cosmological horizon), hence one needs to extend $f$ by $0$ on $(r_{\text{max}},\infty)$ for the analysis.

\newpage

\section{Photon surfaces, null geodesics, and phase space}\label{sec:nullgeophase}
\subsection{Relating null geodesics and spherically symmetric photon surfaces}\label{subsec:generating}
As discussed in \cite{cederbaum2019photon} (see \Cref{thm:isotropic}), ``almost'' all photon surfaces in ``most'' spacetimes of class~$\mathcal{S}$ are spherically symmetric. It was also demonstrated in  \cite{cederbaum2019photon} that spherically symmetric photon surfaces can directly be related to null geodesics. As we will appeal to these results in \Cref{sec:topologyphasespace}, we will briefly summarize those considerations here.

As we will put in more concrete terms in an instant, null geodesics generate hypersurfaces the causal character of which depends on the angular momentum of the generating null geodesic(s). More precisely, let $\gamma\colon J\to \R\times \slice$ be a null geodesic in a spacetime $(\R\times \slice,\mathfrak{g})\in\mathcal{S}$ defined on some interval $J\subseteq\R.$ By the usual conservation laws for Killing vector fields, its \emph{energy} $E\definedas -\mathfrak{g}_{\gamma}(\partial_{t}\vert_{\gamma},\dot{\gamma})$ and \emph{(total) angular momentum} $\ell\definedas(G^{IJ}\ell_{I}\ell_{J})^{\frac{1}{2}}\geq0$ are constant along $\gamma$. Here, the numbers $\ell_{I}\definedas\mathfrak{g}(X_{I},\dot{\gamma})$ denote the angular momenta with respect to a local choice of linearly independent Killing vector fields $\lbrace X_{I}\rbrace_{I=1}^{n-1}$ spanning the Killing subalgebra over $\R$ of the spacetime corresponding to spherical symmetry, and $(G^{IJ})$ denotes the inverse of $(G_{IJ})\definedas(\mathfrak{g}_{\gamma}(X_{I}\vert_{\gamma},X_{J}\vert_{\gamma}))$. It is noteworthy that $\ell$ is known (and shown in \cite{cederbaum2019photon}) to be constant along $\gamma$ and moreover independent of the choice of the local linearly independent system $\lbrace X_{I}\rbrace_{I=1}^{n-1}.$

With these concepts at hand we can now make precise the notion of null geodesics generating a hypersurface and recall the relation of generating null geodesics and photon surfaces.

\begin{Prop}[\!\!{\cite[Proposition 3.13]{cederbaum2019photon}}]\label{defprop:generating}
Let $(\R\times \slice,\mathfrak{g})\in\mathcal{S}$ and let $\gamma\colon J\to \R\times \slice$, $\gamma(s)=(t(s),r(s),\xi(s))$ be a (not necessarily maximally extended) null geodesic defined on some open interval $J\subseteq\R$ with angular momentum $\ell$, where $\xi(s)\in\mathbb{S}^{n-1}$ for all $s\in J$. Then \emph{$\gamma$ generates the hypersurface $H^{n}_{\gamma}$} defined as
\begin{align*}
H^{n}_{\gamma}&\definedas\{(t,p)\in\R\times \slice\,\vert\,\exists s_{*}\in J, \xi_{*}\in\mathbb{S}^{n-1}:\,t=t(s_{*}),p=(r(s_{*}),\xi_{*})\},
\end{align*}
which is a smooth, connected, spherically symmetric hypersurface in $\R\times \slice$. The hypersurface $H^{n}_{\gamma}$ is timelike if $\ell>0$ and null if $\ell=0$.
\end{Prop}

\begin{Prop}[\!\!{\cite[Proposition 3.14]{cederbaum2019photon}}]\label{prop:generating}
Let $(\R\times \slice,\mathfrak{g})\in\mathcal{S}$ and let $H^{n}\hookrightarrow(\R\times \slice,\mathfrak{g})$ be a connected, spherically symmetric, timelike hypersurface. Then $H^{n}$ is generated by a null geodesic $\gamma\colon J\to \R\times \slice$ if and only if $H^{n}$ is a photon surface. Moreover, $H^{n}$ is a maximal photon surface if and only if any null geodesic $\gamma\colon J\to \R\times \slice$ generating $H^{n}$ is maximally extended. 

The umbilicity factor $\lambda$ of a photon surface $P^{n}$ is related to the energy $E$ and angular momentum $\ell$ of its generating null geodesics by $\lambda=\frac{E}{\ell}$.
\end{Prop}

Here (see \cite[Definition 3.12]{cederbaum2019photon}), a connected, spherically symmetric photon surface is called \emph{maximal} if it is not a subset of a strictly larger, connected, spherically symmetric photon surface.

The remaining null geodesics with $\ell=0$ -- called \emph{principal null geodesics} -- generate the so-called \emph{principal null hypersurfaces} which are called \emph{maximal} whenever the generating null geodesic is maximally extended (see \cite[Definition 3.17]{cederbaum2019photon}).

\subsection{Moving to phase space}\label{sec:topologyphasespace}
A different view on spherically symmetric photon surfaces in spacetimes of class $\mathcal{S}$ can be gained by lifting them to the phase space, that is, to the cotangent bundle of the manifold. It was shown in \cite[Proposition 3.18]{cederbaum2019photon} that the null section $\{\omega\in T^{*}\mathfrak{L}^{n+1}:\,\mathfrak{g}(\omega^{\sharp},\omega^{\sharp})=0)\}$ of the phase space of a spacetime in class $\mathcal{S}$ is partitioned by the canonical lifts of the null bundles over all maximal spherically symmetric photon surfaces together with the canonical lifts of the null bundles over all maximal principal null hypersurfaces. The following proposition is a characterization of the topologies of these lifts.

\begin{thm}[Lifts of null bundles over photon surfaces and principal null hypersurfaces in class~$\mathcal{S}$]\label{thm:topology}
Consider an $n+1$-dimensional spacetime $(\mathfrak L^{n+1},\mathfrak{g})$ of class $\mathcal S$. The lift of the future (or past) null bundle over any maximal  spherically symmetric photon surface in $\mathfrak L^{n+1}$ is a smooth submanifold of phase space and has topology $T^1\mathbb S^{n-1}\times \mathbb R^2$, where $T^{1}\mathbb S^{n-1}$ is the unit tangent bundle over $\mathbb S^{n-1}$. The lift of the future (or past) null bundle over any maximal principal null hypersurface to phase space is a smooth submanifold and has topology $\mathbb S^{n-1}\times \mathbb R^2$. 
\end{thm}

\Cref{thm:topology} generalizes Bugden's result on the topology of the lift of the future null bundle over the photon sphere in positive mass Schwarzschild spacetimes~\cite{Bugden} to general spherically symmetric photon surfaces in spacetimes of class $\mathcal{S}$.

\begin{proof}
We will consider the lifts to the tangent bundle which is equivalent to considering the lifts to the cotangent bundle via the canonical metric isomorphism. It is clear that both types of lifts are smooth submanifolds: the tangent bundle of any hypersurface $P^n$ is, of course, a smooth submanifold of $T\mathfrak L^{n+1}$, and the null section in $TP^n$ is in turn a smooth submanifold of $TP^n$. It remains to prove the topological claims about the lifts. We focus on the future case, noting that the past case is completely analogous.\\[-5ex]

\paragraph*{Photon surface case.} Let $P^n$ be a maximal spherically symmetric photon surface with umbilicity factor $\lambda>0$. We fix a time $t_0$ and denote the coordinate radius of the photon surface $P^n$ in the slice $\{t=t_0\}$ by $r_0$. Since all null geodesics on $P^n$ exist at the coordinate time $t_0$, we may parametrize them such that they start at $t(0)=t_0$. We denote differentiation with respect to the affine parameter $s$ of the null geodesics by a dot, and evaluation of any quantity at parameter time $s=0$ by a subscript $0$. By existence and uniqueness of geodesics, every maximally extended geodesic $\gamma$ is uniquely characterized by $\left(\gamma (0), \dot \gamma (0)\right)\in T_{\gamma(0)}\mathfrak{L}^{n+1}$. A future directed maximally extended null geodesic tangent to $P^{n}$ and parametrized such that $t(0)=t_{0}$ obeys $\left(\gamma (0), \dot \gamma (0)\right)=(t_0, r_0, \xi_0, -E, \dot r_0, \dot \xi_0 )$ with $(\xi_{0},\dot{\xi}_{0})\in T\mathbb{S}^{n-1}$ and $\dot{r}_{0}\in\R$, where $E=-\frac{dt}{ds}=-\mathfrak g (\partial_t, \dot \gamma)>0$ is the energy of $\gamma$ (see also \Cref{subsec:generating}). Moreover, the null condition implies that
\begin{align}\label{eq:nullcondition}
\frac{\dot{r}_{0}^{2}}{f(r_{0})r_{0}^{2}}+\Omega(\dot{\xi}_{0},\dot{\xi}_{0})&=\frac{f(r_{0})E^{2}}{r_{0}^{2}},
\end{align}
and the condition of being tangent to the photon surface $P^{n}$ informs us that either the photon sphere condition~\eqref{eq:photonsphere} holds and $\a=\frac{\sqrt{f(r_{0})}}{r_{0}}$, or the photon surface ODE~\eqref{eq:ODE3} holds for the radial profile of $P^{n}$. In the photon sphere case, we have $\dot{r}_{0}=0$ and~\eqref{eq:nullcondition} simplifies to 
\begin{align}\label{eq:nullconditionPSphere}
\Omega(\dot{\xi}_{0},\dot{\xi}_{0})&=\frac{f(r_{0})^{2}E^{2}}{\lambda^{2}r_{0}^{4}}.
\end{align}
In the other case, we have $\dot r_{0} =\left.\frac{dr}{dt}\right\vert_{t_{0}}\left.\frac{dt}{ds}\right\vert_{s=0}=-\left.\frac{dr}{dt}\right\vert_{t_{0}} E$ and hence, by~\eqref{eq:ODE3},
\begin{align}\label{eq:dotr0}
\dot{r}_{0}^{2}&=\frac{f(r_{0})^{2}E^{2}\left(\lambda^{2}r_{0}^{2}-f(r_{0})\right)}{\lambda^{2}r_{0}^{2}}.
\end{align}
Plugging this into~\eqref{eq:nullcondition} also gives~\eqref{eq:nullconditionPSphere}. Note that the photon sphere case $\dot{r}_{0}=0$, $\a=\frac{\sqrt{f(r_{0})}}{r_{0}}$ can be viewed as a special case of~\eqref{eq:dotr0}. Finally, we abbreviate $\sigma\definedas\operatorname{sgn}\dot r_0$, with $\sigma=0$ if $\dot{r}_{0}=0$.

Now let $\Phi\colon \mathbb R\times T\mathfrak L^{n+1}\rightarrow T\mathfrak L^{n+1}$ denote the geodesic flow on the tangent bundle $T\mathfrak L^{n+1}$ of $\mathfrak L^{n+1}$. Since the lift $\mathcal{N}\!^{+}\!P^n$ of the future null bundle over a photon surface in the tangent bundle is invariant under $\Phi$ and by the above characterization of future null geodesics along $P^n$, the map
\begin{align*}
T^{1}\mathbb{S}^{n-1} \times \mathbb R\times \mathbb R^{+} &\rightarrow \mathcal{N}\!^{+}\!P^n,\\
(\xi, \dot \xi, s, E) &\mapsto \Phi_s  \left(t_0, r_0, \xi, -E, \sigma \frac{f(r_0)E\sqrt{\lambda^2 r_0^2-f(r_0)}}{\lambda r_0}, 
\frac{f(r_0)E}{\lambda r_0^2}\dot{\xi}\right)
\end{align*}
is clearly a homeomorphism. This shows that $\mathcal{N}\!^{+}\!P^n$ has topology $T^{1}\mathbb{S}^{n-1}\times\mathbb{R}^{2}$ as claimed.\\[-5ex]

\paragraph*{Principal null hypersurface case.} Now assume that $P^n$ is a maximal spherically symmetric principal null hypersurface. The lift $\mathcal{N}\!^{+}\!P^n$ of the future null bundle over $P^{n}$ is invariant unter the rescaling $\left(t,r, \xi, \dot t, \dot r, \dot \xi \right)\mapsto \left(t,r, \xi, \alpha \dot t, \alpha \dot r, \alpha \dot \xi \right)$ for $\alpha >0$ and under the geodesic flow $\Phi$. Hence we may focus on a slice of constant energy $E=1$ and constant time $t=t_0$. Also recall that there is only one null direction tangent to a null hypersurface. Thus the topology of $\mathcal{N}\!^{+}\!P^n$ is $\left(P^n\cap\{t=t_0\}\right)\times \mathbb R^2$, and since by definition of spherically symmetric principal null hypersurfaces $P^n\cap\{t=t_0\}$ is homeomorphic to $\mathbb S^{n-1}$, the claim follows.
\end{proof}

\begin{Cor}
The null section of the phase space of a spacetime in class $\mathcal{S}$ is partitioned, but not smoothly foliated, by the canonical lifts of the null bundles over all maximal spherically symmetric photon surfaces together with the canonical lifts of the null bundles over all maximal principal null hypersurfaces.
\end{Cor}

\section{Uniqueness of subextremal equipotential photon\\ surfaces in electrostatic electro-vacuum spacetimes}\label{sec:unique}
This section is dedicated to proving the uniqueness result for quasi-locally subextremal equipotential photon surfaces in asymptotically flat electrostatic electro-vacuum spacetimes stated in \Cref{thm:uniqueness}. 
In \Cref{sec:prelimunique}, we will collect a few useful properties of such photon surfaces. In \Cref{sec: H>0}, we will carefully study the geometry of electrostatic spacetimes near equipotential photon surfaces and prove positivity of the mean curvature of such photon surfaces, see \Cref{prop: H>0}. Finally, in \Cref{sec:proofuniqueness}, we will prove  \Cref{thm:uniqueness}.

\subsection{Equipotential photon surfaces and their quasi-local properties}\label{sec:prelimunique}
Let us proceed to making some useful definitions and showing some quasi-local properties of equipotential photon surfaces, generalizing results from~\cite{CederPhoto,cedergal, YazaLazov, Cedrgal2, cederbaum2019photon,jahns2019photon} but also bringing in the new idea to study the evolution of $u$ in \Cref{prop:evou}.

\begin{defprop}[Canonical time slices, function $u$, unit normals $\eta$ and $\nu$, vertical vector field $Z$]\label{defprop:u}
Let $(\mathfrak{L}^{n+1},\mathfrak{g})$ be a Lorentzian manifold, $\photo\hookrightarrow\mathfrak{L}^{n+1}$ a photon surface. Let $\eta$ denote a unit normal to $\photo$ in $\mathfrak{L}^{n+1}$ (choose $\eta$ to point outward to the asymptotically flat end if $(\mathfrak{L}^{n+1},\mathfrak{g})$ is asymptotically flat and $\photo$ is a connected component of $\partial\mathfrak{L}^{n+1}$) and let $\mathfrak{h}$ and $\mathfrak{H}$ denote the second fundamental form and mean curvature of $\photo$ with respect to $\eta$, respectively. Assume that $(\mathfrak{L}^{n+1},\mathfrak{g})$ is (electro-)static and that $\photo$ is equipotential. Then there is a maximal open interval $\mathcal{T}\subseteq\R$ such that the \emph{canonical time slice of $\photo$}, $\surf(t)\definedas \photo\cap\slice(t)$, is a smooth $n-1$-dimensional submanifold for each $t\in\mathcal{T}$. Set $u\colon \mathcal{T}\to\R\colon t\mapsto N\vert_{\surf(t)}$ and call $u$ the \emph{lapse function along $\photo$}. Then $u$ is smooth and positive. Let $\gamma=(t,x)\colon I\subseteq\mathcal{T} \to\photo$ be a smooth curve with $t(s)=s$ for all $s\in I$ and such that $\dot{\gamma}(s)\perp T\surf(s)$ for all $s\in I$. If $\photo$ is non-degenerate~then 
\begin{align}\label{eq:eta}
\eta&=\frac{\nu+\frac{\dot{u}}{u^{2}\nu(N)}\partial_{t}}{\sqrt{1-\frac{\dot{u}^{2}}{u^{2}\nu(N)^{2}}}},\\\label{eq:Z}
Z&\definedas\dot{\gamma}(t)=\partial_{t}+\frac{\dot{u}}{\nu(N)}\nu
\end{align}
for $t\in I$, where $\dot{u}\definedas \frac{d}{ds}(N\circ\gamma)$ denotes the \emph{time-derivative of $u$}, $\nu$ denotes the induced \emph{(spatial) unit normal $\nu$ to $\surf(t)$}, $\nu(N)$ denotes the $\nu$-derivative of $N$, and $Z$ denotes the \emph{vertical vector field} along $\photo$. Finally, for any smooth function $w\colon\photo\to\R$ which is constant on each $\surf(t)$, we have 
\begin{align}
\dot{w}\definedas \frac{d}{ds}(w\circ\gamma)=Z(w).
\end{align}
\end{defprop}
\begin{proof}
As $\photo$ is timelike, it meets $\slice(t)$ transversally, hence for each $t\in\R$, $\surf(t)$ is either empty or a smooth $n-1$-dimensional manifold. As $\photo$ is connected, the set $\mathcal{T}\definedas\{t\in\R\,\vert\,\surf(t)\neq\emptyset\}$ is open, connected and hence an interval and maximal. The facts that $u$ is smooth and positive follow from smoothness of $\photo$ and smoothness and positivity of $N$, respectively. Formulas \eqref{eq:eta}, \eqref{eq:Z} are asserted\footnote{There is a typo in \cite[Formula (4.6)]{cederbaum2019photon}, the formula printed here is correct.} in \cite[Formulas (4.6), (4.5)]{cederbaum2019photon} for outward directed photon surfaces and can be applied here because $\nu(N)^{2}=\vert dN\vert^{2}\neq0$ by non-degeneracy of $\photo$; they can also be re-derived by an easy computation. To see that $\dot{w}=Z(w)$, apply the chain rule to $w\circ\gamma$.
\end{proof}

\begin{prop}[Quasi-local properties of equipotential photon surfaces \rm{I}]\label{prop:mathfrakH}
Let $(\mathfrak{L}^{n+1},\mathfrak{g})$ be a Lorentzian manifold, $\photo\hookrightarrow\mathfrak{L}^{n+1}$ a photon surface (not necessarily equipotential). Then its mean curvature  $\mathfrak{H}$ is constant along $\photo$ if and only if
\begin{align}\label{eq:mathfrakHconst}
\mathfrak{Ric}(\cdot\vert_{T\photo},\eta)&=0.
\end{align}
If $(\mathfrak{L}^{n+1}=\R\times \slice,\mathfrak{g}=-N^{2}dt^{2}+g)$ is static and $\photo$ is equipotential, \eqref{eq:mathfrakHconst} is equivalent to
\begin{align}
\begin{split}\label{eq:mathfrakHconst2}
N\Ric(\cdot\vert_{T\surf(t)},\nu)-\nabla^{2}N(\cdot\vert_{T\surf(t)},\nu)&=0,\\
\left(N\Ric(\nu,\nu)-\nabla^{2}N(\nu,\nu)-\triangle N\right)\dot{u}(t)&=0
\end{split}
\end{align}
on $\surf(t)$ for all $t\in\mathcal{T}$. 

In particular, if there is an electric potential $\Psi\colon\slice\to\R$ such that $(\mathfrak{L}^{n+1},\mathfrak{g},\Psi)$ is electrostatic electro-vacuum, we have $\mathfrak{H}=\text{const}.$ along any equipotential photon surface $\photo$.
\end{prop}
\begin{proof}
It follows from the contracted Codazzi identity that $\mathfrak{H}=\text{const}.$ is equivalent to \eqref{eq:mathfrakHconst} for totally umbilic hypersurfaces in arbitrary semi-Riemannian manifolds. Standard formulas for warped products tell us that
\begin{align}
\begin{split}\label{RIC}
\mathfrak{Ric}(\partial_{t},\partial_{t})&=N\triangle N,\\
\mathfrak{Ric}(\cdot\vert_{T(\slice(t))},\partial_{t})&=0,\\
\mathfrak{Ric}(\cdot\vert_{T(\slice(t))},\cdot\vert_{T(\slice(t))})&=\Ric-\frac{1}{N}\nabla^{2}N.
\end{split}
\end{align} 
A direct computation using that $Z$ and vector fields parallel to $T\surf(t)$ for each $t\in\mathcal{T}$ span $\Gamma(T\photo)$ shows equivalence of \eqref{eq:mathfrakHconst} and \eqref{eq:mathfrakHconst2} for equipotential photon surfaces in static spacetimes. Finally, using the electrostatic electro-vacuum equations~\eqref{EEVE3}--\eqref{EEVE1} together with $\vert d\Psi\vert^{2}=\nu(\Psi)^{2}$ because $\Psi=\text{const.}$ along $\surf(t)$ shows that \eqref{eq:mathfrakHconst2} is automatically satisfied in electro-vacuum and hence $\mathfrak{H}=\text{const}.$ along $\photo$.
\end{proof}

\begin{prop}[Quasi-local properties of equipotential photon surfaces \rm{II}]\label{prop:evou}
Let $(\mathfrak{L}^{n+1},\mathfrak{g})$ be a static spacetime, $\photo\hookrightarrow\mathfrak{L}^{n+1}$ a non-degenerate equipotential photon surface. Then 
\begin{align}\label{eq:evo}
\frac{u\,\mathfrak{H}}{n}\sqrt{1-\frac{\dot{u}^{2}}{u^{2}\nu(N)^{2}}}^{\,3}&=\left(1-\frac{\dot{u}^{2}}{u^{2}\nu(N)^{2}}\right)^{2}\nu(N)+\frac{1}{u}\,Z\!\left(\frac{\dot{u}}{\nu(N)}\right)-\frac{\dot{u}^{2}}{u^{2}\nu(N)}
\end{align}
holds along $\photo$. If there is a sequence $(t_{i})_{i\in\mathbb{N}}\subset\mathcal{T}$ with accumulation point $t_{0}\in\mathcal{T}$ such that $\dot{u}(t_{i})=0$ for all $i\in\mathbb{N}$ then
\begin{align}\label{eq:evoacc}
u(t_{0})\mathfrak{H}(t_{0})&=n\nu(N)(t_{0}).
\end{align}
If $\photo$ is a photon sphere, this gives
\begin{align}\label{eq:evophoto}
N\mathfrak{H}&=n\nu(N)
\end{align}
on $\photo$.
\end{prop}
\begin{proof}
From total umbilicity of $\photo$, we know that 
\begin{align*}
\mathfrak{h}(Z,Z)&=\frac{\mathfrak{H}}{n}\mathfrak{g}(Z,Z)=-\frac{u^{2}\mathfrak{H}}{n}\left(1-\frac{\dot{u}^{2}}{u^{2}\nu(N)^{2}}\right)^{2}.
\end{align*}
Now extend $\nu$ smoothly as a unit vector field perpendicular to $\partial_{t}$ to a neighborhood of $\photo$. Using this, we compute from \eqref{eq:eta}
\begin{align*}
-\sqrt{1-\frac{\dot{u}^{2}}{u^{2}\nu(N)^{2}}}\,\mathfrak{h}(Z,Z)&=\sqrt{1-\frac{\dot{u}^{2}}{u^{2}\nu(N)^{2}}}\,\mathfrak{g}(\,^{\mathfrak{g}}\nabla_{Z}Z,\eta)=\mathfrak{g}(\,^{\mathfrak{g}}\nabla_{Z}Z,\nu)+\frac{\dot{u}}{u^{2}\nu(N)}\mathfrak{g}(\,^{\mathfrak{g}}\nabla_{Z}Z,\partial_{t})\\
&=\mathfrak{g}(\,^{\mathfrak{g}}\nabla_{Z}\partial_{t},\nu)+Z\left(\frac{\dot{u}}{\nu(N)}\right)\underbrace{\mathfrak{g}(\nu,\nu)}_{=1}+\frac{\dot{u}}{\nu(N)}\underbrace{\mathfrak{g}(\,^{\mathfrak{g}}\nabla_{Z}\nu,\nu)}_{=\frac{1}{2}Z\left(\mathfrak{g}(\nu,\nu)\right)=0}\\
&\quad+\frac{\dot{u}}{u^{2}\nu(N)}\left[Z\left(\underbrace{\mathfrak{g}(Z,\partial_{t})}_{=-u^{2}}\right)-\mathfrak{g}(Z,\,^{\mathfrak{g}}\nabla_{Z}\partial_{t})\right]
\end{align*}
\begin{align*}
\phantom{\sqrt{-\frac{\dot{u}^{2}}{u^{2}}}\,\mathfrak{h}(Z,Z)}&=\mathfrak{g}(\,^{\mathfrak{g}}\nabla_{Z}\partial_{t},\nu)+Z\left(\frac{\dot{u}}{\nu(N)}\right)-\frac{2\dot{u}^{2}}{u\nu(N)}\\
&\quad-\frac{\dot{u}}{u^{2}\nu(N)}\left[\underbrace{\mathfrak{g}(\partial_{t},\,^{\mathfrak{g}}\nabla_{Z}\partial_{t})}_{=\frac{1}{2}Z(\mathfrak{g}(\partial_{t},\partial_{t}))=\frac{1}{2}Z(-u^{2})=-u\dot{u}}+\frac{\dot{u}}{\nu(N)}\mathfrak{g}(\nu,\,^{\mathfrak{g}}\nabla_{Z}\partial_{t})\right]\\
&=\left(1-\frac{\dot{u}^{2}}{u^{2}\nu(N)^{2}}\right)^{2}\underbrace{\mathfrak{g}(\,^{\mathfrak{g}}\nabla_{Z}\partial_{t},\nu)}_{=u\nu(N)}+Z\left(\frac{\dot{u}}{\nu(N)}\right)+\frac{\dot{u}^{2}}{u\nu(N)}\\
&=\left(1-\frac{\dot{u}^{2}}{u^{2}\nu(N)^{2}}\right)^{2}u\nu(N)+Z\left(\frac{\dot{u}}{\nu(N)}\right)+\frac{\dot{u}^{2}}{u\nu(N)}.
\end{align*}
Taken together, we obtain \eqref{eq:evo}. Now, if $\dot{u}(t_{i})=0$ along a sequence accumulating at $t_{0}$, we find $\dot{u}(t_{0})=0$ by continuity of $\dot{u}$ and $\ddot{u}(t_{0})=0$ by the definition of derivatives. Using that 
\begin{align*}
Z\left(\frac{\dot{u}}{\nu(N)}\right)=\frac{\ddot{u}}{\nu(N)}-\frac{\dot{u}Z(\nu(N))}{\nu(N)^{2}}
\end{align*}
gives \eqref{eq:evoacc} and in particular \eqref{eq:evophoto} if $\photo$ is a photon sphere in which case all $t_{0}\in\R$ are such accumulation times.
\end{proof}

\begin{prop}[Quasi-local properties of equipotential photon surfaces \rm{III}]\label{prop:scal}
Let $(\mathfrak{L}^{n+1},\mathfrak{g},\Psi)$ be an electrostatic electro-vacuum spacetime, $\photo\hookrightarrow\mathfrak{L}^{n+1}$ a non-degenerate equipotential photon surface. Then the scalar curvature $\operatorname{R}_{p}$ of $\photo$ with respect to the induced metric $p$ of $\photo$ satisfies 
\begin{align}\label{eq:Rp}
\operatorname{R}_{p}&=\frac{2(n-2)\vert d\Psi\vert^{2}}{(n-1)u^{2}}+\frac{(n-1)\mathfrak{H}^{2}}{n}.
\end{align}
In particular, $\operatorname{R}_{p}=\text{const.}$ along $\photo$ if and only if $\frac{\vert d\Psi\vert}{u}=\text{const.}$ along $\photo$. In particular, $\operatorname{R}_{p}\vert_{\surf(t)}=\text{const.}$ for each $t\in\mathcal{T}$.
\end{prop}
\begin{proof}
Using the contracted Gauss equation and exploiting the total umbilicity of $\photo$ as well as staticity of the spacetime (see \eqref{RIC}), we find 
\begin{align}
\operatorname{R}_{p}-\frac{n-1}{n}\mathfrak{H}^{2}&=\operatorname{R}-\frac{2\triangle N}{N}-\frac{2}{1-\frac{\dot{u}^{2}}{u^{2}\nu(N)^{2}}}\left(\Ric(\nu,\nu)-\frac{1}{N}\nabla N(\nu,\nu)+\frac{\dot{u}^{2}}{u^{4}\nu(N)^{2}}N\triangle N\right)
\end{align}
along $\photo$. Plugging in the electro-vacuum equations and using $\vert d\Psi\vert^{2}=\nu(\Psi)^{2}$ along $\photo$ leads to
\begin{align*}
\frac{2(n-2)\vert d\Psi\vert^{2}}{(n-1)u^{2}}&=\operatorname{R}_{p}-\frac{n-1}{n}\mathfrak{H}^{2}
\end{align*}
along $\photo$. Using~\Cref{prop:mathfrakH} to get $\mathfrak{H}=\text{const.}$ demonstrates the remaining claims.
\end{proof}

\begin{prop}[Quasi-local properties of equipotential photon surfaces \rm{IV}]\label{prop:I}
Let $(\mathfrak{L}^{n+1},\mathfrak{g})$ be a static spacetime, $\photo\hookrightarrow\mathfrak{L}^{n+1}$ a non-degenerate equipotential photon surface, $\surf$ a canonical time slice of $\photo$. Let $h$, $H$, $\sigma$, and $\operatorname{R}_{\sigma}$ denote the second fundamental form and mean curvature of $\surf\hookrightarrow\slice$, its induced metric and scalar curvature, respectively. Then $\surf$ is totally umbilic in $\slice$ and we have
\begin{align}\label{eq:HmathfrakH}
H&=\frac{n-1}{n}\mathfrak{H}\sqrt{1-\frac{\dot{u}^{2}}{u^{2}\nu(N)^{2}}}
\end{align}
on $\surf$. If $(\mathfrak{L}^{n+1},\mathfrak{g},\Psi)$ is an electrostatic electro-vacuum spacetime, then $H$, $\nu(N)$, and $\operatorname{R}_{\sigma}$ are constant on $\surf$ and the \emph{photon surface identity}
\begin{align}\label{foundtheconstant}
\operatorname{R}_\sigma&= \frac {2\nu(\Psi)^2}{N^2} + \frac {2H\nu(N)}{N} +\frac{(n-2)H^{2}}{(n-1)}
\end{align}
holds on $\surf$.
\end{prop}
\begin{proof}
Total umbilicity of $\surf$ in $\slice$ as well as~\eqref{eq:HmathfrakH} follow by a direct computation just as in~\cite[Formula (4.9)]{cederbaum2019photon}, where again only non-degeneracy of $\photo$ is used. By \Cref{prop:mathfrakH}, $\mathfrak{H}$ is constant along $\surf$. Let $X\in\Gamma(T\surf)$, then by \eqref{eq:HmathfrakH}, 
\begin{align}\label{eq:*}
X(H)&=\frac{(n-1)\mathfrak{H} \dot{u}^{2}}{nu^{2}\,\nu(N)^{3}\sqrt{1-\frac{\dot{u}^{2}}{u^{2}\nu(N)^{2}}}}X(\nu(N)).
\end{align}
On the other hand, by~\eqref{EEVE3} and using that $N\vert_{\surf}\equiv u$ and $X(\Psi)=0$, one finds
\begin{align*}
X(\nu(N))&=\nabla^2 N(X, \nu)=  N\Ric(X, \nu) +\frac{2\,X(\Psi) \nu(\Psi)}{N}-\frac{2|d\Psi|^2}{(n-1)N}\underbrace{g(X, \nu)}_{=0}= u\Ric(X, \nu).
\end{align*}
Using total umbilicity of $\surf$ in $\slice$, the contracted Codazzi equation gives $\Ric(X, \nu)= \frac{n-2}{n-1}X(H)$ and hence
\begin{align*}
X(H)&=\frac{(n-2)\mathfrak{H} \dot{u}^{2}}{nu\,\nu(N)^{3}\sqrt{1-\frac{\dot{u}^{2}}{u^{2}\nu(N)^{2}}}}X(H)
\end{align*}
so that $X(H)=0$ except possibly where 
\begin{align}\label{eq:**}
\frac{(n-2)\mathfrak{H} \dot{u}^{2}}{nu}&=\nu(N)^{3}\sqrt{1-\frac{\dot{u}^{2}}{u^{2}\nu(N)^{2}}}.
\end{align}
Assume towards a contradiction that $X(H)\neq0$ on an open subset $\emptyset\neq U\subseteq\surf$. Observe that the left hand side of \eqref{eq:**} is constant along $\surf$ and thus, taking an $X$-derivative, we find
\begin{align*}
0&=\left(3-\frac{2\dot{u}^{2}}{u^{2}\nu(N)^{2}}\right)\underbrace{X(\nu(N))}_{\neq0}
\end{align*}
on $U$, where we have used non-degeneracy of $\photo$ and \eqref{eq:*}. Repeating the argument on the first factor in this identity gives $X(\nu(N))=0$ on $U$ which contradicts \eqref{eq:*}. Hence $U=\emptyset$ and $H$ and $\nu(N)$ are constant along $\surf$.

Finally, from~\eqref{EEVE1} and \eqref{EEVE3} as well as by the standard decomposition of the Laplacian on functions, we find that
\begin{align*}
\frac{2(n-2)|d\Psi|^2}{(n-1)N} =\Delta N&=\Delta_\sigma N+\nabla^2 N(\nu, \nu)+H\nu(N)\\
&=N\Ric(\nu, \nu)+\frac{2(n-2)|d\Psi|^2}{(n-1)N}+H\nu(N)
\end{align*}
on $\surf$, where $\Delta_{\sigma}$ denotes the Laplacian with respect to $\sigma$. Plugging the resulting identity $\Ric(\nu, \nu)=-\frac H N \nu(N) $ into the contracted Gauss equation, we get 
\begin{align*}
\operatorname{R}_\sigma-\frac{n-2}{n-1}H^2&=\operatorname{R}+ \frac {2H} N \nu(N)=\frac {2|d\Psi|^2} {N^2} + \frac {2H} N \nu(N),
\end{align*}
where we have also used~\eqref{EEVE4}. This proves~\eqref{foundtheconstant} and shows that $\operatorname{R}_{\sigma}=\text{const.}$ along $\surf$ because $|d\Psi|^2=\nu(\Psi)^{2}$ along $\surf$ as $\Psi$ is constant on $\surf$.
\end{proof}

\begin{Cor}[Photon sphere constraint]\label{coro:photoconstraint}
Let $(\mathfrak{L}^{n+1},\mathfrak{g})$ be a static spacetime, $\photo\hookrightarrow\mathfrak{L}^{n+1}$ a non-degenerate photon sphere. Then 
\begin{align}\label{eq:photoconstraint}
NH&=(n-1)\nu(N)
\end{align}
holds for all canonical time slices of $\photo$ (without assuming any field equations!). In particular, photon spheres satisfy $\nu(N)H>0$ (or, alternatively, $\eta(N)\mathfrak{H}>0$) and have $H=\text{const.}$ if and only if $\nu(N)=\text{const.}$ along any, or every, time slice of $\photo$ (or, alternatively, $\mathfrak{H}=\text{const.}$ if and only if $\eta(N)=\text{const.}$ along $\photo$).
\end{Cor}
This was previously known only in vacuum and in electro-vacuum for $n=3$.
\begin{proof}
From \Cref{prop:mathfrakH}, in particular from \eqref{eq:evophoto}, and from \Cref{prop:I}, in particular from \eqref{eq:HmathfrakH}, we know that\\[-4ex]
\begin{align*}
(n-1)\nu(N)&=\frac{(n-1)}{n}N\mathfrak{H}=NH,
\end{align*}
where we have used that $\dot{u}=0$ along photon spheres. This implies all remaining claims via \eqref{eq:HmathfrakH} and \eqref{eq:eta} which gives $\eta=\nu$.
\end{proof}

\begin{Rem}[Comparison to vacuum case: the constant $c$]\label{rem:constc}
Let $(\mathfrak{L}^{n+1},\mathfrak{g},\Psi)$ be an asymptotically flat electrostatic electro-vacuum spacetime, $\photo\hookrightarrow\mathfrak{L}^{n+1}$ a non-degenerate equipotential photon surface, $\surf$ a canonical time slice of $\photo$. From \Cref{prop:I}, we know that $H$, $\nu(N)$, and $\operatorname{R}_{\sigma}$ are constant on $\surf$, and we know that $N$, $\Psi$, and $\vert d\Psi\vert$ are constant on $\surf$ by definition. Provided that $H\neq0$ (see \Cref{prop: H>0} below), we can hence introduce a constant $c\definedas\frac{2\nu(N)}{NH}$ such that~\eqref{foundtheconstant} and the definition of $c$ can be rewritten as the system of algebraic equations
\begin{align}\label{eq:c}
\operatorname{R}_{\sigma}&=\frac{2\vert d\Psi \vert^2}{N^2} +\left(c+\frac{n-2}{n-1}\right)H^2\\
2\nu(N)&=cNH
\end{align}
on $\surf$.
In the photon sphere case, \Cref{coro:photoconstraint} tells us that $c=\frac{2}{n-1}$. 

This system of equations coincides with the ``photon surface constraints'' assumed in the vacuum uniqueness theorems proved in~\cite{ndimunique}. The constant $c$ defined here generalizes that introduced in~\cite{ndimunique}; it is related to the constant $c$ used in the application of said uniqueness theorem to equipotential photon surfaces in \cite[Theorem 4.2]{cederbaum2019photon} by addition of $\frac{n-2}{n-1}$, see \cite[Formula (4.15)]{cederbaum2019photon}. 

In the (vacuum) Schwarzschild spacetime of mass $m$, $c$ can be thought of as a function of the (area) radius coordinate. Then $c$ is strictly decreasing with image $\R^{+}$ for $m>0$, $c\equiv0$ for $m=0$, and strictly increasing with image $(-\frac{n-2}{2(n-1)},0)$ for $m<0$. In this sense, $c$ is a useful indicator function in the vacuum case. The situation is more complicated in electro-vacuum: For the subextremal and extremal Reissner--Nordstr\"om spacetimes of mass $m$ and charge $q$, $c$ is still strictly decreasing with image $\R^{+}$. However, in the superextremal case with $q\neq0$, $c$ is strictly increasing with image $(-\frac{n-2}{n-1},0)$ in case $m\leq0$, but changes sign and has two critical points in case $m>0$. Hence $c$ stops being as useful as an indicator function in electro-vacuum, except when one restricts ones attention to the extremal and subextremal cases.
\end{Rem}

The next part of our analysis of quasi-local properties of equipotential photon surfaces consists of analyzing the frequency of so-called ``exceptional'' time slices of equipotential photon surfaces, namely time slices where the photon surfaces ``turn around'' (compare \Cref{fig:casesintro}). This will be useful for our proof of positivity of the mean curvature of photon surfaces in \Cref{prop: H>0}.
\begin{Def}[Exceptional/generic time slices]
Let $(\mathfrak{L}^{n+1},\mathfrak{g},\Psi)$ be an electrostatic spacetime, $\photo\hookrightarrow\mathfrak{L}^{n+1}$ a non-degenerate equipotential photon surface which is not a photon sphere, i.e., for which there exists $t_{*}\in\mathcal{T}$ with $\dot{u}(t_{*})\neq0$. A time $t_{0}\in\mathcal{T}$ with $\dot{u}(t_{0})=0$ is called \emph{exceptional (time) for $\photo$} and the canonical time slice $\surf(t_{0})$ is called \emph{exceptional time slice (of $\photo$)}. Non-exceptional times and time slices are called \emph{generic}. Slightly abusing notation, these definitions also apply to equipotential photon surfaces in static spacetimes without electric potential. 
\end{Def}

Before we justify these names, let us prove the following lemma which will be used in its proof. Note that this lemma does not assume any field equations.
\begin{Lem}[Sufficient condition for existence of photon spheres]\label{lem:isphoto}
Let $(\mathfrak{L}^{n+1},\mathfrak{g},\Psi)$ be an electrostatic spacetime and let $\surf\hookrightarrow\slice$ be a smooth, totally umbilic hypersurface of a canonical time slice $\slice(t)$ and assume that $N$, $\Psi$, and $\vert d\Psi\vert$ are constant on $\surf$ and $dN$ is non-vanishing on $\surf$. Assume furthermore that \eqref{eq:photoconstraint} is satisfied on $\surf$. Then $\R\times\surf\hookrightarrow\mathfrak{L}^{n+1}$ is a non-degenerate photon sphere in $(\mathfrak{L}^{n+1},\mathfrak{g},\Psi)$.
\end{Lem}
\begin{proof}
Let $\nu$ denote a unit normal to $\surf$. Then $\R\times\surf$ is a timelike hypersurface because $\nu$ naturally extends by time translation to a spacelike unit normal to $\R\times\surf$. Let $h$, $H$ and $\mathfrak{h}$, $\mathfrak{H}$ denote the second fundamental forms and mean curvatures of $\surf\hookrightarrow\slice$ and $\R\times\surf\hookrightarrow\mathfrak{L}^{n+1}$, respectively, and let $\sigma$ and $p$ denote the induced metrics on $\surf$ and $\R\times\surf$, respectively. Then \eqref{eq:photoconstraint} and $h=\frac{H}{n-1}\sigma$ imply
\begin{align*}
-\mathfrak{h}(\partial_{t},\partial_{t})&=\mathfrak{g}(\,^{\mathfrak{g}}\nabla_{\partial_{t}}\partial_{t},\nu)=N\nu(N)=\frac{N^{2}H}{n-1}=-\frac{H}{n-1}p(\partial_{t},\partial_{t}),\\
-\mathfrak{h}(X,\partial_{t})&=\mathfrak{g}(\,^{\mathfrak{g}}\nabla_{X}\partial_{t},\nu)=0=-\frac{H}{n-1}p(X,\partial_{t}),\\
-\mathfrak{h}(X,Y)&=\mathfrak{g}(\,^{\mathfrak{g}}\nabla_{X}Y,\nu)=g(\nabla_{X}Y,\nu)=-h(X,Y)=-\frac{H}{n-1}\sigma(X,Y)=-\frac{H}{n-1}p(X,Y)
\end{align*}
for all $X,Y\in\Gamma(T(\R\times \surf))$ with $X,Y\perp_{\mathfrak{g}}\partial_{t}$, using time-symmetry (i.e., vanishing of the second fundamental form) of $\slice(t)\hookrightarrow\mathfrak{L}^{n+1}$ in the last equation. Hence $\photo\hookrightarrow\mathfrak{L}^{n+1}$ is totally umbilic and thus a non-degenerate photon sphere in $(\mathfrak{L}^{n+1},\mathfrak{g},\Psi)$.
\end{proof}

The following proposition justifies the terminology of exceptional and generic time slices (without assuming any field equations!). 
\begin{Prop}[Exceptional time slices are exceptional]\label{prop:exceptional}
Let $(\mathfrak{L}^{n+1},\mathfrak{g},\Psi)$ be an electrostatic spacetime, $\photo\hookrightarrow\mathfrak{L}^{n+1}$ a non-degenerate equipotential photon surface which is not a photon sphere. Then $\photo$ has at most countably many exceptional time slices $\surf(t_{i})$ with $\vert t_{i}-t_{j}\vert>\varepsilon$ for all $i\neq j$ and for a suitable constant $\varepsilon>0$. In other words, if a non-degenerate equipotential photon surface $\photo$ in an electrostatic spacetime has a sequence of exceptional times $(t_{i})_{i\in\mathbb{N}}\subseteq\mathcal{T}$ which accumulates at $t_{0}\in\mathcal{T}$ then $\photo$ must be a photon sphere.
\end{Prop}
\begin{proof}
Suppose towards a contradiction that there is a time $t_{0}\in\mathcal{T}$ which is an accumulation point of exceptional times $t_{i}\neq t_{0}$, $i\in\mathbb{N}$. Note that $t_{0}$ is also an exceptional time by continuity of $\dot{u}$, hence \eqref{eq:HmathfrakH} informs us that $H(t_{0})=\frac{n-1}{n}\mathfrak{H}(t_{0})$. By \Cref{prop:evou}, we  know that $u(t_{0})\mathfrak{H}(t_{0})=n\nu(N)(t_{0})$, hence \eqref{eq:photoconstraint} holds on the totally umbilic hypersurface $\surf(t_{0})\hookrightarrow \slice(t_{0})$ and \Cref{lem:isphoto} applies. Thus $\R\times \surf(t_{0})$ is a photon surface in $(\mathfrak{L}^{n+1},\mathfrak{g},\Psi)$. A direct computation also shows that $\R\times \surf(t_{0})\hookrightarrow\mathfrak{L}^{n+1}$ has mean curvature $\mathfrak{H}(t_{0})$. Moreover, $N$, $\Psi$, and $\vert d\Psi\vert$ are constant along $\R\times \surf(t_{0})$ as $N$, $\Psi$, and $g$ are time-independent and $N$, $\Psi$, and $\vert d\Psi\vert$ are constant along $\surf(t_{0})$. Finally, $dN\neq0$ along $\R\times\surf(t_{0})$ as $dN\neq0$ on $\surf(t_{0})$. Hence $\R\times\surf(t_{0})$ is a non-degenerate photon sphere in $(\mathfrak{L}^{n+1},\mathfrak{g},\Psi)$.

Next, note that the vertical vector field $Z$ of $\photo$ restricts to $Z\vert_{\surf(t_{0})}=\partial_{t}\vert_{\surf(t_{0})}$  as $\dot{u}(t_{0})=0$ by \eqref{eq:Z} and hence coincides with the vertical vector field of $\R\times\surf(t_{0})$ on $\surf(t_{0})$. Now let $\mu\colon I\to\mathfrak{L}^{n+1}$, $0\in I$, $I$ open, be a future directed null geodesic with $\mu(I)\subset\photo$, $\mu(0)\in\surf(t_{0})$ (which exists by total umbilicity of $\photo\hookrightarrow\mathfrak{L}^{n+1}$ via \Cref{prop:umbilic}). Then $\dot{\mu}(0)=X_{0}+\frac{\vert X_{0}\vert}{u(t_{0})} \partial_{t}\vert_{\mu(0)}$ for a unique $X_{0}\in T_{\mu(0)}\surf(t_{0})$. Thus $\mu$ touches $\R\times\surf(t_{0})$ in  $\mu(0)$ and $\mu(I)\subset\R\times\surf(t_{0})$ by total umbilicity of $\R\times\surf(t_{0})\hookrightarrow\mathfrak{L}^{n+1}$ via \Cref{prop:umbilic}. Hence, writing $\mu=(\tau,x)\colon I\to\R\times\slice$, we obtain $u(\tau)=N(x)=u(t_{0})$ along $\mu$ which implies that $u$ is constant on the open interval $\tau(I)\subseteq\mathcal{T}$ as $\mu$ is null. The same holds for $\Psi$ and $\vert d\Psi\vert$, hence $\photo\cap(\tau(I)\times\slice)$ is a photon sphere. If $\tau(I)=\mathcal{T}$, we have achieved the desired contradiction. Hence suppose that $\tau(I)\subsetneq\mathcal{T}$, and without loss of generality that $\sup\tau(I)<\sup\mathcal{T}$ (otherwise flip the time orientation of the spacetime). But then $\sup\tau(I)$ is an accumulation point of a sequence of exceptional times in $\tau(I)$ for $\photo$ and we can repeat the above argument to conclude that $\photo$ coincides with the photon sphere $\R\times\surf(t_{0})$ also to the future of $\tau(I)$, achieving the desired contradiction.
\end{proof}

\begin{Rem}[On the assumption of $\vert d\Psi\vert=\vert d\Psi\vert(t)$ along $\photo$]\label{rem:normpsi}
In \Cref{def:photonsphere} and \Cref{def:equipot}, we assumed that not only $N$ and $\Psi$ but also  $\vert d\Psi\vert$ is constant along each time slice of a photon sphere or equipotential photon surface. We would like to point out that we have not used this condition in proving any of the Propositions~\ref{defprop:u}--\ref{prop:I} and \ref{prop:exceptional} nor in \Cref{coro:photoconstraint} and \Cref{lem:isphoto}, except for the last sentence of \Cref{prop:scal} and \Cref{rem:constc} (which are both just stated for convenience and will not be used later). In fact, the only reason why we have included this condition into both definitions is that for proving our uniqueness result \Cref{thm:uniqueness} as well as the local characterization result \Cref{prop:normalquasiloc} we will need that $\operatorname{R}_{\sigma}$ and $\vert d\Psi\vert$ are constant along each canonical time slice $\surf$ of an equipotential photon surface $\photo$. This assumption already (implicitly) appeared in the $3+1$-dimensional case treated in~\cite{Cedrgal2} as well as in the $n+1$-dimensional photon sphere case studied in~\cite{jahns2019photon}. In vacuum,  \eqref{foundtheconstant} shows that $\operatorname{R}_{\sigma}$ is automatically constant on $\surf$ without any further assumptions than $N$ being constant along each canonical time slice $\surf$ of a photon surface $\photo$.

Coming back to electro-vacuum, it is well-known that the lapse function $N$ and electric potential $\Psi$ of an asymptotically flat electrostatic electro-vacuum system with \emph{connected} non-degenerate black hole inner boundary are functionally related (see e.g.~\cite[Corollary 9.6]{Heusler} in dimension $3+1$, the same argument applies in higher dimensions). This functional relationship can in fact also be established if the inner boundary is a \emph{connected} non-degenerate equipotential photon surface (see e.g.~\cite{YazaLazov} for the photon sphere case in dimension $3+1$, the method being a straightforward generalization of the black hole case). Using this functional relationship and the fact that $\nu(N)$ is constant on every canonical time slice (see \Cref{prop:I}), it follows by chain rule that $\vert\nu(\Psi)\vert=\vert d\Psi\vert$ is constant on every canonical time slice as well. However, this functional relationship and hence constancy of $\vert d\Psi\vert$ on canonical time slices of otherwise equipotential photon surfaces cannot be established in the presence of multiple boundary components to the best knowledge of the authors.
\end{Rem}

The final part of our investigation into quasi-local properties of equipotential photon surfaces is to define quasi-local charge, mass, and extremality properties.

\begin{defprop}[Quasi-local mass, charge, and subextremality]\label{def:subex}
Let $(\mathfrak{L}^{n+1},\mathfrak{g},\Psi)$ be an electrostatic electro-vacuum spacetime, $\photo\hookrightarrow\mathfrak{L}^{n+1}$ a (not necessarily non-degenerate) equipotential photon surface, and $\surf\hookrightarrow\photo$ a standard time slice with $\operatorname{R}_{\sigma}>0$. The \emph{scalar curvature radius $r$ of $\surf$} is defined as
 \begin{align}\label{eq:radius}
  r&\definedas r(\surf)\definedas\sqrt{\frac{(n-1)(n-2)}{\operatorname{R}_{\sigma}}},
 \end{align}
and the \emph{(quasi-local, electric) charge $q\in\R$ of $\surf$} is defined as
\begin{align}\label{qidef}
 q&\definedas -\sqrt{\frac{2}{(n-1)(n-2)}}\frac{\nu(\Psi)}{N}r^{n-1}.
\end{align}
Now assume that $r^{n-2}> \vert q\vert$. The unique solution $m\in\R$ of  
\begin{align}\label{defmi}
\frac{\vert\nu (N)\vert}{N} \sqrt{1-\frac{2m}{r^{n-2}}+\frac{q^2}{r^{2(n-2)}}}&=(n-2)\left(\frac{m}{r^{n-1}}-\frac{q^2}{r^{2n-3}}\right)
\end{align}
satisfies $m>0$. If $\eta(N)\mathfrak{H}>0$, $m$ is called the \emph{(quasi-local) mass of $\surf$}. If $\eta(N)\neq0$, these quasi-local quantities satisfy the inequalities $r^{n-2}>\frac{q^{2}}{m}$ and $r>r_{m,q}$, while $r^{n-2}=\frac{q^{2}}{m}$ if $\eta(N)=0$. 

Provided that $r^{n-2}>\vert q\vert$, we say $\surf$ is \emph{(quasi-locally exterior) subextremal} if $m>\vert q\vert$, \emph{(quasi-locally) extremal} if $m=\vert q\vert$, and \emph{(quasi-locally) superextremal} if $m<\vert q\vert$. If all canonical time slices of $\photo$ are (exterior) subextremal / extremal / superextremal, $\photo$ itself is called \emph{(quasi-locally) (exterior) subextremal / extremal / superextremal}.
\end{defprop}
\begin{proof}
By the assumption $\operatorname{R}_{\sigma}>0$, $r>0$ and $q\in\R$ are well-defined. Let us now prove that \eqref{defmi} admits a unique solution $m\in\R$ provided that $r^{n-2}> \vert q\vert$: First, using $0<\gamma\definedas\frac{q^2}{r^{2(n-2)}}<1$, we set
\begin{align*}
a(\mu)&\definedas \frac{\vert\nu (N)\vert}{N} \sqrt{1-2\mu+\gamma}-\frac{(n-2)}{r}\left(\mu-\gamma\right)
\end{align*}
where $\mu<\frac{1+\gamma}{2}<1$. Then $a\colon(-\infty,\frac{1+\gamma}{2})\to\R$ is a smooth function with 
\begin{align*}
a'(\mu)&= -\frac{\vert\nu (N)\vert}{N\sqrt{1-2\mu+\gamma}} -\frac{(n-2)}{r}<0.
\end{align*}
This shows that there is at most one solution $m\definedas r^{n-2}\mu$ of \eqref{defmi}. As 
\begin{align*}
a(0)&=\frac{\vert\nu (N)\vert}{N} \sqrt{1+\gamma}+\frac{(n-2)\gamma}{r}>0,\\
\lim_{\mu\nearrow\frac{1+\gamma}{2}}a(\mu)&=\frac{(n-2)(\gamma-1)}{2r}<0,
\end{align*}
there must hence exist a unique root $\mu$ of $a$, and it must satisfy $0<\mu<\frac{1+\gamma}{2}$. In other words, there is a unique solution $m\in\R$ of \eqref{defmi} which must satisfy $m>0$. If $\eta(N)=0$ and hence $\nu(N)=0$ by \eqref{eq:HmathfrakH}, the vanishing of the left hand side in \eqref{defmi} gives that $r^{n-2}=\frac{q^{2}}{m}$ as claimed. If $\eta(N)\neq0$ and hence $\nu(N)\neq0$ by \eqref{eq:HmathfrakH}, \eqref{defmi} implies that $r^{n-2}\geq\frac{q^{2}}{m}$ with equality if and only if $r=r_{m,q}>0$. Hence, equality can only arise when $m\geq \vert q\vert$ and $r^{n-2}=m+\sqrt{m^{2}-q^{2}}=\frac{q^{2}}{m}$ which implies $m=\vert q\vert$ and thus contradicts the assumption $r^{n-2}>\vert q\vert$.
 \end{proof}

It can be checked by a direct calculation that the quasi-local mass $m$ and charge $q$ of any such time slice of any non-degenerate equipotential photon surface in a Reissner--Nordstr\"om spacetime coincide with its parameters $m$ and $q$ whenever they are defined. Hence, any such photon surface is exterior subextremal resp.\ extremal if and only if the Reissner--Nordstr\"om spacetime is exterior subextremal resp.\ extremal. For superextremality, this equivalence continues to hold provided that $m>0$ and $r^{n-2}>\vert q\vert$; compare also the condition $r^{n-2}>\frac{q^{2}}{m}$ to the outward directedness threshold asserted in \Cref{coro:RNoutward}.

\begin{Rem}[On $r^{n-2}>\vert q\vert$]\label{rem:roleofqrn-1}
The assumption $r^{n-2}>\vert q\vert$ in \Cref{def:subex} has the following significance: Inserting the definition of $q$ and $\operatorname{R}_{\sigma_{\ast}}=(n-1)(n-2)$ into \eqref{foundtheconstant} for $r^{n-2}=\vert q\vert$, the scalar curvature term and the charge term cancel, and we remain with the condition $c+\frac{n-2}{n-1}=0$ for $c$ defined in \Cref{rem:constc} (unless $H=0$). This is to say that the coefficient of $H^{2}$ in \eqref{eq:c} changes sign precisely at $r^{n-2}=\vert q\vert$ which means that for prescribed scalar curvature, electric charge, and radius, the (positive) mean curvature fails to be uniquely determined for $r^{n-2}\leq\vert q\vert$.  See also \cite[Section 8.1]{AlbaCarlaStefano} for a related discussion.
\end{Rem}

\begin{Rem}[Local results]\label{rem:local}
Note that we have not assumed that the photon surfaces under consideration in any of the results in this and the next section are maximal (in the sense of inclusion) or that their canonical time slices are (geodesically) complete or even closed (compact without boundary). In other words, these results are all purely local in character. This is also why we do not assume outward directedness in combination with $\mathfrak{H}>0$ but rather $\eta(N)\mathfrak{H}>0$ which is invariant under change of normal direction. 
\end{Rem}

Our next step will be to prove that $m$ and $q$ are in fact constant in this setting. In particular, subextremality will turn out to be a property of the photon surface itself, not just of its individual slices (see \Cref{coro:qconstmconst}). This will require a careful understanding of the intrinsic geometry of non-degenerate equipotential photon surfaces in electrostatic electro-vacuum spacetimes.

\subsection{Geometry near equipotential photon surfaces}\label{sec: H>0}
This section is dedicated to a careful analysis of the intrinsic geometry of non-degenerate equipotential photon surfaces with $\eta(N)\mathfrak{H}>0$ in electrostatic electro-vacuum spacetimes. In particular, we will show that $q$ and $m$ (see \Cref{def:subex}) are constant along such photon surfaces. We will use this detailed understanding to assert that $\mathfrak{H}>0$ must hold along outward directed equipotential photon surfaces if they arise as (connected components of) the inner boundary of the spacetime under consideration, provided that the other connected components of the inner boundary (if there are any) are non-degenerate black hole horizons (see \Cref{prop: H>0}). This will then be used in \Cref{sec:proofuniqueness} to show the main uniqueness result of this paper, \Cref{thm:uniqueness}.

The analysis of the intrinsic geometry of equipotential photon surfaces turns out to be most conveniently carried out by locally projecting the photon surface (along the static Killing vector field $\partial_{t}$) onto a fixed time slice $\slice(t_{0})$, for two reasons: First, we already know the induced electro-vacuum equations on such a time slice, \eqref{EEVE3}--\eqref{EEVE2}. Secondly, projecting onto a time slice frees us from incorporating the information on the time-dependence of the photon surface captured in the function $u$ and its evolution equation \eqref{eq:evo} --- so that the intrinsic geometry can be recovered from the analysis we will carry out, using $\partial_{N}=\left(Z-\partial_{t}\right)$. In addition, we will also explicitly need a detailed understanding of the geometry of a time slice $\slice(t_{0})\hookrightarrow\mathfrak{L}^{n+1}$ on the projection of such equipotential photon surfaces for proving \Cref{prop: H>0} at the end of this section. 

\begin{defprop}[Geometry near canonical time slices of photon surfaces]\label{defprop:projection}
Let $(\mathfrak{L}^{n+1},\mathfrak{g},\Psi)$ be an electrostatic electro-vacuum spacetime, $\photo\hookrightarrow\mathfrak{L}^{n+1}$ a non-degenerate equipotential photon surface which is not a photon sphere. Let $t_{0}$ be a generic time for $\photo$. Then there exists a threshold $\delta>0$ such that $\photo\cap((t_{0}-\delta,t_{0}+\delta)\times\slice)$ projects bijectively onto an open neighborhood $V\subseteq\slice(t_{0})$ of $\surf(t_{0})$. One can choose $N$ as a smooth coordinate function on $V$ and combine the coordinate $N$ with local coordinates $(y^{I})$ on some open subset $U\subseteq\surf(t_{0})$ flown to a suitable open neighborhood $V_{U}\subseteq V$ along $\frac{\grad N}{\vert dN\vert^{2}}$. The neighborhood $V_{U}\subseteq\slice(t_{0})$ is regularly foliated by totally umbilic connected level sets of $N$, denoted as $\surf_{N}\definedas\surf(u^{-1}(N))$, on which $\nu(N)$, $\Psi$, $\nu(\Psi)$, $H$, and $\operatorname{R}_{\sigma}$ are constant and \eqref{foundtheconstant} is satisfied.  In the coordinates $(N,y^{I})$, the spatial metric $g$ satisfies
\begin{align}\label{eq:metricIsrael}
g&=\rho^{2}dN^{2}+r^{2}\sigma_{\ast}\definedas\rho^{2}dN^{2}+r^{2}(\sigma_{\ast})_{IJ}\,dy^{I}dy^{J},
\end{align}
on $V_{U}$, with smooth functions
\begin{align}\label{def:rho}
\rho\colon(u(t_{0})-\varepsilon,u(t_{0})+\varepsilon)&\to\R\colon N\mapsto (\nu(N))^{-1}\vert_{\surf_N},\\\label{def:radialfunction}
r\colon(u(t_{0})-\varepsilon,u(t_{0})+\varepsilon)&\to\R\colon N\mapsto r_{0}\sqrt{\exp\left(\int_{u(t_{0})}^{N}\frac{2H(v)\rho(v)}{n-1}\,dv\right)},
\end{align}
where $N(V_{U})=(u(t_{0})-\varepsilon,u(t_{0})+\varepsilon)$ for some $\varepsilon>0$. Here and in what follows, we suggestively write $\Psi=\Psi(N)$, $H=H(N)$, etc.\ as shorthands for the smooth function $\Psi=\Psi(\surf_{N})$ and so on, where $N\in N(V_{U})$. The metric $\sigma_{\ast}$ is given by
\begin{align}\label{def:sigmaast}
\sigma_{\ast}\definedas &\frac{1}{r_{0}^{2}}\sigma\vert_{\surf_{u(t_{0})}},
\end{align}
where the number $r_{0}>0$ is determined as follows: If $R_{\sigma(u(t_{0}))}>0$, we set $r_{0}$ to be the scalar curvature radius of $\surf_{u(t_{0})}$, see \eqref{eq:radius}, and we have $\operatorname{R}_{\sigma_{\ast}}=(n-1)(n-2)$. If $R_{\sigma(t_{0})}<0$, 
we choose $r_{0}$ as the \emph{(negative) scalar curvature radius} of $\surf_{u(t_{0})}$, i.e., $\operatorname{R}_{\sigma_{\ast}}=-(n-1)(n-2)$ and $\operatorname{R}_{\sigma(u(t_{0}))}=-\frac{(n-1)(n-2)}{r_{0}^{2}}$ on $\surf_{u(t_{0})}$. Finally, if $R_{\sigma_{\ast}}=0$, we set $r_{0}\definedas1$. In any case, we have $r(u(t_{0}))=r_{0}$. Furthermore, $r(N)$ is the (positive or negative) scalar curvature radius of $\surf(N)$ for $N\in N(V_{U})$ as long as $\operatorname{R}_{\sigma(N)}\neq0$. Last but not least, we have
\begin{align}\label{eq:rhor}
\nu(N)&=\frac{1}{\rho}\\\label{eq:Hr}
H&=\frac{(n-1)r'}{r\rho},\\\label{eq:nupsi}
\nu(\Psi)&=\frac{\Psi'}{\rho}
\end{align}
on $(u(t_{0})-\varepsilon,u(t_{0})+\varepsilon)$, where $'$ denotes differentiation by $N$.  
\end{defprop}
\begin{proof}
The assumptions that $\photo$ is non-degenerate and that $\dot{u}(t_{0})\neq0$ guarantee the claims about the projection onto the time slice $\slice(t_{0})$, including the claims about the existence and properties of $\delta$, $V$, $U$, and $V_{U}$, and $(y^{I})$. By construction and as $\photo$ is non-degenerate, $V_{U}$ is regularly foliated by connected level sets of $N$, and we know from \Cref{prop:I} that these level sets are totally umbilic with $\nu(N)$, $\Psi$, $\nu(\Psi)$, $H$, and $\operatorname{R}_{\sigma}$ constant along each such level set. It is hence legitimate to interpret these quantities as smooth functions on $N(V_{U})$. Furthermore, from the construction of the coordinates $(N,y^{I})$, it is immediate that $g$ can be written as
\begin{align*}
g&=\rho^{2}dN^{2}+\sigma_{N}=\rho^{2}dN^{2}+(\sigma_{N})_{IJ}\,dy^{I}dy^{J}
\end{align*}
on $V_{U}$, with $\sigma_{N}$ denoting the metric induced on $\surf_{N}$. This gives $\nu=\rho^{-1}\partial_{N}$ and hence \eqref{eq:nupsi}. Using this, total umbilicity of $\surf_{N}\hookrightarrow\slice(t_{0})$ and the definition of the mean curvature $H$, we find the evolution equation
\begin{align*}
(\sigma_{N})_{IJ,N}&=\frac{2H\rho}{(n-1)}(\sigma_{N})_{IJ}
\end{align*}
which has the unique solution
\begin{align*}
(\sigma_{N})_{IJ}&=\exp\left(\int_{u(t_{0})}^{N}\frac{2H(v)\rho(v)}{(n-1)}\,dv\right)(\sigma_{u(t_{0})})_{IJ}.
\end{align*}
Using \eqref{def:sigmaast}, \eqref{def:radialfunction} and the definition of $r_{0}$, one finds $\sigma_{N}=r(N)^{2}\sigma_{\ast}$ which proves \eqref{eq:metricIsrael}, the scalar curvature radius nature of $r$, and $r(u(t_{0}))=r_{0}$. From \eqref{def:radialfunction}, we find \eqref{eq:Hr} by differentiation.
\end{proof}

\begin{Thmdef}[Local characterization of the spatial metric]\label{prop:qconstmconst}\label{prop:geometrynearslices}\label{prop:normalquasiloc}
Under the assumptions of \Cref{defprop:projection}, we extend the definition of $q=q(N)$ in \eqref{qidef} identically to our now more general setting, using the same formula but with $r=r(N)$ as defined in \Cref{defprop:projection}. Then $q$ is constant on $N(V_{U})$. 

If $\mathfrak{H}=0$ then $(V_{U},g)$ is isometric to a cylinder $(I\times\surf_{\ast},d\tau^{2}+r_{0}^{2}\sigma_{\ast})$ for some interval $I\subseteq\R^{+}$ and some $r_{0}>0$. If moreover $q=0$ then $(\surf_{\ast},\sigma_{\ast})$ is Ricci flat, $N=\alpha\tau$ and $\Psi=\beta$ for some $\alpha>0$, $\beta\in\R$. If instead $\mathfrak{H}=0$ and $q\neq0$, $(\surf_{\ast},\sigma_{\ast})$ is an Einstein manifold with scalar curvature $\operatorname{R}_{\sigma_{\ast}}=(n-1)(n-2)$, 
\begin{align}\label{eq:cylinder}
N&=\alpha \sqrt{e^{\frac{2(n-2)\tau}{r_{0}}}+\lambda}, \quad
\Psi=\pm\alpha \sqrt{\frac{(n-1)}{2(n-2)}}\sqrt{e^{\frac{2(n-2)\tau}{r_{0}}}-\lambda}+\beta
\end{align}
for some $\alpha>0$, $\beta\in\R$, and $\lambda\in\{\pm1,0\}$. Moreover, one has $q^{2}=r_{0}^{n-2}$.

If $\mathfrak{H}\neq0$, $(V_{U},g)$ is a suitable piece of a spatial factor of a generalized Reissner--Nordstr\"om spacetime of mass $m\in\R$, charge $q\in\R$, parameter $\kappa\in\{\pm1,0\}$ and base $(\surf_{\ast},\sigma_{\ast})$ of scalar curvature $\operatorname{R}_{\sigma_{\ast}}=(n-1)(n-2)\kappa$, with metric $g$, lapse function $N$, and electric potential $\Psi$ satisfying
\begin{align}\label{eq:foundgpositive}
g&=\frac{1}{N_{\kappa,m,q}^{2}}dr^{2}+r^{2}\sigma_{\ast},\\\label{eq:foundNpositive}
N&=\alpha N_{\kappa,m,q},\\\label{eq:foundpsipositive}
\Psi&=\pm\alpha \Psi_{\kappa,m,q}+\beta
\end{align}
restricted to $V_{U}$ for some constants $\alpha>0$, $\beta\in\R$, and radial coordinate $r\in J_{\kappa,m,q}$. 

In particular, if $\eta(N)\mathfrak{H}>0$, the mass parameter $m$ coincides with the mass $m(r)$ defined in \Cref{def:subex} for $\{r\}\times\surf_{\ast}$, i.e., $m(r)=m>0$ for all $r\in r(V_{U})$, and one has $r^{n-2}>\vert q\vert$ as well as $\operatorname{R}_{\sigma_{\ast}}>0$, $\kappa=1$.
\end{Thmdef}

\begin{proof}
We first rewrite the field equations \eqref{EEVE2} and \eqref{EEVE1}, \eqref{EEVE3} in $NN$-direction as well as the constraint \eqref{foundtheconstant} with the help of \eqref{eq:metricIsrael}, obtaining
\begin{align}\label{eq: two}
\left(\frac{\Psi'}{N\rho^{2}}\right)'+\frac{\Psi'}{N\rho^{2}}\left[\frac{(n-1)r'}{r}+\frac{\rho'}{\rho}\right]&=0,\\\label{eq: three}
\frac{(n-1)r'}{r}-\frac{\rho'}{\rho}&=\frac{2(n-2)\Psi'^{2}}{(n-1)N},\\\label{eq:temp}
r''-\frac{r'\rho'}{\rho}+\frac{r\rho'}{(n-1)N\rho}&=\frac{2(n-2)r\Psi'^{2}}{(n-1)^{2}N^{2}}\\\label{eq: one}
\frac{\operatorname{R}_{\sigma_{\ast}}}{r^{2}}-\frac{2(n-1)r'}{Nr\rho^{2}}-\frac{(n-1)(n-2)r'^{2}}{r^{2}\rho^{2}}&=\frac{2\Psi'^{2}}{N^{2}\rho^{2}}
\end{align}
on $N(V_{U})$, while \eqref{EEVE3} in tangential directions just adds that $\sigma_{\ast}$ is an Einstein metric. Using \eqref{eq: three}, \eqref{eq:temp} can be rewritten as
\begin{align}
\label{eq:NN}
r''-\frac{r'}{N}-\frac{(n-1)r'^{2}}{r}&=-\frac{2(n-2)r'\psi'^{2}}{(n-1)N}.
\end{align}
By definition of $q$ and \eqref{eq:nupsi}, we find the evolution law
\begin{align*}
q'(N)&=-\sqrt{\frac{2}{(n-1)(n-2)}}\left[\left(\frac{\Psi'}{N\rho}\right)'+(n-1)\left(\frac{r'\Psi'}{Nr\rho}\right)\right]r^{n-1}.
\end{align*}
The term in square brackets vanishes by \eqref{eq: two} so that $q$ is constant on $N(V_{U})$ and 
\begin{align}\label{eq:foundpsi}
\Psi'&=-\sqrt{\frac{(n-1)(n-2)}{2}}\frac{qN\rho}{r^{n-1}}
\end{align}
holds by definition of $q$. On the other hand, one can explicitly solve for $\Psi'$ as follows: First, subtract $\frac{\Psi'}{N\rho^{2}}$ times \eqref{eq: three} from \eqref{eq: two} to get the following first order ODE for $\Psi'$
\begin{align*}
\Psi''N-\Psi'+\frac{2(n-2)}{(n-1)}(\Psi')^3=0
\end{align*}
on $N(V_{U})$. Where $\Psi'\neq0$, set $\varphi\definedas \frac{N^{2}}{\Psi'^{2}}>0$ and rewrite the above ODE in terms of $\varphi$ as $\varphi'=\frac{4(n-2)N}{(n-1)}$ which can be solved to give  
\begin{align}\label{eq:psiA}
\Psi'^{2}&=\frac{(n-1)N^{2}}{2(n-2)\left(N^{2}-A\right)}
\end{align}
 for some constant $A\in\R$ satisfying $A<N^{2}$. As this expression can never vanish, it follows that either $\Psi'\equiv0$ on $N(V_{U})$ (corresponding to $q=0$ in view of \eqref{eq:foundpsi}) or $\Psi'$ is indeed given by \eqref{eq:psiA} on $N(V_{U})$.  Combined with \eqref{eq:foundpsi} and \eqref{eq: three}, this leads to
\begin{align}\label{eq:foundrho}
\rho=\begin{cases}
\frac{\varepsilon r^{n-1}}{(n-2)q\sqrt{N^{2}-A}}&\text{when }q\neq0\\
Ar^{n-1}&\text{when }q=0
\end{cases}
\end{align}
for suitable $A\in\R$ and $\varepsilon\in\{\pm1\}$. Next, we insert \eqref{eq:foundpsi} into \eqref{eq:NN} and \eqref{eq: one} and obtain
\begin{align}\label{eq:NNpsi}
r''-\frac{r'}{N}-\frac{(n-1)r'^{2}}{r}&=-\frac{(n-2)^{2}q^{2}Nr'\rho^{2}}{r^{2(n-1)}},\\\label{eq:scalnorho}
\frac{\operatorname{R}_{\sigma_{\ast}}}{(n-1)(n-2)}&=\frac{q^{2}}{r^{2(n-2)}}+\frac{2rr'}{(n-2)N\rho^{2}}+\frac{r'^{2}}{\rho^{2}}.
\end{align}
To integrate this, we set $B\definedas\frac{r'}{N\rho}$ and use \eqref{eq: three} and \eqref{eq:foundpsi} in \eqref{eq:NNpsi}, obtaining $B'=0$ on $N(V_{U})$, so that $B$ is in fact constant. Plugging in $r'=BN\rho$ into \eqref{eq:scalnorho} gives
\begin{align}\label{eq:constraintB}
\frac{\operatorname{R}_{\sigma_{\ast}}}{(n-1)(n-2)}&=\frac{q^{2}}{r^{2(n-2)}}+\frac{2Br}{(n-2)\rho}+B^{2}N^{2}.
\end{align}
\underline{First, let $\mathfrak{H}=0$} so that $r'=0$ by \eqref{eq:Hr}, \eqref{eq:HmathfrakH}, and we have $B=0$ and $r\equiv r_{0}$. Hence \eqref{eq:constraintB} gives
\begin{align*}
\operatorname{R}_{\sigma_{\ast}}&=\frac{(n-1)(n-2)q^{2}}{r_{0}^{2(n-2)}}.
\end{align*}

\underline{When $q=0$}, using \eqref{eq:foundrho} and the choice $r_{0}=1$ in \Cref{defprop:projection}, this asserts that the metric $g$ is cylindrical over a scalar flat Einstein manifold, $g=A^{2}dN^{2}+\sigma_{\ast}$ for some $A\neq0$. Setting $\tau\definedas \vert A\vert N$, $\alpha\definedas \frac{1}{\vert A\vert}$, we get the desired form of $g$ and $N$, and the claimed form of $\Psi$ follows directly from \eqref{eq:foundpsi}. \underline{When $q\neq0$}, the choice of $r_{0}$ in \Cref{defprop:projection} gives the desired scalar curvature as well as $q^{2}=r_{0}^{2(n-2)}$. If $A=0$, \eqref{eq:foundrho} and $\tau\definedas \frac{r_{0}}{(n-2)}\ln N+C>0$ with a suitably large choice of $C\in\R^{+}$, one finds the claimed form of $N$ for $\alpha\definedas e^{-\frac{(n-2)C}{r_{0}}}>0$ and $\lambda\definedas0$. Using \eqref{eq:foundpsi}, one computes the desired form of $\Psi$. If $A\neq0$, one sets 
\begin{align*}
\tau\definedas \frac{r_{0}}{2(n-2)}\ln\left\vert\frac{1+\frac{N}{\sqrt{N^{2}-A}}}{1-\frac{N}{\sqrt{N^{2}-A}}}\right\vert>0
\end{align*}
and finds the claimed forms of $N$ and $\Psi$ for $\alpha\definedas \sqrt{\frac{\vert A\vert}{2}}$ and $\lambda\definedas\sgn A$ via \eqref{eq:foundrho} and \eqref{eq:foundpsi}.\\

\noindent\underline{Now let $\mathfrak{H}\neq0$} so that $r'\neq0$ by \eqref{eq:Hr}, \eqref{eq:HmathfrakH} and hence $B\neq0$ from $r'=BN\rho$. By chain rule, using $r'=BN\rho\neq0$, we find that indeed
\begin{align}\label{eq:formmetric}
\rho^{2}dN^{2}&=\frac{\rho^{2}}{r'^{2}}dr^{2}=\frac{1}{B^{2}N(r)^{2}}dr^{2},
\end{align}
where $N=N(r)$ denotes the inverse function of $r=r(N)$ defined on $r(N(V_{U}))$.  In addition, setting $\alpha\definedas\frac{1}{\vert B\vert}$, we find from \eqref{eq:foundpsi} and $r'=BN\rho$ that  $\Psi=\sgn(B)\alpha\Psi_{\kappa,m,q}+\beta$ for any $\kappa\in\{\pm1,0\}$, $m\in\R$, and some $\beta\in\R$ as claimed.

\underline{When $q=0$}, set $\kappa\definedas\frac{\operatorname{R}_{\sigma_{\ast}}}{(n-1)(n-2)}$, $m\definedas \frac{B}{(n-2)A}$ and use \eqref{eq:constraintB} to see that indeed $N=\alpha N_{\kappa,m,q=0}$ as $\kappa\in\{\pm1,0\}$ by \Cref{defprop:projection}. By \eqref{eq:formmetric}, $g$ also takes the desired form. \underline{When $q\neq0$}, \eqref{eq:psiA} shows that
\begin{align}\label{eq:foundN2}
\sqrt{N^{2}-A}&=\frac{\delta q}{Br^{n-2}}+D>0
\end{align}
for some $D\in\R$, $\delta\in\{\pm1\}$, where we have used the chain rule for $N=N(r)$ as the inverse function of $r=r(N)$. Plugging this into \eqref{eq:constraintB}, we find from \eqref{eq:foundrho} that
\begin{align*}
\frac{\operatorname{R}_{\sigma_{\ast}}}{(n-1)(n-2)}&=\frac{2q^{2}}{r^{2(n-2)}}+\frac{2\delta\varepsilon q^{2}}{r^{2(n-2)}}+\frac{2(\delta+\varepsilon)Dq}{B\alpha^{2}r^{n-2}}+\frac{D^{2}+A}{\alpha^{2}}
\end{align*}
for all $r\in r(V_{U})$. By linear independence of $\{1,\frac{1}{r^{n-2}},\frac{1}{r^{2(n-2)}}\}$ on the open interval $r(V_{U})$ (openness follows from $r'\neq0$), we find $\delta\varepsilon=1$, $\delta=-\varepsilon$, and $D^{2}+A=\kappa\alpha^{2}$, where we again set $\kappa\definedas\frac{\operatorname{R}_{\sigma_{\ast}}}{(n-1)(n-2)}$ and note that $\kappa\in\{\pm1,0\}$ by \Cref{defprop:projection}. Inserting these findings back into \eqref{eq:foundN2}, we obtain $N=\alpha N_{\kappa,m,q}$ for $m\definedas\varepsilon qBD$, with $\vert\varepsilon\vert=1$. Again, by \eqref{eq:formmetric}, $g$ also takes the desired form. 

It remains to show that if  $\eta(N)\mathfrak{H}>0$, we have that $m(r)=m$  for all $r\in r(V_{U})$ as well as $\operatorname{R}_{\sigma_{\ast}}>0$ and $\kappa=1$. In fact, by \eqref{foundtheconstant}, \eqref{eq:eta}, and \eqref{eq:HmathfrakH}, it immediately follows that $\operatorname{R}_{\sigma_{\ast}}>0$ and hence $\kappa=1$. Recall from \Cref{rem:BHgenRN} that $J_{1,m,q}$ has two connected components $(0,\overline{r}_{m,q})$ and $(r_{m,q},\infty)$ if $m>\vert q\vert>0$, while $J_{1,m,q}=(r_{m,q},\infty)$ for $m\geq\vert q\vert$ when $q=0$ or $m=\vert q\vert\neq0$, and $J_{1,m,q}=\R^{+}$ if $m<\vert q\vert$ or $m=q=0$. One computes that $H=\frac{(n-1)N_{1,m,q}(r)}{r}>0$ holds for the cross sections $\{r\}\times\surf_{\ast}$ with respect to $\nu=N_{1,m,q}\partial_{r}$; via \eqref{eq:HmathfrakH}, this gives $\mathfrak{H}>0$ and hence by assumption $\eta(N)>0$ and $\nu(N)>0$. As $N=\alpha N_{1,m,q}$ for $\alpha>0$ and $\nu=N_{1,m,q}\partial_{r}$, $\nu(N)>0$ on $\{r\}\times\surf_{\ast}$ is equivalent to $N_{1,m,q}'(r)>0$ or $m>\frac{q^{2}}{r^{n-2}}$. Arguing as in the proof of  \Cref{coro:RNoutward}, this shows that $r^{n-2}>\vert q\vert$ except possibly in case $r\in(0,\overline{r}_{m,q})$ for $m>\vert q\vert$. However, combining this estimate with $m>\frac{q^{2}}{r^{n-2}}$, we find $\sqrt{m^{2}-q^{2}}>m>0$, a contradiction. Hence $r^{n-2}>\vert q\vert$ is satisfied for all $r\in r(V_{U})$ so that  \Cref{def:subex} applies. A direct computation shows that $m$ solves \eqref{defmi}, hence $m(r)=m$ for all $r\in r(V_{U})$ as claimed.
\end{proof}

\begin{Cor}[Subextremality is global]\label{coro:qconstmconst}
Under the assumptions of \Cref{prop:qconstmconst}, it follows that $q$ is in fact constant along $\photo$ and $m$ is constant along $\photo$ provided that $\eta(N)\mathfrak{H}>0$. In particular, if $\eta(N)\mathfrak{H}>0$ holds on $\photo$, subextremality/extremality/superextremality is a global property of equipotential photon surfaces in electrostatic electro-vacuum spacetimes.
\end{Cor}
\begin{proof}
Let $\photo$ be a non-degenerate equipotential photon surface. If $\photo$ is a photon sphere, then clearly $q$ and $m$ are constant along $\photo$ (for $m$, provided that $\eta(N)\mathfrak{H}>0$). If $\photo$ is not a photon sphere, we argue as follows: On any open interval $I\subseteq\mathcal{T}$ which contains only generic times, \Cref{prop:qconstmconst} tells us that $q$ and $m$ are constant on $\photo\cap(I\times\slice)$  (for $m$, provided that $\eta(N)\mathfrak{H}>0$). By \Cref{prop:exceptional}, we know that the open set of all generic times is dense in $\mathcal{T}$, hence by continuity $q$ and $m$ are constant along $\photo$  (for $m$, provided that $\eta(N)\mathfrak{H}>0$). The claim about globality of subextremality follows directly.
\end{proof}

\begin{Rem}[Photon surface as boundary]\label{rem:boundary}
If a non-degenerate equipotential photon surface arises as part of the boundary of an electrostatic electro-vacuum spacetime, we can still perform the projection to the generic time slice $\slice(t_{0})$ performed in \Cref{defprop:projection} and \Cref{prop:geometrynearslices}: We can either do so only on one (namely the interior to $\slice(t_{0})$) side of the boundary component $\surf(t_{0})\subseteq\slice(t_{0})$, or, alternatively, by suitably (canonically) extending $\slice(t_{0})$ by projection of nearby time slices $\slice(t)$ onto $\slice(t_{0})$ if necessary, see also \Cref{rem:slice}.
\end{Rem}

\begin{Rem}
The computations in the proofs of \Cref{defprop:projection} and \Cref{prop:normalquasiloc} are somewhat similar to computations appearing in the proofs of \cite[Theorems 1.2, 1.6]{Leandro}, but in a different setup and under different assumptions.
\end{Rem}

Before we proceed to proving positivity of mean curvature of photon surface inner boundary components, let us collect a few useful facts about generalized Reissner--Nordstr\"om spacetimes.
\begin{Lem}[On the sign of $\nu(N)H$ in generalized Reissner--Nordstr\"om systems]\label{lem:etaNH-RN}
Let $\kappa\in\{\pm1,0\}$, $(\surf_{*},\sigma_{*})$ be an Einstein metric of scalar curvature $\operatorname{R}_{*}=(n-1)(n-2)\kappa$, $m,q\in\R$, and consider the spatial factor $(J_{\kappa,m,q}\times\surf_{*},g_{\kappa,m,q},N_{\kappa,m,q},\Psi_{\kappa,m,q})$ of a generalized Reissner--Nordstr\"om spacetime. Then for hypersurfaces $\{r\}\times\surf_{*}$ and with respect to any unit normal $\nu$, for $\kappa=1$, one has
\begin{center}\begin{tabular}{c||c|c||c||c|}
 & $m>\vert q\vert$, $r>r_{m,q}$ & $m>\vert q\vert>0$, $0<r<\overline{r}_{m,q}$ & $m=\vert q\vert\neq0$ & $m=q=0$\\[0.5ex]\hline&&&&\\[-1.5ex]
$\nu(N)H$& \textcolor{blue}{$>0$} & \textcolor{red}{$<0$} & \textcolor{blue}{$>0$} & $=0$ \\[0.5ex]
\end{tabular}

\phantom{.}\vspace{3ex}

\begin{tabular}{c||c|c|c|c|}
 & $0<m<\vert q\vert$, $r^{n-2}>\frac{q^{2}}{m}$  & $0<m<\vert q\vert$, $r^{n-2}<\frac{q^{2}}{m}$ & $m=0$, $q\neq0$ & $m<0$\\[0.5ex]\hline&&&&\\[-1.5ex]
$\nu(N)H$& \textcolor{blue}{$>0$} & \textcolor{red}{$<0$} & \textcolor{red}{$<0$} & \textcolor{red}{$<0$}\\[0.5ex]
\end{tabular}
\end{center}
In contrast, for $\kappa=0$ and $\kappa=-1$, one has $\textcolor{red}{\nu(N)H <0}$ unless $m=q=0$ in which case one has $\nu(N)H=0$.
\end{Lem}
\begin{proof}
For the sake of legibility, we will suppress all indices $\kappa,m,q$ and use $\nu=N\partial_{r}$. Then $H=\frac{(n-1)N}{r}>0$ and $\nu(N)=\frac{(n-2)}{r^{n-1}}\left(m-\frac{q^{2}}{r^{n-2}}\right)$. Combining this with \Cref{rem:BHgenRN}, we immediately find what is claimed when $\kappa=0$. The same applies to $\kappa=-1$ upon realizing that $\frac{q^{2}}{m}>-m+\sqrt{m^{2}+q^{2}}$ when $m>0$ and $q\neq0$. When $\kappa=1$, $m\geq\vert q\vert$ and $r>r_{m,q}$ imply that $r^{n-2}>\frac{q^{2}}{m}$ and hence $\nu(N)>0$, while $m>\vert q\vert$ and $0<r<\overline{r}_{m,q}$ gives $r^{n-2}<\frac{q^{2}}{m}$ and hence $\nu(N)<0$. If $m\leq0$, one has $\nu(N)<0$ unless $m=q=0$.
\end{proof}

In particular, we learn from \Cref{lem:etaNH-RN} that equipotential photon surfaces with $\eta(N)\mathfrak{H}>0$ can only exist in exterior subextremal and in extremal generalized Reissner--Nordstr\"om spacetimes, and in the region $r^{n-2}>\frac{q^{2}}{m}$ of superextremal generalized Reissner--Nordstr\"om spacetimes with $m>0$. Note that, in view of \Cref{def:subex}, $m>0$ in all those cases.

Next, recall from \Cref{coro:photoconstraint} that $NH=(n-1)\nu(N)$ characterizes photon spheres. This motivates the following considerations which will be useful in the proof of the positivity of mean curvature of photon surface inner boundary components.
\begin{Lem}[Photon spheres in generalized Reissner--Nordstr\"om spacetimes.]\label{lem:photogenRN}
Let $\kappa\in\{\pm1,0\}$, $(\surf_{*},\sigma_{*})$ be an Einstein metric of scalar curvature $\operatorname{R}_{*}=(n-1)(n-2)\kappa$, $m,q\in\R$, and consider the generalized Reissner--Nordstr\"om spacetime of mass $m$, charge $q$, base $(\surf_{*},\sigma_{*})$ and parameter $\kappa$. For $\kappa=1$, it has photon spheres under the same conditions as the Reissner--Nordstr\"om spacetime of mass $m$ and charge $q$, sitting at the exact same photon surface radii, see \Cref{coro:SSphoto}. There are no photon spheres when $\kappa=0$ or $\kappa=-1$.

Moreover, choosing $\nu=N_{\kappa,m,q}\partial_{r}$ as unit normal, one finds that $H>0$ for all $r\in J_{\kappa,m,q}$. Then \textcolor{orange}{$NH-(n-1)\nu(N)>0$} for $\kappa=0$ and $\kappa=-1$ and, for $\kappa=1$, one has
\begin{center}\begin{tabular}{c||c|c|c||}
$m>\vert q\vert$ & $r>r_{*}$ & $r_{*}>r>r_{m,q}$ & $q\neq0$, $\overline{r}_{m,q}>r$ \\[0.5ex]\hline&&&\\[-1.5ex]
$NH-(n-1)\nu(N)$&\textcolor{orange}{$>0$} & \textcolor{green}{$<0$} & \textcolor{orange}{$>0$}\\[0.5ex]
\end{tabular}

\phantom{.hallo}\vspace{1ex}
\begin{tabular}{c||c|c||c||c||}
 & $m=\vert q\vert\neq0$, $r>r_{*}$ & $m=\vert q\vert\neq0$, $r_{*}>r$ & $m\leq0$ & $0<m<\frac{2\sqrt{n-1}}{n}\vert q\vert$\\[0.5ex]\hline&&&&\\[-1.5ex]
$NH-(n-1)\nu(N)$ & \textcolor{orange}{$>0$} & \textcolor{green}{$<0$}  & \textcolor{orange}{$>0$} & \textcolor{orange}{$>0$}\\[0.5ex]
\end{tabular}

\phantom{.hallo}\vspace{1ex}
\begin{tabular}{c||c|c||}
$m=\frac{2\sqrt{n-1}}{n}\vert q\vert\neq0$ & $r>r_{*,+}=r_{*,-}$ & $r_{*,-}>r$\\[0.5ex]\hline&&\\[-1.5ex]
 $NH-(n-1)\nu(N)$ & \textcolor{orange}{$>0$} & \textcolor{orange}{$>0$} \\[0.5ex]
 \end{tabular}

\phantom{.hallo}\vspace{1ex}
\begin{tabular}{c||c|c|c||}
$\vert q\vert>m>\frac{2\sqrt{n-1}}{n}\vert q\vert$ & $r>r_{*,+}$ & $r_{*,+}>r>r_{*,-}$ & $r_{*,-}>r$\\[0.5ex]\hline&&&\\[-1.5ex]
 $NH-(n-1)\nu(N)$ & \textcolor{orange}{$>0$} & \textcolor{green}{$<0$}  & \textcolor{orange}{$>0$}\\[0.5ex]
 \end{tabular}\,.
\end{center}
\end{Lem}
\begin{proof}
From \Cref{coro:photoconstraint}, we know that $\nu(N)H>0$ must hold on any time slice of a photon sphere. From \Cref{lem:etaNH-RN}, it hence follows that there are no photon spheres when $\kappa=0$ or $\kappa=-1$ and neither in any of the cases with $\nu(N)H\leq0$ when $\kappa=1$. Again from the proof of \Cref{lem:etaNH-RN}, using the same notation and choosing $\nu=N\partial_{r}$, we find that
\begin{align}\label{eq:wherephoto}
NH-(n-1)\nu(N)&=\frac{n-1}{r}\left(\kappa-\frac{nm}{r^{n-2}}+\frac{(n-1)q^{2}}{r^{2(n-2)}}\right).
\end{align}
Comparing this to the characterization of photon spheres in Reissner--Nordstr\"om spacetimes in the proof of \Cref{coro:SSphoto}  gives the desired result on the location of photon spheres.

The claims on the sign of $NH-(n-1)\nu(N)$ then immediately follow from \eqref{eq:wherephoto} in combination with \Cref{lem:etaNH-RN} and \Cref{coro:SSphoto}: When $\nu(N)\leq0$, $NH-(n-1)\nu(N)>0$ follows automatically. This settles the claims for $\kappa=0$ and $\kappa=-1$ and some cases of $\kappa=1$. Where $\nu(N)>0$, we rewrite \eqref{eq:wherephoto} for $\kappa=1$ as $NH-(n-1)\nu(N)=\frac{n-1}{r^{2n-3}}f(r^{n-2})$, with $f(x)\definedas x^{2}-nmx+(n-1)q^{2}$. Thus, the sign of $NH-(n-1)\nu(N)$ at some radius $r$ coincides with the sign of $f(r^{n-2})$. We note that $f$ is a convex quadratic polynomial. Its zeros in $J_{1,m,q}$ (if any) coincide with photon sphere radii. If $m>\vert q\vert$, $f(r_{m,q}^{n-2})<0$ which ensures that the claims for $m>\vert q\vert$ hold. If $m=\vert q\vert\neq0$, $f(r_{m,q}^{n-2})=0$ which settles the claims for this case. If $\vert q\vert>m>\frac{2\sqrt{n-1}}{n}\vert q\vert$, i.e., when there are two distinct photon spheres, $f$ is positive outside both and negative between their radii $r_{\pm}$ as claimed. If $m=\frac{2\sqrt{n-1}}{n}\vert q\vert$, i.e., when there is precisely one photon sphere at $r_{*,+}=r_{*,-}$, $f$ touches the $r$-axis at $r_{*,+}=r_{*,-}$ and hence $f>0$ away from $r_{*,+}=r_{*,-}$. If $0<m<\frac{2\sqrt{n-1}}{n}\vert q\vert$, i.e., when there are no photon spheres, $f>0$. This finishes the proof.
\end{proof}

\subsection{Positivity of mean curvature of equipotential photon surfaces}\label{subsec:positivity}
It is left open in \cite{cederbaum2019photon} whether outward directed equipotential photon surfaces in asymptotically flat static vacuum spacetimes must have positive mean curvature. This is of course obvious for photon spheres by \Cref{coro:photoconstraint}, regardless of the field equations under consideration. In the vacuum case, we will close this gap in \Cref{cor:outwardpositive}. In the general case, we will instead assume that $\eta(N)\mathfrak{H}>0$ along equipotential photon surfaces in electrostatic electro-vacuum spacetimes and will restrict our attention to subextremal, extremal, and some superextremal equipotential photon surfaces, see also \Cref{cor:outwardpositiveelectro} and \Cref{rem:local}. 

We now proceed to proving said positivity in our more general setting. To do so, we get inspiration from the proof of \cite[Lemma~2.6]{cedergal} which in turn is inspired by \cite[Theorem~3.1]{GalMiao}.

\begin{Thm}[Positivity of mean curvature]\label{prop: H>0}
Let $(\mathfrak{L}^{n+1},\mathfrak{g},\Psi)$ be an electrostatic electro-vacuum spacetime such that the/each associated electrostatic system $(\slice(t),g,N,\Psi)$, $t\in\mathcal{T}$, is asymptotically flat. Assume that $(\mathfrak{L}^{n+1},\mathfrak{g},\Psi)$ has a (possibly disconnected) inner boundary consisting of equipotential photon surfaces with $\eta(N)\mathfrak{H}>0$ and/or of non-degenerate black hole horizons. Assume furthermore that the equipotential photon surface boundary components are (exterior) subextremal, extremal, or superextremal with $\vert q\vert >m\geq\frac{2\sqrt{n-1}}{n}\vert q\vert$ and scalar curvature radius $r\geq r_{*,\pm}$. Then all photon surface components of the inner boundary have $\mathfrak{H}>0$ and are outward directed.
\end{Thm}

\begin{proof}
We will give the proof for a fixed time slice $(\slice(t_{0}),g,N,\Psi)$, chosen such that $t_{0}$ is a generic time for all boundary components which are neither photon spheres nor black hole horizons (such a time exists by \Cref{prop:exceptional}). We will denote the connected components of $\partial\slice(t_{0})$ by $\surf_{i}(t_{0})$, where $i=1,\dots,I$ for some $I\geq1$. Where applicable, the photon surface corresponding to $\surf_{i}(t_{0})$ will be denoted as $\photo_{i}$. As $(\slice(t_{0}),g,N,\Psi)$ is asymptotically flat, each $\surf_{i}(t_{0})$ must be closed and $(\slice(t_{0}),g)$ must be complete (see \Cref{rem:complete}). Our goal now is to show that if the mean curvature $H_{i}(t_{0})$ of $\surf_{i}(t_{0})$ satisfies $H_{i}(t_{0})\leq0$ for some $i\in\{1,\dots,I\}$, $\surf_{i}(t_{0})$ must be a black hole horizon component. This then proves that all time slices of photon surface boundary components have $H_{i}(t_{0})>0$ and hence $\mathfrak{H}_{i}>0$ by \eqref{eq:HmathfrakH}. By the assumption $\eta(N)\mathfrak{H}>0$, they are also outward directed. From now on, we will drop the explicit reference to the selected generic time $t_{0}$. Moreover, we will suggestively attach an index $i$ to any quantity referring to $\surf_{i}$ or $\photo_{i}$.

Suppose towards a contradiction that there exists a photon surface boundary component $\photo_{i_{0}}$ with $\mathfrak{H}_{i_{0}}\leq0$. Then in fact $\mathfrak{H}_{i_{0}}<0$ and $\photo_{i_{0}}$ is inward directed as $\eta(N)_{i}\mathfrak{H}_{i}>0$ holds for all $i=1,\dots,I$. By \eqref{eq:HmathfrakH}, we also know that $H_{i_{0}}<0$ on $\surf_{i_{0}}$. 

Let us now prepare $(\slice,g,N,\Psi)$ for the main part of the argument: Consider any non-photon sphere photon surface boundary component $\photo_{j}$ with $j\neq i_{0}$. By non-degeneracy of $\photo_{j}$ and because $t_{0}$ is generic, we can apply \Cref{prop:geometrynearslices}, deducing that there is a neighborhood $\surf_{j}\subseteq V_{j}\subset\slice$ on which $(V_{j},g,N,\Psi)$ satisfies \eqref{eq:foundgpositive}--\eqref{eq:foundpsipositive} for suitable parameters $m_{j},q_{j}\in\R$, $\kappa_{j}\in\{\pm1,0\}$, $\alpha_{j}>0$, $\beta_{j}\in\R$, and such that the induced metric $\sigma_{j}$ on $(\surf_{*})_{j}=\surf_{j}$ is an Einstein metric of scalar curvature $\operatorname{R}_{\sigma_{j}}=\frac{(n-1)(n-2)\kappa_{j}}{r_{j}^{2}}$ with $r_{j}>0$ as in \Cref{prop:geometrynearslices}. In particular, $(V_{j},g,N,\Psi)$ is in fact a piece of a spacial factor of a generalized Reissner--Nordstr\"om spacetime, up to scaling (see \Cref{rem:scaling}). Moreover, by the assumption that $\eta(N)_{j}\mathfrak{H}_{j}>0$,  we learn from \Cref{lem:etaNH-RN} that $\kappa_{j}=1$ and that $m_{j},q_{j}$ correspond to one of the \textcolor{blue}{blue cases} in  \Cref{lem:etaNH-RN}.

If $\photo_{j}$ is not a photon sphere, we know from \Cref{lem:isphoto} that $N_{j}H_{j}-(n-1)\nu(N)_{j}\neq0$. \underline{If $N_{j}H_{j}-(n-1)\nu(N)_{j}<0$ (Case A)}, we remove a small neighborhood $U_{j}$ of $\surf_{j}$, i.e., $\surf_{j}\subset U_{j}\subsetneq V_{j}$, for $\slice_{j}$ such that $\slice_{j}\setminus U_{j}$ obtains a new inner boundary component $\widetilde{\Sigma}^{n-1}_{j}$ with mean curvature $\widetilde{H}_{j}<0$; we do so such that $\widetilde{\Sigma}^{n-1}_{j}$ is a radial surface in the rescaled generalized Reissner--Nordstr\"om system $(V_{j},g,N,\Psi)$ and hence also a time slice of a non-degenerate equipotential photon surface in the corresponding generalised Reissner--Nordstr\"om spacetime (at least locally in time), see \Cref{sec:spherical}. Depending on whether $\nu_{j}$ points in direction of $\partial_{r}$ or in the opposite direction, this corresponds to  removing a piece corresponding to a radius interval $[r_{j},r_{j}+\varepsilon_{j})$ (same direction) or $(r_{j}-\varepsilon_{j},r_{j}]$ (opposite direction), for some suitably small $\varepsilon_{j}>0$. We choose $\varepsilon_{j}$ small enough to ensure that $\overline{N}_{j}\overline{H}_{j}-(n-1)\overline{\nu}(\overline{N})_{j}<0$ also holds for the new boundary component $\overline{\Sigma}^{\,n-1}_{j}$ and such that the neighborhoods excised around different boundary components do not intersect.

On the other hand, \underline{if $N_{j}H_{j}-(n-1)\nu(N)_{j}>0$ (Case B)}, we need to refer to \Cref{lem:photogenRN}: From there, we learn that either $\nu_{j}$ points in the same direction as $\partial_{r}$ and we are in of the \textcolor{orange}{orange cases} or $\nu_{j}$ points in the opposite direction from $\partial_{r}$ and we are in of the \textcolor{green}{green cases} of \Cref{lem:photogenRN}. Recall that we already know that $\kappa_{j}=1$ and that by assumption and by \Cref{lem:photogenRN}, $\photo_{j}$ is either (exterior) subextremal or extremal and hence has a unique photon sphere at some radius $r_{*_{j}}>r_{m_{j},q_{j}}$ or $\photo_{j}$ is superextremal with one or two photon spheres at $r=r_{\ast_{j},\pm}>0$ and $r_{j}> r_{\ast_{j},\pm}$. This leaves the following possibilities when there is only one photon sphere: Either $\surf_{j}$ lies \textcolor{orange}{``outside''} the respective photon sphere, $r_{j}>r_{*_{j}}$, and $\nu_{j}$ points in the \textcolor{orange}{same direction} as $\partial_{r}$ or $\surf_{j}$ lies \textcolor{green}{``inside''} the respective photon sphere, $r_{*_{j}}>r_{j}>r_{m_{j},q_{j}}$, and $\nu_{j}$ points in the \textcolor{green}{opposite direction} from $\partial_{r}$. The latter possibility does not occur in the superextremal case with $m=\frac{2\sqrt{n-1}}{n}\vert q\vert$ as $r>r_{*_{j},+}=r_{*_{j},-}$. When there are two photon spheres, i.e., in the superextremal case with $\vert q\vert >m>\frac{2\sqrt{n-1}}{n}\vert q\vert$, either $\surf_{j}$ lies \textcolor{orange}{``outside''} the outer photon sphere, $r_{j}>r_{*_{j},+}$, and $\nu_{j}$ points in the \textcolor{orange}{same direction} as $\partial_{r}$ or $\surf_{j}$ lies \textcolor{green}{``between''} the two photon spheres, $r_{*_{j},-}<r_{j}<r_{*_{j},+}$, and $\nu_{j}$ points in the \textcolor{green}{opposite direction} from $\partial_{r}$.

 \underline{In the \textcolor{orange}{orange case $r_{j}>r_{*_{j}}$}}, we smoothly extend $(\slice,g,N,\Psi)$ near $V_{j}$/across $\surf_{j}$ by the exact same formulas \eqref{eq:foundgpositive}--\eqref{eq:foundpsipositive} on $J_{j}\times\surf_{j}$ with $J_{j}=[r_{*_{j}},r_{j})\subset \{r>r_{m_{j},q_{j}}\}$, i.e., all the way ``down'' to the (outer) photon sphere at $r_{*_{j}}$ ($r_{*_{j},+}$). On the other hand, \underline{in the \textcolor{green}{green case $r_{*_{j}}>r_{j}>r_{m_{j},q_{j}}$}} \underline{\textcolor{green}{(or $r_{*_{j},+}>r_{j}>r_{*_{j},-}$)}}, we smoothly extend $(\slice,g,N,\Psi)$ near $V_{j}$/across $\surf_{j}$ by the exact same formulas \eqref{eq:foundgpositive}--\eqref{eq:foundpsipositive} on some manifold $J_{j}\times\surf_{j}$ with $J_{j}=(r_{j},r_{*_{j}}]\subset \{r>r_{m_{j},q_{j}}\}$, i.e., all the way ``up'' to the photon sphere at $r_{*_{j}}$ ($r_{*_{j},+}$). 
 
Note that the new boundary components $\overline{\Sigma}^{\,n-1}_{j}$ are time-slices of photon spheres by \Cref{coro:photoconstraint}. Hence we can also extend the spacetime $(\mathfrak{L}^{n+1},\mathfrak{g},\Psi)$ across $\photo_{j}$ to a larger electrostatic electro-vacuum spacetime with a new photon sphere inner boundary component $\overline{P}^{\,n}_{j}$ instead of $\photo_{j}$.

After this preparation, we are now left with a smooth electrostatic electro-vacuum system $(\overline{M}^{\,n},\overline{g},\overline{N},\overline{\Psi})$ with inner boundary consisting of $\surf_{i_{0}}$, and potentially also of time-slices of non-degenerate black hole horizons and/or photon spheres and/or of time-slices $\overline{\Sigma}^{\,n-1}_{j}$ of equipotential photon surfaces on which $\overline{N}_{j}\overline{H}_{j}-(n-1)\overline{\nu}(\overline{N})_{j}<0$. This system is asymptotically flat and metrically complete (see \Cref{rem:complete}).

At the same time, not removing anything (i.e., ignoring the preparations in Case A) but extending across the photon surfaces in Case B, we have obtained a smooth electrostatic electro-vacuum spacetime extension $(\overline{\mathfrak{L}}^{\,n+1},\overline{\mathfrak{g}},\overline{\Psi})$ of $(\mathfrak{L}^{n+1},\mathfrak{g},\Psi)$ the inner boundary of which consists of $\photo_{i_{0}}$, and possibly of non-degenerate black hole horizons and/or photon spheres and/or and equipotential photon surfaces $\photo_{j}$ with $\overline{N}_{j}\overline{H}_{j}-(n-1)\overline{\nu}(\overline{N})_{j}<0$ on $\surf_{j}$. This spacetime is still standard static (but see \Cref{rem:boundary}) and will become relevant later in the proof.

As in \cite{cedergal}, we now consider the \emph{Fermat metric} or \emph{optical metric}
\begin{align*}
\widehat{g}\definedas\frac{1}{N^2}g,
\end{align*}
yet here on the (modified) time-slice $\overline{M}^{\,n}$, see for example Frankel~\cite{Frankel}. Without introducing a new name for $\overline{M}^{\,n}$, we hence remove the time-slices of black hole horizons from $\overline{M}^{\,n}$, as $N=0$ there which would lead to degeneracies in the Fermat metric. As is well-known, horizons with respect to $g$ become asymptotic cylinders in the Fermat metric $\widehat{g}$ and are hence ``infinitely far removed''. We will denote quantities corresponding to the Fermat metric $\widehat{g}$ by oversetting them with $\,\widehat{}\,$. 

Now let $R>0$ be large enough such that the (asymptotic) coordinate sphere $\mathbb{S}^{n-1}_R(0)$ in the asymptotic end of $(\overline{M}^{\,n},\overline{g},\overline{N},\overline{\Psi})$ has positive mean curvature with respect to the metric $g$ and the outward pointing unit normal. By compactness of $\surf_{i_{0}}$ and $\mathbb{S}^{n-1}_R(0)$, there exist points $p\in \surf_{i_{0}}$ and $q\in\mathbb{S}^{n-1}_R(0)$ such that $\operatorname{dist}_{\widehat{g}}(p,q)=\operatorname{dist}_{\widehat{g}}(\surf_{i_{0}},\mathbb{S}^{n-1}_R(0))$. From work by Alexander~\cite{alexander2006distance},  we know that there is a length minimizing unit speed $C^{1}$-curve $\widehat{\gamma}\colon [0,T]\to \overline{M}^{\,n}$ with $\widehat{\gamma}(0)=p$ and $\widehat{\gamma}(T)=q$ such that $T=\operatorname{dist}_{\widehat{g}}(p,q)=\operatorname{dist}_{\widehat{g}}(\surf_{i_{0}},\mathbb{S}^{n-1}_R(0))$ in the smooth, metrically complete Riemannian manifold with boundary $(\overline{M}^{\,n},\widehat{g}\,)$. Alexander~\cite{alexander2006distance} also proves that $\widehat{\gamma}$ is a smooth geodesic away from the boundary of $\overline{M}^{\,n}$ and that $\widehat{\gamma}$ is transversal to $\surf_{i_{0}}$ and $\mathbb{S}^{n-1}_R(0)$. Moreover, in case $\widehat{\gamma}$ intersects any boundary components of $\overline{M}^{\,n}$ other than $\surf_{i_{0}}$, it has to touch it tangentially because it is $C^{1}$. 

A direct computation shows that the mean curvature $\widehat{H}$ with respect to $\widehat{g}$ of any smooth hypersurface $\surf$ of $\overline{M}^{\,n}$ satisfies $\widehat{H}=\overline{N}\,\overline{H}-(n-1)\overline{\nu}(\overline{N})$. Next, recall that the (open/removed) boundary components corresponding to time-slices of black hole horizons are infinitely far removed with respect to $\widehat{g}$ and hence will not be reached by $\widehat{\gamma}$. As observed already in \cite[Lemma~2.6]{cedergal}, a time-slice $\overline{\Sigma}^{\,n-1}_{j}$ of a photon sphere will have $\widehat{H}_{j}=\overline{N}_{j}\overline{H}_{j}-(n-1)\overline{\nu}(\overline{N})_{j}=0$ and will hence be totally geodesic (as umbilicity is invariant under conformal transformations). This is enough to conclude that $\widehat{\gamma}$ cannot touch $\overline{\Sigma}^{\,n-1}_{j}$ by uniqueness of maximally extended geodesics and because $\widehat{\gamma}$ is $C^{1}$ and $\widehat{\gamma}(T)=q\in\mathbb{S}^{n-1}_{R}(0)$. Finally, any other boundary component $\overline{\Sigma}^{\,n-1}_{j}$ must also be umbilic (again by conformal invariance of umbilicity) and have mean curvature $\widehat{H}_{j}=\overline{N}_{j}\overline{H}_{j}-(n-1)\overline{\nu}(\overline{N})_{j}<0$ with respect to $\widehat{g}$. Inserting back the removed piece of $\slice$ across $\overline{\Sigma}^{\,n-1}_{j}$, $\overline{\Sigma}^{\,n-1}_{j}$ is now a smooth hypersurface of $\slice$ contained in a (two-sided) tubular neighborhood. Assuming that $\widehat{\gamma}$ touches some such $\overline{\Sigma}^{\,n-1}_{j}$, consider the first contact point $\widehat{\gamma}(t_0)$ of $\widehat{\gamma}$ with one such $\overline{\Sigma}^{\,n-1}_{j_{0}}$. Note that up to that point $\widehat{\gamma}$ is a smooth $\widehat{g}$-geodesic in $\overline{M}^{\,n}\cup (V_{j_{0}}\setminus U_{j_{0}})$ (although not necessarily length minimizing). By a local result by Bishop~\cite{bishop1974infinitesimal}, one concludes that the second fundamental form $\widehat{h}_{j}$ of $\overline{\Sigma}^{\,n-1}_{j}$ with respect to $\widehat{g}$ must be positive semi-definite, contradicting $\widehat{H}_{j}<0$. Hence $\widehat{\gamma}$ is a smooth, length minimizing geodesic in $(\overline{M}^{\,n},\widehat{g}\,)$ (not touching any boundary components other than $\surf_{i_{0}}$). Consequently, it must meet $\surf_{i_{0}}$ and $\mathbb{S}^{n-1}_{R}(0)$ orthogonally.

This brings us precisely to the setting studied by Galloway and Miao in the proof of~\cite[Theorem 3.1]{GalMiao} to which we refer for details (see also \cite[Lemma 2.6]{cedergal} where photon spheres serve as barriers just as in our case): $\widehat{\gamma}$ is a unit speed, length minimizing $\widehat{g}$-geodesic from $\surf_{i_{0}}$ to $\mathbb{S}^{n-1}_{R}(0)$ which does not meet the boundary of $\overline{M}^{n}$ and hits both $\surf_{i_{0}}$ and $\mathbb{S}^{n-1}_{R}(0)$ orthogonally. As $\widehat{\gamma}$ is length minimizing, there will be not $\widehat{g}$-cut points to $\surf_{i_{0}}$ along $\widehat{\gamma}$, except possibly at the endpoint $q$, a technical possibility that can be handled by slightly perturbing $\mathbb{S}^{n-1}_{R}(0)$, but preserving its positive mean curvature. Assuming that $q$ is not a cut point allows to conclude that a suitable restriction of the set $W(t)\definedas\lbrace p \in M^n\,:\, \mathrm{dist}_{\widehat{g}}(\surf_{i_{0}}, p)=t\rbrace$ is a smooth hypersurface in some neighborhood of $\widehat{\gamma}(s)$ for each $s\in [0,T]$. As $\widehat{\gamma}$ is length minimizing, we know that $W(T)$ lies to the inside of $\mathbb{S}^{n-1}_{R}(0)$ but touches it at $q$. Moving back to the original metric $g$ on $\overline{M}^{n}$, let $H(t)$ denote the mean curvature of $W(t)$ with respect to $g$ (and the unit normal pointing to the asymptotically flat end). By the maximum principle and our choice of $\mathbb{S}^{n-1}_{R}(0)$ with $H>0$, we know that $H(T)>0$. On the other hand, by the monotonicity formula \cite[Equation 3.24]{GalMiao}, the map $t\mapsto \frac{H(t)}{N\circ\widehat{\gamma}(t)}$, $t\in[0,T]$, must be non-increasing as the electro-vacuum extended spacetime $(\overline{\mathfrak{L}}^{\,n+1},\overline{\mathfrak{g}},\overline{\Psi})$ satisfies the null energy condition. This leads to a contradiction as we assumed that $\surf_{i_{0}}$ has $H(0)=H_{i_{0}}<0$ and we have just argued that $H(T)>0$. Hence all non-horizon boundary components have positive mean curvature and are outward directed as claimed.
 \end{proof}

\begin{Cor}[Mean curvature positivity in the vacuum case]\label{cor:outwardpositive}
Let $(\mathfrak{L}^{n+1},\mathfrak{g})$ be a static vacuum spacetime such that the/each associated static system $(\slice(t),g,N)$, $t\in\mathcal{T}$ is asymptotically flat. Assume that $(\mathfrak{L}^{n+1},\mathfrak{g})$ has a (possibly disconnected) inner boundary consisting of outward directed equipotential photon surfaces and/or of non-degenerate black hole horizons. Then all photon surface components of the inner boundary have $\mathfrak{H}>0$  and $\eta(N)\mathfrak{H}>0$.
\end{Cor}
\newpage
\begin{proof}
We will follow the notation and conventions from the proof of \Cref{prop: H>0}. Suppose towards a contradiction that there exists a photon surface boundary component $\photo_{i_{0}}$ with $\mathfrak{H}_{i_{0}}\leq0$. Again, we need to prepare $(\slice,g,N)$ for the main part of the argument: Consider any non-photon sphere photon surface boundary component $\photo_{j}$ with $j\neq i_{0}$. If $\eta(N)_{j}\mathfrak{H}_{j}>0$ then $\surf_{j}$ is necessarily subextremal (i.e., $m_{j}>0$) by \Cref{lem:etaNH-RN} and  we prepare $(\slice,g,N)$ as in the proof of \Cref{prop: H>0}. If $\eta(N)_{j}\mathfrak{H}_{j}\leq0$, using the assumption that $\eta(N)_{j}>0$ and thus of course $\nu(N)_{j}>0$, we conclude that $\mathfrak{H}_{j}\leq0$ and thus $H_{j}\leq0$ by \Cref{prop:I}. Hence $N_{j}H_{j}-(n-1)\nu(N)_{j}<0$ and we can excise a piece of $(\slice,g,N)$ near $\surf_{j}$ as in Case A in the proof of \Cref{prop: H>0}. Now we are in the same situation as in the proof of \Cref{prop: H>0} and can hence conclude the claim.
\end{proof}

\begin{Cor}[Mean curvature positivity]\label{cor:outwardpositiveelectro}
Let $(\mathfrak{L}^{n+1},\mathfrak{g},\Psi)$ be an electrostatic electro-vacuum spacetime such that the/each associated electrostatic system $(\slice(t),g,N,\Psi)$, $t\in\mathcal{T}$ is asymptotically flat. Assume that $(\mathfrak{L}^{n+1},\mathfrak{g},\Psi)$ has a (possibly disconnected) inner boundary consisting of outward directed equipotential photon surfaces and/or of non-degenerate black hole horizons. Assume furthermore that each equipotential photon surface inner boundary component with $\eta(N)\mathfrak{H}>0$ is (exterior) subextremal, extremal, or superextremal with $\vert q\vert >m\geq\frac{2\sqrt{n-1}}{n}\vert q\vert$ and scalar curvature radius $r\geq r_{*,\pm}$. Then all photon surface components of $\partial\mathfrak{L}$ have $\mathfrak{H}>0$ and $\eta(N)\mathfrak{H}>0$.
\end{Cor}
\begin{proof}
Again, it suffices to show that $N_{j}H_{j}-(n-1)\nu(N)_{j}<0$ holds whenever $\eta(N)_{j}\mathfrak{H}_{j}\leq0$. This follows for the same reasons as those given in the proof of \Cref{cor:outwardpositive}.
\end{proof}

\begin{Rem}[(Stable) minimal surfaces are time-slices of Killing horizons]
Relatedly, Huang, Martin, and Miao~\cite{huang2018static} (in vacuum) and Coutinho and Leandro~\cite{coutinho2022mean} (in electro-vacuum) show that stable minimal boundary components must have vanishing lapse und thus be time-slices of static Killing horizons.
\end{Rem}

\begin{Rem}[Alternative approach]
Instead of extending $(\slice,N,g,\Psi)$ smoothly via \Cref{prop:geometrynearslices} in Case B of the proof of \Cref{prop: H>0}, one could also consider gluing in suitable extensions only in a $C^{1,1}$-fashion (such that the corresponding extended spacetime satisfies the null energy condition), for example extensions such as those constructed in the proof of Theorem \ref{thm:uniqueness} below, see also \Cref{rem: gluing regularity}. Note that $C^{1,1}$-regularity is enough to provide existence and uniqueness of length minimizing curves. Moreover, \cite{alexander2006distance} yields $C^1$-regularity of length minimizers, which are smooth geodesics away from the gluing hypersurfaces. However, it would require a careful analysis to ensure that the equidistant hypersurfaces $W(t)$ along the geodesic $\widehat{\gamma}$ retain sufficient regularity such that the monotonicity formula still holds across the possible lower regularity points of the geodesic $\widehat{\gamma}$. This is beyond the scope of this paper, but see \cite[p. 58]{Galloway} and the references cited therein for a related construction.
\end{Rem}
\subsection{Proof of the uniqueness theorem}\label{sec:proofuniqueness}
This section is dedicated to proving our main uniqueness theorem for subextremal equipotential photon surfaces in asymptotically flat electrostatic electro-vacuum spacetimes:

\begin{thm}\label{thm:uniqueness}
Let $(\mathfrak{L}^{n+1},\mathfrak{g},\Psi)$ be an electrostatic electro-vacuum spacetime of dimension $n+1\geq 4$ such that the/each associated electrostatic system $(\slice(t),g,N,\Psi)$, $t\in\mathcal{T}$,  is asymptotically flat with mass $m$ and charge $q$ and that a version\footnote{For example, $n\leq7$ or $\slice$ is spin, see \Cref{rem:PMT}.} of the positive mass theorem applies to a class of manifolds that includes $\slice$. Assume that $(\mathfrak{L}^{n+1},\mathfrak{g},\Psi)$ has a (possibly disconnected) inner boundary consisting of subextremal equipotential photon surfaces with $\eta(N)\mathfrak{H}>0$ and/or of non-degenerate black hole horizons.  Then $(\mathfrak{L}^{n+1},\mathfrak{g})$ is isometric to a suitable piece of the Reissner--Nordstr\"om spacetime of mass $m$ and charge $q$, $m>\vert q\vert$, and $\Psi$ corresponds to $\Psi_{m,q}$ under this isometry. In particular, $\partial\mathfrak{L}^{n+1}$ is connected, and is a spherically symmetric photon surface or the (outer) non-degenerate black hole horizon in this Reissner--Nordstr\"om spacetime. 
\end{thm}

The proof relies on an application of \cite[Theorem~3]{jahns2019photon} and also builds upon some of the techniques used therein. However, see also \Cref{rem: gluing regularity}.

\begin{thm}[Theorem 3,~\cite{jahns2019photon}]\label{thm:uniquess-old}
 Assume $n\geq 3$. Let $\left(M^n, g, N, \Psi\right)$ be a system which fulfils the electrostatic equations~\eqref{EEVE4},~\eqref{EEVE2} and the trace of the electrostatic equation for the Ricci tensor~\eqref{EEVE1}. Assume also that $\left(M^n, g, N, \Psi\right)$ is asymptotically flat of mass $m$ and charge $q$. If  $\slice$ has an orientable, compact inner boundary whose connected components are  non-degenerate static horizons or photon sphere components, then
 $\left(M^n, g\right)$ is isometric to a piece of the Reissner--Nordstr\"om manifold of mass $m$ and charge $q$, and $m>|q|$.
 \end{thm} 

\begin{Rem}[Application of positive mass theorem]\label{rem:PMT}
\Cref{thm:uniquess-old} uses a weak regularity version of the rigidity case of the positive mass theorem applied to some metric on $\slice$ (which is conformally related to $g$). It hence relies on the weak regularity result by McFeron and Sz\'ekelyhidi \cite{MSz} (but see also \cite{Miao} as well as \cite{LeeLefloch} in the spin case) as well as on some version of the (smooth) positive mass theorem (rigidity case) with asymptotic flatness assumptions as those specified in \Cref{def:asymptotics}. For $n\leq7$, see \cite{SchoenYau}, for the spin case, see \cite{Witten}, but see also \cite{Lohkamp,SYn} and related works.
\end{Rem}

\begin{proofof}{\Cref{thm:uniqueness}}
We choose an arbitrary associated electrostatic system $(\slice(t_{0},g,N,\Psi)$, $t_{0}\in\mathcal{T}$, of $(\mathfrak{L}^{n+1},\mathfrak{g},\Psi)$, and drop the explicit reference to $t_{0}$ for simplicity. By \Cref{rem:complete}, we know that the inner boundary $\partial M$ consists of finitely many closed, orientable components. We now leave alone any boundary component $\surf_{i}$ which is already a non-degenerate electrostatic black hole horizon or a photon sphere. For any other boundary component $\surf_{j}$, we will glue in a ``neck'' down to a non-degenerate electrostatic black hole horizon, in $C^{1,1}$-regularity. The necks will satisfy  \eqref{EEVE4}, \eqref{EEVE2}, and the trace of \eqref{EEVE1}. By virtue of \Cref{thm:uniquess-old} it then follows that the extended electrostatic system $(\overline{M},g,N,\Psi)$ is isometric to a piece of the Reissner--Nordstr\"om manifold of mass $m$ and charge $q$, and $m>|q|$, with $\Psi$ corresponding to $\Psi_{m,q}$ via the isometry. This of course then applies by restriction to the original time-slice $(\slice(t_{0},g,N,\Psi)$ and hence to $(\mathfrak{L}^{n+1},\mathfrak{g},\Psi)$. Moreover, $\partial\mathfrak{L}^{n+1}$ is connected as all equipotential photon surfaces and electrostatic black hole horizons with closed cross sections in Reissner--Nordstr\"om spacetimes are necessarily spherically symmetric (see \Cref{coro:RNphoton}) and hence only such boundary component can occur. This will imply \Cref{thm:uniqueness} once we will have constructed the promised necks.

This construction is done as follows, using the notation introduced in \Cref{def:subex}: Consider a fixed non-photon sphere, non-black hole inner boundary component $\surf_{j}$ with scalar curvature radius $r_{j}$, mass $m_j$ and charge $q_j$, and induced metric $\sigma_{j}$. As $\surf_{j}$ is assumed to be subextremal, we know that $m_j>\vert q_j\vert$. We also know that $r_{j}>m_j+\sqrt{m_j^2-q_j^2}$ as $\eta(N)_{j}\mathfrak{H}_{j}>0$ via \Cref{lem:etaNH-RN}. Just as in~\cite{jahns2019photon}, we set 
\begin{align*}
I_j \definedas \left[\left(m_j+\sqrt{m_j^2-q_j^2}\right)^{\frac{1}{n-2}},r_i\right]
\end{align*}
and define a metric $\overline{g}_{j}$, a lapse function $\overline{N}_{j}$, and an electric potential $\overline{\Psi}_{j}$ on $I_{j}\times\surf_{j}$ by
 \begin{align*}
 \overline{g}_j &\definedas \frac{1}{N_{j}(r)^2}dr^2+\frac{r^2}{r_j^2}\sigma_j,\\ 
\overline{N}_{j}(r) &\definedas \sqrt{1-\frac{2m_{j}}{r^{n-2}}+\frac{q_{j}^2}{r^{2(n-2)}}},\\
\overline{\Psi}_{j}(r)&\definedas\frac{\sqrt{n-1}\,q_{j}}{\sqrt{2(n-2)}\, r^{n-2}}.
\end{align*} 

Next,  we choose
\begin{align*}
 \alpha_j &\definedas \frac{N_j}{\overline{N}_{j}(r_j)} >0,\\
 \beta_j &\definedas \Psi_j -\alpha_j \overline{\Psi}_{j}(r_j),
\end{align*}
where $N_{j}$ and $\Psi_{j}$ denote the manifestly constant values of $N$ and $\Psi$ on $\surf_{j}$. Now we can combine $(\slice, g, N, \Psi)$ with $(I_j\times \surf_{j}, \overline{g}_j, \alpha_j \overline{N}_{j}, \alpha_j\overline{\Psi}_{j}+\beta_j)$ to a new system $(\widetilde{M}^{n}, \widetilde g, \widetilde N,\widetilde \Psi)$ by gluing along all the boundary components $\surf_{j}$ (which are neither time-slices of photon spheres nor electrostatic black hole horizons)\footnote{In fact, if $\surf_{j}$ happened to be a time-slice of a photon sphere, the above construction would work just fine and reproduce what is done in \cite{jahns2019photon}.}.  

We will proceed to show that all these quantities are $C^{1,1}$ across the gluing surface $\surf_{j}$. Since the proof is similar to the ones in~\cite{Cedrgal2,jahns2019photon} and others, we will be brief: By the choice of $\alpha_j$ and $\beta_j$, both $\widetilde N$ and $\widetilde \Psi$ are continuous across the gluing surface $\surf_{j}$, and by choice of $m_j$ and $q_j$, the normal derivatives of $\widetilde N$ and $\widetilde \Psi$ coincide from both sides. Hence, $\widetilde N$ can be used as a smooth collar function to show that the gluing construction gives a smooth manifold. 

It is easy to see by the form of the metric $\widetilde{g}_j$ that the mean curvature of $\surf_{j}$ in the glued in ``neck manifold'' $(I_j\times \surf_{j}, \widetilde{g}_{j})$ is the same as the mean curvature of the sphere of coordinate radius $r_j$ in a Reissner--Nordstr\"om manifold of mass $m_j$ and charge $q_j$. Since the coordinate sphere of radius $r_j$ in the Reissner--Nordstr\"om manifold also fulfills Equation~\eqref{foundtheconstant} and has the same intrinsic scalar curvature as $\surf_{j}$, the mean curvature of $\surf_{j}$ must coincide from both sides (considering also the positivity of the mean curvature asserted in \Cref{prop: H>0}). Following the steps of the argument in~\cite{Cedrgal2} (which were repeated in~\cite{jahns2019photon}), 
one can now float local coordinates on $\surf_{j}$ along the level sets of $\widetilde N$ to show that the metric $\widetilde g$ is $C^{1,1}$ regular across the gluing surface $\surf_{j}$, exploiting the fact that the mean curvature of $\surf_{j}$ coincides from both sides as well as the umbilicity of $\surf_{j}$. Details can be found in~\cite{Cedrgal2,jahns2019photon}.  

By \cite[Lemma~25]{jahns2019photon}, the system $(I_j\times \surf_{j}, \overline{g}_j, \overline{N}_{j}, \overline{\Psi}_{j})$ fulfills the first two electrostatic equations~\eqref{EEVE4} and~\eqref{EEVE2} and the trace of the third electrostatic equation~\eqref{EEVE1}, and its boundary $\{\left(m_j+\sqrt{m_j^2-q_j^2}\right)^{\frac{1}{n-2}}\}\times \surf_{j}$ is a non-degenerate electrostatic horizon. This shows that we have constructed the desired necks.
\end{proofof}

\begin{Cor}[Alternative assumption]
Under the assumptions of \Cref{thm:uniqueness}, but asking instead that the equipotential photon surface inner boundary components are outward directed and in addition subextremal if $\eta(N)\mathfrak{H}>0$, the same conclusion can be drawn.
\end{Cor}
\begin{proof}
This follows by using \Cref{cor:outwardpositiveelectro} instead of \Cref{prop: H>0}.
\end{proof}

\begin{Rem}[Glue-in necks versus generalized Reissner--Nordstr\"om systems]\label{rem:versus}
While generalized Reissner--Nordstr\"om spacetimes/systems satisfy \eqref{EEVE4}--\eqref{EEVE1} by \Cref{prop:genRN}, the necks constructed in the proof of \Cref{thm:uniqueness} can a priori be more general, as they put no restriction on $\operatorname{R}_{*}$ (or $\kappa$). Of course, a posteriori, we know that $\kappa=1$ from \Cref{thm:uniqueness}.
\end{Rem}

\begin{Rem}[Gluing regularity, alternative proof]\label{rem: gluing regularity}
 The fact that it suffices to glue the pieces in the proof of \Cref{thm:uniqueness} with $C^{1,1}$-regularity easily follows from the proof of \Cref{thm:uniquess-old}, see~\cite{jahns2019photon}. In fact, the proof of \Cref{thm:uniquess-old} relies on constructing a new manifold by gluing new pieces to the original manifold with $C^{1,1}$-regularity and an application of a low regularity version of the rigidity case of the positive mass theorem, see \Cref{rem:PMT}.
 
 Furthermore, \Cref{prop:geometrynearslices} asserts that the gluing performed in the proof of \Cref{thm:uniquess-old} is in fact \emph{smooth} whenever the equipotential photon surface time-slice $\surf_{j}$ is generic. In other words, in the generic case, the gluing does nothing else than adding the piece of the respective rescaled Reissner--Nordstr\"om manifold arising in a neighborhood of $\surf_{j}$ all the way down to the horizon. The proof of \Cref{thm:uniquess-old} hence only adds new insights for exceptional times (which we could of course argue away by \Cref{prop:exceptional}). Differently put, in view of our refined analysis in \Cref{sec:prelimunique,sec: H>0}, the only case where glue-in necks really need to be constructed is the photon sphere boundary case which has already been dealt with by Jahns in \cite{jahns2019photon}. On the other hand, we believe it to be instructive to show that the glue-in procedure nevertheless works also for time-slices of equipotential photon surfaces, keeping in mind that our proof of positivity of $H_{j}$ needed the more refined analysis near generic time-slices of equipotential photon surfaces.
  \end{Rem}

\appendix
\section{Appendix}\label{sec:appendix}
\begin{prop}\label{prop: appendix}
Let $\mathcal{I}=[0, \infty)$ and let $F\colon \mathcal{I} \to \mathbb{R}^{+}_{0}$ be a bounded $C^3$-function with $F(0)=0$ and $F(r)>0$ for $r>0$. Assume moreover that $F'$ is bounded on $\mathcal{I}$ and that $F'(0)>0$. Then the initial value problem
\begin{align}\label{eq:IVP alpha}
\dot{r}(t)&= \left\{
\begin{array}{ll}
\sqrt{F(r(t))} & \text{for }r(t) \in (0,\infty), \\
0 & \, \text{otherwise,} \\
\end{array}
\right.\\\label{eq:inir}
r(0)&=0
\end{align}
has a unique non-trivial global $C^{1}$-solution $r\colon \mathbb{R}\to [0, \infty)$. Moreover, this solution $r$ satisfies $r(t)=0$ for $t\leq0$.
\end{prop}

\begin{proof}
Let us first only consider the range $r>0$ and multiply both sides of \eqref{eq:IVP alpha} by $r^{-\frac{1}{2}}$, obtaining
\begin{align}\label{eq:DG r^-1+alpha}
r^{-\frac{1}{2}}\,\dot{r}=\left(F(r)\right)^{\frac{1}{2}} r^{-\frac{1}{2}}.
\end{align}
Next we substitute $u\defeq r^\frac{1}{2}$ in~\eqref{eq:DG r^-1+alpha}, which yields the differential equation
\begin{align}\label{eq: DG u}
\dot{u}=\frac{\left(F\!\left(u^2\right)\right)^\frac{1}{2}}{2u}.
\end{align}
We would like to show that the right hand side of \eqref{eq: DG u} can be continuously extended by $\frac{\left(F'(0)\right)^{\frac{1}{2}}}{2}$ as $u\searrow0$ and that the so-extended right hand side is indeed Lipschitz continuous on $\mathcal{I}$. To see this, we observe that as $u\searrow0$, we have
\begin{align*}
\frac{\left(F\!\left(u^2\right)\right)^\frac{1}{2}}{2u}&=\frac{\left(F(0)+F'(0)u^{2}+o\left(u^{2}\right)\right)^\frac{1}{2}}{2u}=\frac{\left(F'(0)+o(1)\right)^{\frac{1}{2}}}{2}
\end{align*}
by Taylor's theorem and $F(0)=0$, so that the right hand side of \eqref{eq: DG u} indeed converges to $\frac{\left(F'(0)\right)^{\frac{1}{2}}}{2}$ as $u\searrow0$. To see that the right hand side of \eqref{eq: DG u} (extended by $\frac{\left(F'(0)\right)^{\frac{1}{2}}}{2}$ in $u=0$) is Lipschitz continuous on $\mathcal{I}$, we compute its derivative 
\begin{align}\label{eq: u derivative}
\frac{\mathrm{d}}{\mathrm{d}u}\frac{\left(F\!\left(u^2\right)\right)^\frac{1}{2}}{2u}&=\frac{F'\!\left(u^2\right)u^2-F\!\left(u^2\right)}{2F\!\left(u^2\right)^{\frac{1}{2}}u^2}
\end{align}
and note that both the numerator and the denominator of the right hand side of \eqref{eq: u derivative} degenerate to $0$ as $u\searrow0$. For the sake of applying l'H\^{o}pital's rule, we take the quotient of the derivatives with respect to $u$ of the numerator and denominator of the right hand side of \eqref{eq: u derivative}, resulting in the term
\begin{align}\label{eq:hopital 1}
\frac{F''\!\left(u^2\right)u^2}{2F\!\left(u^2\right)+F'\!\left(u^2\right)u^2}\times F\!\left(u^2\right)^\frac{1}{2},
\end{align}
where again both the numerator and the denominator of the first factor degenerate to $0$ as $u\searrow0$. Repeating the process with said first factor, we obtain
\begin{align*}
\frac{F'''\!\left(u^{2}\right)u^{2}+F''\!\left(u^{2}\right)}{3F'\!\left(u^{2}\right)+F''\!\left(u^{2}\right)u^{2}}
\end{align*}
which tends to $\frac{F''(0)}{3F'(0)}$ as $u\searrow0$. Hence, by l'H\^{o}pital's rule and recalling that $F'(0)>0$, we have 
\begin{align*}
\lim_{u\searrow0}\frac{\mathrm{d}}{\mathrm{d}u}\frac{\left(F\!\left(u^2\right)\right)^\frac{1}{2}}{2u}&=0.
\end{align*}
Together with the boundedness assumptions on $F$ and $F'$, this implies Lipschitz continuity of the right hand side of~\eqref{eq: DG u} in $u$ uniformly in $t$. In other words, the symmetrized initial value problem
\begin{align}\label{eq:u1}
\dot{u}(t)&=
\left\{
\begin{array}{ll}
\frac{\sqrt{F\left(\left(u(t)\right)^2\right)}}{2\vert u(t)\vert} & \text{for }u(t) \neq0, \\[1ex]
\frac{\sqrt{F'(0)}}{2} & \, \text{otherwise,} \\
\end{array}
\right.\\\label{eq:u2}
u(0)&=0
\end{align}
has a unique global $C^{1}$-solution $u\colon\R\to\R$ by the global Picard--Lindel\"{o}f theorem (see e.g. \cite[p.\,55]{GDG}). Moreover, as $\dot{u}(0)>0$ by \eqref{eq:u1} and $F'(0)>0$, we can conclude that $u$ is non-trivial. Note that $u$ must be odd in $t$, $u(-t)=-u(t)$ for all $t\in\R$, as otherwise $\widetilde{u}(t)\definedas -u(-t)$ would yield another solution to \eqref{eq:u1}, \eqref{eq:u2}.

Next, setting $r(t)\definedas \left(u(t)\right)^{2}$ for $t\geq0$ and $r(t)\definedas0$ for $t\leq0$ for the unique solution $u$ of \eqref{eq:u1}, \eqref{eq:u2} yields a non-trivial global $C^{1}$-solution $r\colon\R\to\mathcal{I}$ of \eqref{eq:IVP alpha}, \eqref{eq:inir} as $\dot{r}(t)=2\dot{u}(t)u(t)\to0$ for $t\searrow0$ by construction. To see that this function $r$ is indeed the unique non-trivial global $C^{1}$-solution to \eqref{eq:IVP alpha}, \eqref{eq:inir}, we argue as follows: Given any solution $r$ of \eqref{eq:IVP alpha}, \eqref{eq:inir}, the fact that $\dot{r}\geq0$ implies that $r(t)=0$ for $t\leq0$. Setting 
\begin{align*}
u(t)&\definedas \left\{
\begin{array}{ll}
\sqrt{r(t)}& \text{for }t \geq0, \\
-\sqrt{r(t)} & \, \text{otherwise,}
\end{array}
\right.
\end{align*}
gives an odd function with
\begin{align*}
\dot{u}(t)&=\frac{\dot{r}(t)}{2\sqrt{r(t)}}=\frac{\sqrt{F(r(t))}}{2\sqrt{r(t)}}=\frac{\sqrt{\left(F(u(t)\right)^{2}}}{2\vert u(t)\vert}
\end{align*}
for $t>0$. Hence $u$ solves the initial value problem \eqref{eq:u1}, \eqref{eq:u2}.
\end{proof}
\begin{Cor}\label{cor: appendix}
Let $\mathcal{I}$ and $F$ be in as in \Cref{prop: appendix}. Then the initial value problem
\begin{align}\label{eq:IVP alpha-}
\dot{s}(t)&= \left\{
\begin{array}{ll}
-\sqrt{F(s(t))} & \text{for }s(t) \in (0,\infty), \\
0 & \, \text{otherwise,} \\
\end{array}
\right.\\\label{eq:inir-}
s(0)&=0
\end{align}
has a unique non-trivial global $C^{1}$-solution $s\colon \mathbb{R}\to [0,\infty)$. Moreover, this solution $s$ satisfies $s(t)=0$ for $t\geq0$.
\end{Cor}
\begin{proof}
Any $C^{1}$-solution $s\colon \mathbb{R}\to [0,\infty)$ of \eqref{eq:IVP alpha-}, \eqref{eq:inir-} induces a solution $r\colon \mathbb{R}\to [0,\infty)$ of \eqref{eq:IVP alpha}, \eqref{eq:inir} via $r(t)\definedas s(-t)$ for all $t\in\R$ and vice versa. Hence, \Cref{prop: appendix} asserts the claim.
\end{proof}

\begin{Ex}
Choosing $F(r)=r$ for all $r\in\mathcal{I}$, the IVP \eqref{eq:IVP alpha}, \eqref{eq:inir} has the unique non-trivial solution $r\colon\R\to[0,\infty)$ given by
\begin{align*}
r(t)&=
\left\{
\begin{array}{ll}
\frac{t^{2}}{4}& \text{for }t \geq0, \\
0 & \, \text{otherwise,} \\
\end{array}
\right.
\end{align*} 
which is $C^{1}$ but not $C^{2}$ at $t=0$. This solution corresponds to $u(t)=\frac{t}{2}$ for $t\in\R$.
\end{Ex}

\begin{Rem}\label{rem:regularityappendix}
As can be seen from its proof, the assertion of \Cref{prop: appendix} continues to hold if $F$ is $C^{3}$ with bounded first derivative only on $[0,R_{*})$ for some threshold $R_{*}>0$ provided that $F$ is Lipschitz continuous on $\mathcal{I}=[0,\infty)$. Moreover, it suffices that $F(r)>0$ for $r\in[0,R_{*})$ and $F(r)\geq0$ for $r\geq R_{*}$. Indeed, the analysis of the right hand side of the ODE~\eqref{eq: DG u} is unaffected by these changes; in particular, said right hand side continues to be Lipschitz continuous which allows for the application of the global Picard--Lindel\"of theorem. \end{Rem}

\begin{prop}\label{prop: appendix P-L lapse}
	The IVP 
	\begin{align*}
		0=&\bigg(u''(r)u(r)r-\frac{r}{2}u'(r)^2+\frac{n}{2}u'(r)u(r)\bigg)\left(C^2+u(r)-1\right)-\frac{r}{2}u'(r)^2u(r)\\
		&+\frac{\kappa^2\mathrm{R}_{\sigma_i}}{8(n-1)}r^{2(n-1)}u'(r)^3-\frac{r^2}{8(n-2)}u'(r)^3,\\
		u(r_0)=&u_0,\\
		u'(r_0)=&u'_0
	\end{align*}
from Proposition \ref{prop:normalquasiloc} has a unique solution on some open interval $(r_0-\varepsilon, r_0+\varepsilon).$
\end{prop}
\begin{proof}
	We substitute $w(r):=\frac{u'(r)}{u(r)}$ in order to obtain the ODE-system of two first order equations
	\begin{align*}
		w'=&-\frac{1}{\left(C^2+u(r)-1\right)r}\bigg(\frac{n}{2}\left(C^2+u(r)-1\right)w+\left(\left(C^2+u(r)-1\right)-u\right)\frac{r}{2}w^2\\
		&+\left[\frac{\kappa^2\mathrm{R}_{\sigma_i}}{(n-1)}r^{2(n-1)}-\frac{r^2}{(n-2)}\right]\frac{u}{8}w^3\bigg),\\
		u'=&wu,
	\end{align*}
where we omit the dependence on $r$ for simplicity.
The second ODE can be solved explicitly by $u(r)=u_0\,\mathrm{e}^{\int_{r_0}^{r}w(s)\mathrm{d}s},$ 
while we can apply the local version of Picard--Lindelöff  with initial value condition $w(r_0)=\frac{u'_0}{u_0}$ to the first, as the denominator $r\left(C^2+u(r)-1\right)$ stays away from $0.$ 
\end{proof}

\bibliographystyle{amsplain}
\bibliography{photon-surfaces}

\end{document}